\documentclass[journal,onecolumn]{IEEEtran}%
\usepackage{eurosym}
\usepackage{amsmath}
\usepackage{amsthm}
\usepackage{amsfonts}
\usepackage[colorlinks=true, linkcolor=blue, citecolor=red]{hyperref}
\usepackage{amssymb}
\usepackage{graphicx}
\usepackage{tikz}
\usepackage{bm}
\usepackage{epstopdf}
\usepackage{booktabs}%
\usepackage{soul}
\providecommand{\U}[1]{\protect\rule{.1in}{.1in}}

\theoremstyle{plain}
\newtheorem{theorem}{Theorem}
\newtheorem{corollary}{Corollary}
\newtheorem{lemma}{Lemma}

\theoremstyle{definition}

\newtheorem{remark}{Remark}

\title{On Multiuser Gain and the Constant-Gap Sum Capacity of the Gaussian Interfering Multiple Access Channel}

\author{
\IEEEauthorblockN{Rick Fritschek and Gerhard Wunder\\}
\IEEEauthorblockA{Heisenberg Communications and Information Theory Group\\
    Freie Universit\"at Berlin, \\
    Takustr. 9,
    D--14195 Berlin, Germany\\
    Email: rick.fritschek@fu-berlin.de, g.wunder@fu-berlin.de
    \thanks{This paper was presented in part at the ICC 2014 \cite{Fritschek2014a}, ITW 2014\cite{Fritschek2014b}, ICC 2015\cite{Fritschek2015a} and ISIT 2015\cite{Fritschek2015b}}}
}

\begin{document}
\maketitle
\begin{abstract}
Recent investigations have shown sum capacity results within a constant bit-gap for several channel models, e.g. the two-user Gaussian interference channel (G-IC), k-user G-IC or the Gaussian X-channel. This has motivated investigations of interference-limited multi-user channels, for example, the Gaussian interfering multiple access channel (G-IMAC). Networks with interference usually require the use of interference alignment (IA) as a technique to achieve the upper bounds of a network. A promising approach in view of constant-gap capacity results is a special form of IA called signal-scale alignment, which works for time-invariant, frequency-flat, single-antenna networks. However, until now, results were limited to the many-to-one interference channel and the Gaussian X-channel. To make progress on this front, we investigate signal-scale IA schemes for the G-IMAC and aim to show a constant-gap capacity result for the G-IMAC. We derive a constant-gap sum capacity approximation for the lower triangular deterministic (LTD)-IMAC and see that the LTD model can overcome difficulties of the linear deterministic model. We show that the schemes can be translated to the Gaussian IMAC and that they achieve capacity within a constant gap. We show that multi-user gain is possible in the whole regime and provide a new look at cellular interference channels.
\end{abstract}

\date{June 13, 2014}

\section{Introduction}
One of the main limiting factors of cellular networks is interference. Throughout the last decades, research was conducted to investigate the role of interference in information theory. Several channel models were proposed, in which interference is one of the main limiting factors. But even for the simplest one, the two-user inference channel (IC), the capacity characterisation is an unsolved problem for more than 40 years. However, Etkin, Tse and Wang \cite{etkin2008} have achieved a major breakthrough, a capacity result to within one bit for the Gaussian IC. They used a fundamental principle: if something is too hard to solve as one, you need to divide it in smaller simpler parts. And in the course of their investigation, defined the concept of {\it generalized degrees of freedom }(GDoF)\begin{equation*}
d_{\text{sym}}(\alpha):=\lim_{\text{SNR,INR} \rightarrow \infty;\tfrac{\log \text{INR}}{\log \text{SNR}}=\alpha} \tfrac{C_{\text{sym}}(\text{INR,SNR})}{C_{\text{awgn}}(\text{SNR})},
\end{equation*}
which can be viewed as an interference dependent notion of degrees of freedom (DoF). This achievement motivated a series of investigations of more complex channel models (e.g. \cite{Bresler2010}, \cite{Niesen-Ali}), towards constant-gap capacity approximations. However, even constant-gap capacity results were hard to obtain, due to noise properties of the Gaussian channel models. A branch of research therefore investigated deterministic models. In particular, the so-called linear deterministic model (LDM) emerged with the work of Avestimehr, Diggavi and Tse \cite{Avestimehr2007}. The LDM approximates the channel by the binary expansion of the real signals. The coefficients of this binary expansion are viewed as bit-vectors and positions within these vectors are called levels. These bit-vectors are truncated at noise level, such that the noise corrupted bits are removed from the channel model, resulting in a deterministic approximation. Moreover, channel gain is approximated by a power of two $2^n$, $n=\left\lceil \log \text{SNR}\right\rceil$, which results in a downshift of the bit-vector. Superposition is modelled as level-wise binary addition. Surprisingly, this rather simple model results in capacity approximations to within a constant bit-gap of the corresponding Gaussian channel models, e.g. a 42-bit gap for the deterministic two-user IC \cite{Bresler2008}. Another research branch investigated the concept of interference alignment, introduced in \cite{4567589}, \cite{4418479} and \cite{4567443}. Interference alignment methods align interfering signal parts in some dimension and therefore make more interference-free dimensions available. These methods can be broadly categorized into two classes: vector-space alignment and Signal-space alignment \cite{Niesen-Ali}. In vector-space alignment methods, the dimensions of multiple antennas, time and frequency are used to align interfering signal parts into some sub-spaces. However, in single-antenna, time-invariant, frequency-flat channel models, these methods fail and the class of signal-scale alignment methods needs to be used. In these methods, techniques such as lattice coding, split the channel in several power layers, allowing for alignment of interference. It was shown in recent investigations, that LDM solutions are the basis for the corresponding signal-scale alignment methods and therefore provide a stepping-stone for constant bit-gap capacity results, see for example \cite{Sridharan08,Suvarup2011,Chaaban16}. An interesting application of interference alignment is its use for the Gaussian interfering multiple access channel (G-IMAC) in \cite{Suh2008}. In this investigation, the G-IMAC serves as a general simple model for cellular networks. It was shown that multiuser gain, in form of additional DoF, can be enabled by vector-space alignment using frequency and delay properties of the channel. The question is now, if multiuser gain is still present in the single-antenna, time-invariant, frequency-flat cellular networks, especially the G-IMAC as simplest model. Notable contributions towards an understanding of the G-IMAC were done in \cite{chaaban2011}, with a result for the strong interference regime.

\subsection{Contributions}
We investigate the single-antenna, time-invariant, frequency-flat Gaussian IMAC (G-IMAC). To make progress on this front, we start by investigating the linear deterministic approximation of the G-IMAC named LD-IMAC. We show that basic achievablity schemes from the linear deterministic MAC-P2P \cite{Buhler2012} can be extended towards the LD-IMAC. In those schemes, the orthogonality of bit-levels in the LDM bit-vectors is used to exploit the signal-scale shift between two cells. This results in the alignment of the interference in half of its bit-levels, effectively reducing the interference by half in the weak interference regime. Due to a coupling of both cells, the achievable schemes are limited to the weak interference regime. Moreover, due to dependence on the ratio of interference-to-direct signal strength $\alpha$, the achievable sum-rate has a step-like curve. However, converse proofs cannot assume orthogonality of bit-levels, which would mean a uniform distribution of the real signals and therefore results in a loose upper bound for certain values of $\alpha$. This yields a constant-gap sum capacity for just certain discrete points depending on $\alpha$. Using signal-scale alignment methods, i.e. layered lattice codes, we transfer the achievable scheme and the result to the Gaussian IMAC. Moreover, one can show that the LD-IMAC bounds are actually within a constant bit-gap of the Gaussian IMAC bounds. This yields a constant bit-gap sum capacity result for the G-IMAC at discrete points, again depending on the channel gains. This shows that the schemes stemming from the LD-IMAC and the LDM itself, is a sub-optimal approximation, which is also sub-optimal as an approximation of the sum-rate of the G-IMAC. As we later show, this sub-optimality stems from the LD models property of just allowing orthogonal bit-level assignment schemes. Interestingly, a new deterministic model was introduced in \cite{Niesen-Ali}, coined the lower triangular deterministic model (LTDM). In this model, the channel gain is not approximated by $2^n$ any more, but the remainder incorporated as binary expansion. This results in a discrete convolution between the bits of the channel gain and the bits of the real signal, and therefore yields a dependence between the bit-levels. This means that the model allows a broader form of achievable schemes, were the bit-levels do not need to be orthogonal. Since the problems with the LDM stem exactly from this necessity, an LTDM scheme could improve upon prior results. We show LTD-IMAC schemes for the whole interference range, which (completely) reach the LDM upper bounds in the weak interference case and hence improve upon prior results. Moreover, we transform the bounds towards the LTDM channel and develop new upper bounds for the remaining interference ranges. We therefore show the deterministic approximation of the sum capacity of the LTD-IMAC to within a constant bit-gap. Extending the proof methods of \cite{Niesen-Ali} towards the structure of our achievable scheme, we can use them to transfer the constant-gap results from the LTD-IMAC to the G-IMAC. This yields a constant-gap sum capacity approximation of the symmetric G-IMAC in the whole interference range, see Figure~\ref{GDOF_IMAC}.
\begin{figure}
\centering
\includegraphics[scale=0.4]{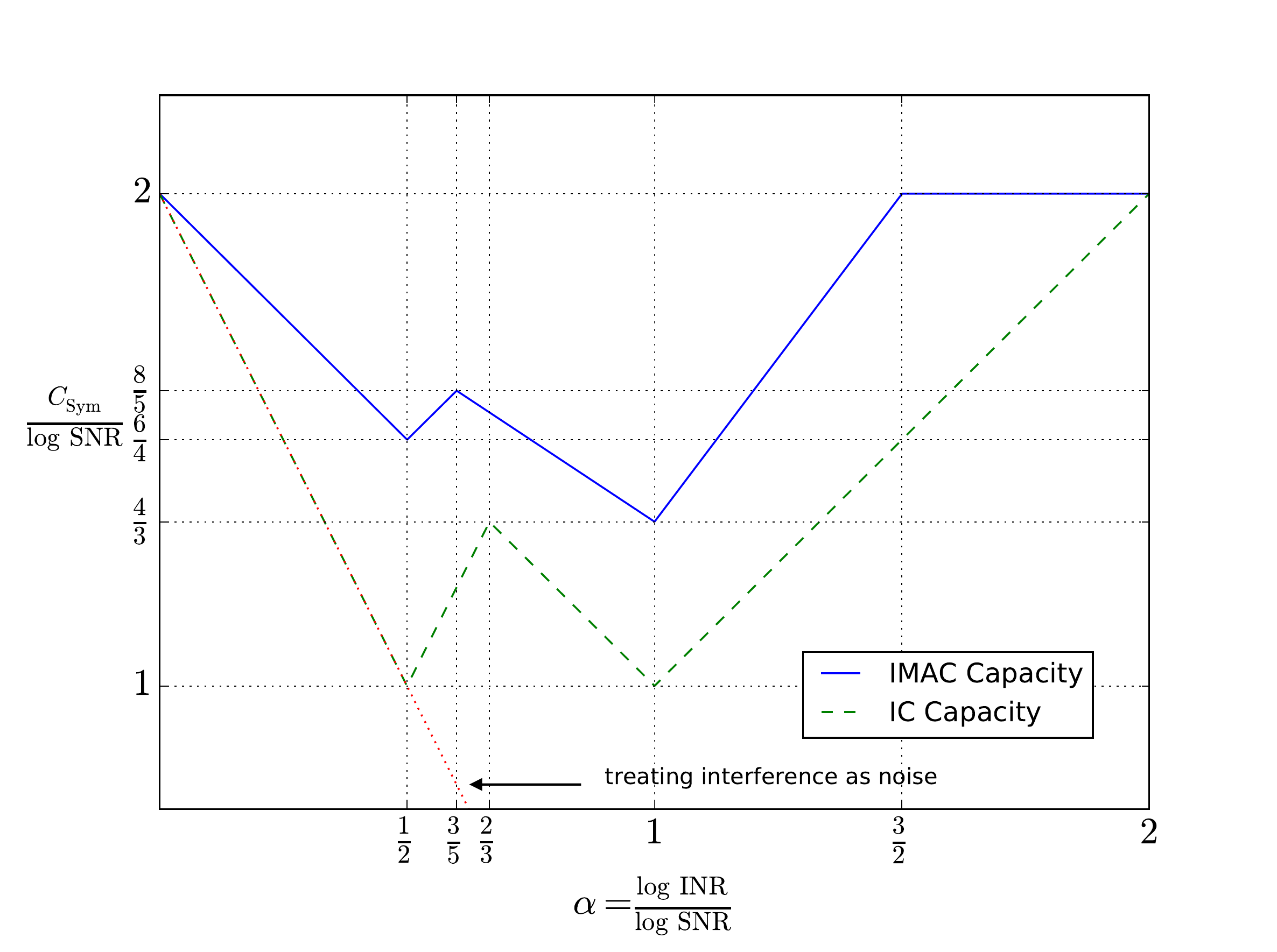}
\caption{The GDoF ``W'' curve for the symmetric Gaussian IC and the symmetric Gaussian IMAC. In the case of the IMAC, the SNR is the signal-noise-ratio of the strong direct links (associated with the coarse channel gain $2^{n_1}$). Moreover, the curve shows the case that the weaker direct links (associated with coarse channel gain $2^{n_2}$) are {\it strong enough} to support full multi-user gain.}
\label{GDOF_IMAC}
\end{figure}

\subsection{Organization}
We introduce the G-IMAC, the LD-IMAC and the LTD-IMAC in section \ref{System Model}. We then present an overview of the main results and theorems of this paper in section \ref{Main Results}. Analysis of the LD-IMAC and proofs for the theorems are presented in section \ref{Proof LD-Theorems}. In section \ref{Proof for LTD-IMAC} we present achievable schemes and upper bounds for the LTD-IMAC and in section \ref{Proof G-IMAC} we show that the LTD results can be transferred to the G-IMAC.

\section{System Model}
\label{System Model}
\subsection{Notation}

We denote vectors and matrices by lower case bold and upper case bold characters, respectively. For two vectors $\mathbf{a}$ and $\mathbf{b}$, we denote by $[\mathbf{a};\mathbf{b}]$ the vector that is obtained by stacking $\mathbf{a}$ over $\mathbf{b}$. 
To specify a particular range of elements in a bit-level vector we use the notation $\mathbf{a}_{[i:j]}$ to indicate that $\mathbf{a}$ is restricted to the bit-levels $i$ to $j$. If $i=1$, it will be omitted $\mathbf{a}_{[:j]}$, the same for $j\!=\!n$ $\mathbf{a}_{[i:]}$. Moreover, we use the notation of $a_{ik}^j$, where $j$ represents the receiver cell, and $i$ and $k$ the transmitter cell and user respectively.

\subsection{The Gaussian Interfering Multiple Access Channel}

\begin{figure}
\centering
\includegraphics[width=0.2 \textwidth]{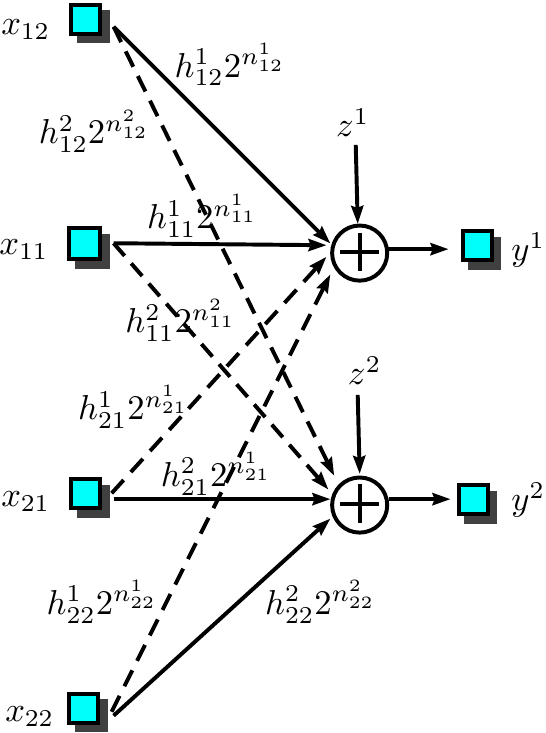}
	\caption{Illustration of the Gaussian IMAC}
    \label{system model}
\end{figure}

We consider the Gaussian interfering multiple access channel (IMAC), in which there are two Gaussian multiple access channels (MACs) interfering among themselves.
Therefore, the system consists of 4 transmitters and 2 receivers. Transmitters $x_{11}$ and $x_{12}$ together with the receiver $y^1$ and $x_{21}$, $x_{22}$ with $y^2$ each form a MAC and both are interfering with each other (see Fig. \ref{system model}).

The channel equations for a fixed time slot are given by 
\begin{IEEEeqnarray}{rCl}
y^1 &=& h_{11}^1 2^{n^1_{11}}x_{11}+h_{12}^1 2^{n_{12}^1} x_{12} + h_{21}^1 2^{n_{21}^1} x_{21}\IEEEyessubnumber\\
&&+\: h_{22}^1 2^{n_{22}^1} x_{22} +z^1\IEEEnonumber\\
y^2 &=& h_{21}^2 2^{n_{21}^2} x_{21}+h_{22}^2 2^{n_{22}^2}x_{22}+h_{11}^2  2^{n_{11}^2} x_{11}\IEEEyessubnumber\\
&&+\:h_{12}^2 2^{n_{12}^2} x_{12}+ z^2\IEEEnonumber,
\end{IEEEeqnarray}

where $z^j\in \mathcal{N}(0,1)$ is assumed to be zero mean and unit variance Gaussian noise. Also each transmitted signal has an associated unit average power constraint $E\{|x_{ik}|^2\}\leq 1$. The channel gains are composed of two parts. A coarse channel gain, $2^n$ with $n\in \mathbf{N}$ and a fine channel gain $h \in (1,2]$. However, both parts can model any real channel gain greater than 1, which is sufficient for a constant-gap analysis. We choose this channel notation following the notation of \cite{Niesen-Ali} to get a clear integration of linear deterministic and lower triangular deterministic ideas in the following investigation of the model. 
It is assumed that $n^1_{21}=n^1_{22}=:n_2^1$, $n^2_{11}=n^2_{12}=:n_1^2$, stating that the total interference strength caused by $\mathbf{x}_{ij}$ at the receivers is the same\footnote{This assumption is not necessary for the techniques to work. In fact, we just need a difference in the ratios of both direct links and both interference links, which we call shift-property. However, it simplifies the investigation, since the shift-property gets simpler.}. Note that the restriction for the coarse gain is justified in the case when the distance between the two cells is significantly larger than the cell dimension itself. However, for the fine channel gains of the interfering signals, we consider a simple modulation scheme. We take a similar approach as in \cite{Niesen-Ali}, to form the channel input such that the fine channel gains of the interfering signals coincide.

Each transmitter has one message to communicate to his receiver. As in the X-channel, we have 4 independent messages $w_{ik}$.
Assume that each message $w_{ik}$ is modulated into the signal $u_{ik}$.
The transmitters can now form the channel input such that
\begin{IEEEeqnarray}{rCl}
\label{modulation}
x_{11}&:=&h_{12}^2 u_{11}\IEEEyessubnumber\\
x_{12}&:=&h_{11}^2 u_{12}\IEEEyessubnumber\\
x_{21}&:=&h_{22}^1 u_{21}\IEEEyessubnumber\\
x_{22}&:=&h_{21}^1 u_{22}\IEEEyessubnumber.
\end{IEEEeqnarray}

We therefore see at the receiver side the following signals
\begin{IEEEeqnarray*}{rCl}
\label{Gauss_Model_h}
y^1 &=& h_{11}^1h_{12}^2 2^{n^1_{11}}u_{11}+h_{12}^1h_{11}^2 2^{n_{12}^1} u_{12}\IEEEyessubnumber\\
&&\:+ h_{21}^1h_{22}^1 2^{n_2^1} (u_{21}+  u_{22}) +z^1\\
y^2 &=& h_{21}^2h_{22}^1 2^{n_{21}^2} u_{21}+h_{22}^2h_{22}^1 2^{n_{22}^2}u_{22}\IEEEyessubnumber\\
&&\:+ h_{11}^2h_{12}^2  2^{n_1^2}( u_{11}+u_{12})+ z^2.
\end{IEEEeqnarray*}
Notice that we have modulated the signals in a way, such that the interference parts align at the unintended receiver, reproducing a similar structure as in the X-channel\footnote{Note that the difference lies in the coarse channel gains, in particular those of the aligning parts. The only case, where both structures are equal, is when {\it all} coarse channel gains are equal. In that case, both channel models reach $\tfrac{4}{3}$ DoF.}.
Now we can define
\begin{IEEEeqnarray*}{ccc"ccc}
g_{11}^1&:=&h_{11}^1h_{12}^2, & g_{21}^2&:=&h_{21}^2h_{22}^1\\
g_{12}^1&:=&h_{12}^1h_{11}^2, & g_{22}^2&:=&h_{22}^2h_{21}^1\\
g_{2}^1&:=&h_{21}^1h_{22}^1, & g_{1}^2&:=&h_{11}^2h_{12}^2.
\end{IEEEeqnarray*}
Where $g_{ik}^j\in(1,4]$ and we can rewrite (\ref{Gauss_Model_h}) as 
\begin{IEEEeqnarray*}{rCl}
\label{Gauss_Model_g}
y^1 &=& g_{11}^1 2^{n^1_{11}}u_{11}+g_{12}^1 2^{n_{12}^1} u_{12}\IEEEyessubnumber\\
&&+\: g_{2}^1 2^{n_2^1} (u_{21}+  u_{22}) +z^1\\
y^2 &=& g_{21}^2 2^{n_{21}^2} u_{21}+g_{22}^2 2^{n_{22}^2}u_{22}\IEEEyessubnumber\\
&&+\: g_{1}^2  2^{n_1^2}( u_{11}+u_{12})+ z^2.
\end{IEEEeqnarray*}


Moreover, we assume without loss of generality that $n^1_{11}\geq n^1_{12}$,  $n_{21}^2\geq n_{22}^2$ and the difference between the two coarse channel gains is denoted as $n^1_{11}-n^1_{12}=:\Delta_1$ and $n_{21}^2-n_{22}^2=:\Delta_2$.

\subsection{Linear Deterministic IMAC}
The linear deterministic model (LDM) models the input symbols at $x_{ik}$ as bit vectors $\mathbf{x}_{ik}$. This is achieved by a binary expansion of the real input signal. The resulting bits constitute the new bit vector. The positions within the vector will be referred to as levels. To model the signal impairment induced by noise, the bit vectors will be truncated at noise level and only the n most significant bits are received at $y^j$. This is done by shifting the incoming bit vector for $q-n$ positions $\mathbf{y}=\mathbf{S}^{q-n}\mathbf{x}$, with $q:=\max{n^j_{ik}}$.
Where the shift matrix $\mathbf{S}$ is defined as 
\begin{equation}
\mathbf{S}=\begin{pmatrix}
0 & 0 &  \cdots & 0 & 0\\
1 & 0 &  \cdots & 0 & 0\\
0 & 1 &  \cdots & 0 & 0\\
\vdots & \vdots & \ddots & \vdots & \vdots \\
0 & 0 &  \cdots & 1 & 0\\
\end{pmatrix}.
\end{equation}
By scaling of the channel gains, we may assume an unit average power constraint in the corresponding Gaussian channel. And therefore a peak power constraint on the channel inputs in the linear deterministic model \cite{Avestimehr2007}.
Superposition at the receivers is modelled via binary addition of the incoming bit vectors on the individual levels. Carry over is not used to limit the superposition on the specific level where it occurs. This enables schemes where bit-levels can be used independently. This property enables a form of bit-level alignment, which is used by our schemes. The channel gain is represented by $n^j_{ik}$-bit levels which corresponds to $\lceil\log \mbox{SNR}\rceil$ of the original channel. We therefore approximate the fine channel gain by one. 
With this definitions the model can be written as
\begin{IEEEeqnarray*}{rCl}
\mathbf{y}_1 & = & \mathbf{S}^{q-n^1_{11}}\mathbf{x}_{11}\oplus \mathbf{S}^{q-n^1_{12}}\mathbf{x}_{12}\oplus \mathbf{S}^{q-n^1_{21}}\mathbf{x}_{21}\IEEEyessubnumber\\
&&\oplus\: \mathbf{S}^{q-n^1_{22}}\mathbf{x}_{22}\\
\mathbf{y}_2 & = & \mathbf{S}^{q-n^2_{11}}\mathbf{x}_{11}\oplus \mathbf{S}^{q-n^2_{12}}\mathbf{x}_{12}\oplus \mathbf{S}^{q-n^2_{21}}\mathbf{x}_{21}\IEEEyessubnumber\\
&&\oplus\: \mathbf{S}^{q-n^2_{22}}\mathbf{x}_{22} \label{MAC-MAC}
\end{IEEEeqnarray*}

\begin{figure}
\centering
\includegraphics[width=0.2 \textwidth]{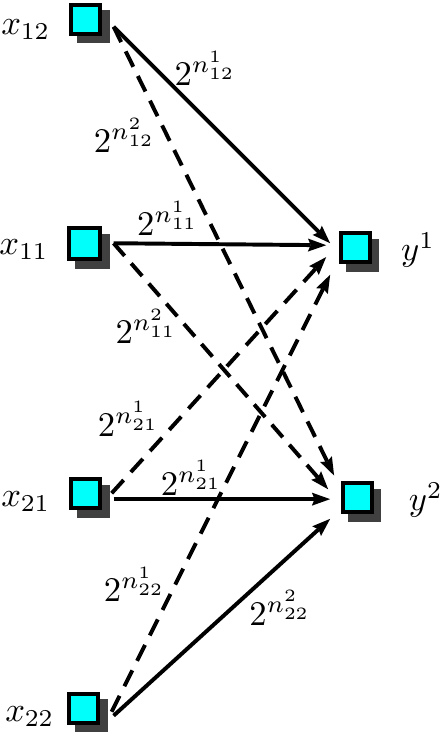}
	\caption{Illustration of the LD-IMAC}
    \label{system model}
\end{figure}

\subsection{Lower Triangular Deterministic IMAC}

%
The lower triangular deterministic model splits the channel gain in a fine channel gain $h\in [1,2)$, and a coarse channel gain $2^n$ with $ n\in \mathbb{N}$. Input symbols $x$ are modelled as bit vectors $\mathbf{x}$, which consist of the bits from the binary expansion of the real input signal. The positions within the vectors will be referred to as levels. Moreover, the fine channel gain $h$ is also written as a binary expansion resulting in a discrete convolution between the bits of $h$ and $x$, see \cite{Niesen-Ali} for a thorough exposition. To model the signal impairment induced by noise, the bits in the convolution with exponent less than zero (at noise level) will be truncated and only the $n$ most significant bits are received at $y$. The incoming bit vector can be written as $\mathbf{H}\mathbf{x}$,
where $\mathbf{H}$ is a lower triangular matrix defined as
\begin{equation}
\mathbf{H}=\begin{pmatrix}
1 & 0 &  \cdots & 0 & 0\\
[h]_1 & 1 &  \cdots & 0 & 0\\
[h]_2 & [h]_1 &  \cdots & 0 & 0\\
\vdots & \vdots & \ddots & \vdots & \vdots \\
[h]_{n-1} & [h]_{n-2} &  \cdots & [h]_1 & 1\\
\end{pmatrix},
\end{equation}
where $[h]_i$ represents the $i$-th bit in the binary expansion of $h$. Superposition at the receivers is modelled via binary addition of the incoming bit vectors. As in the LD-IMAC, signal strength differences in the multi-user model are incorporated by shifting the bit vector by $q-n_{ik}^j$ positions $\mathbf{\bar{x}_{ik}}=\mathbf{S}^{q-n_{ik}^j}\mathbf{x}_{ik}$. We introduce the notation of $\mathbf{\bar{x}}^c_{ik}=\mathbf{S}^{q-n_{ik}^j}\mathbf{x}_{ik}$ for $i \neq j$ for a better distinction between direct and interference signals. Additionally, we make some further symmetry assumptions. Since both MAC channels are coupled for $\alpha>\tfrac{1}{2}$ in a sense, that one cannot separate both achievable schemes, the number of cases would increase. In order to reduce the complexity of the investigation and presenting a clear overview, we assume that $n^1_{11}=n^2_{21}:=n_1$, $n^1_{12}=n^2_{22}:=n_2$, $n_1^2=n_2^1=n_i$. We remark that these assumptions keep the general channel structure intact, in particular the shift-property, and the proofs can be extended to the more general cases as well. Thus, the LTD-IMAC channel model is governed by the following equations
\begin{IEEEeqnarray*}{rCl}
\mathbf{y}^1 &=& \mathbf{H}_{11}^1 \mathbf{\bar{x}}_{11}\oplus\mathbf{H}_{12}^1 \mathbf{\bar{x}}_{12}  \oplus \mathbf{H}_2^1 \left(\mathbf{\bar{x}}^c_{21}\oplus \mathbf{\bar{x}}^c_{22} \right)\IEEEyessubnumber\\
\mathbf{y}^2 &=& \mathbf{H}_1^2(\mathbf{\bar{x}}^c_{11}\oplus\mathbf{\bar{x}}^c_{12})\oplus\mathbf{H}_2^2 \mathbf{\bar{x}}_{21}\oplus\mathbf{H}_{22}^2 \mathbf{\bar{x}}_{22},\IEEEyessubnumber
\end{IEEEeqnarray*}
where $\mathbf{y}^1,\mathbf{y}^2 \in \mathbb{F}_2^{n_1}$, $\mathbf{H}^1_{ik},\mathbf{H}^2_{ik}\in \mathbb{F}_2^{n_1\times n_1}$, $\mathbf{x}_{ik}\in \mathbb{F}_2^{n^j_{ik}}$ and $\mathbf{x}^c_{ik}\in \mathbb{F}_2^{n_i}$ with $i,j,k \in \{1,2\}$. Note that we used the specific modulation of \eqref{modulation}, to group the interference signal parts. The bits of the channel matrices correspond to the real values of $g_{ik}^j$, defined above.

\section{Main Result}
\label{Main Results}

In the following section we give an overview of the main results of this paper. We start by presenting the results for the linear deterministic approximation of the IMAC. We will see that the approximation is insufficient in a sense, that the achievable sum rate of the scheme does not achieve the upper bound at all points. And a possible transfer to the Gaussian IMAC, based on this structure, would also inherit the problem. We will therefore turn to the lower triangular deterministic model and show that it allows schemes which can achieve the upper bound and therefore establish the possibility of a constant-gap capacity result. In the last part we present the constant-gap capacity result for the G-IMAC.


\subsection{Approximate Capacity for the LD-IMAC}
\begin{figure}[b]
\centering
\includegraphics[scale=0.5]{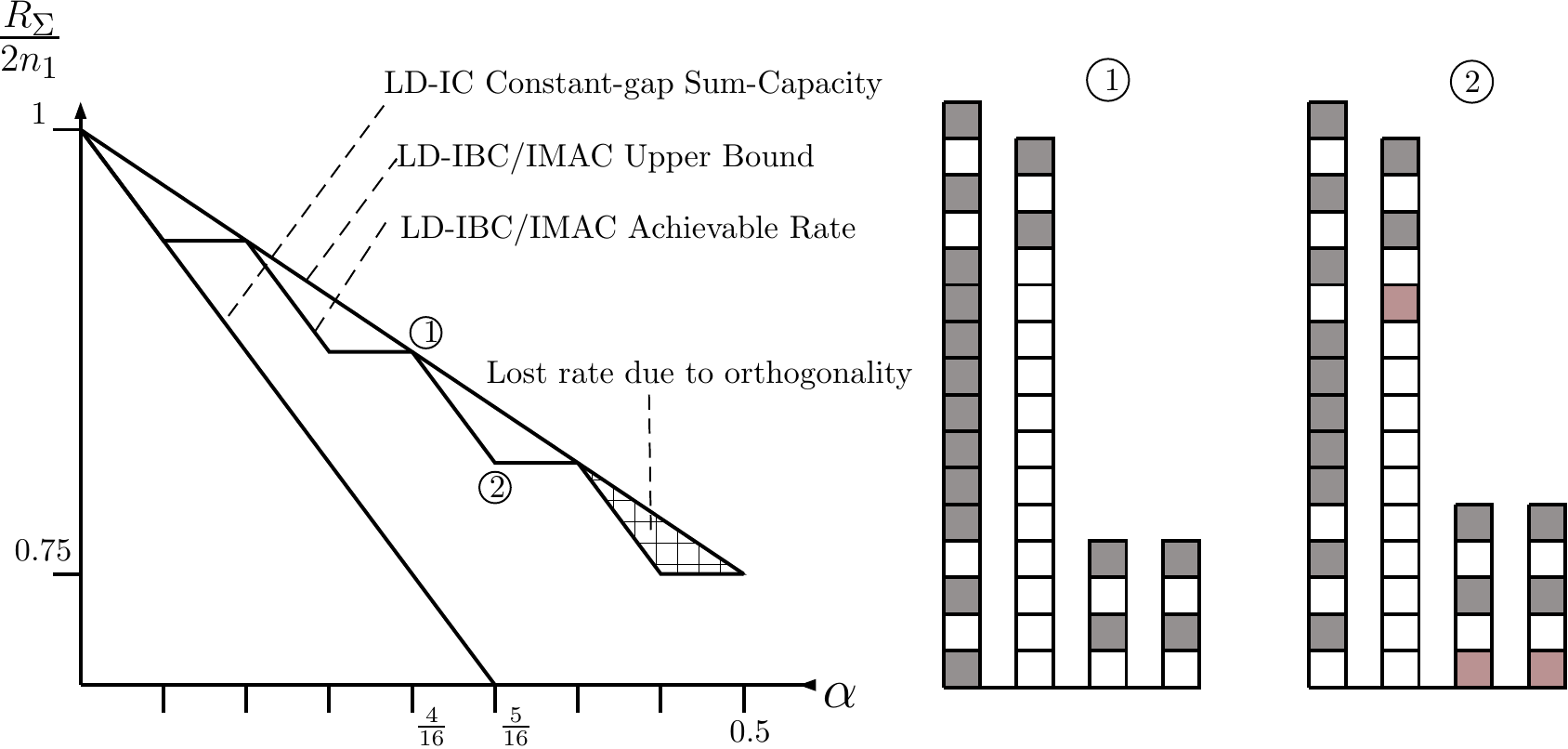}
\caption{GDoF of the LD-IMAC and LD-IBC in the weak interference regime in comparison to the LD-IC channel. The first part of the w-curve is depicted. Multi-user gain be seen, as a second user in each cell provides a DoF-wise gain of one-half of the interference strength in each cell.}
\label{ld-achiev-uB}
\end{figure}
In section \ref{Proof LD-Theorems} we will show that the following sum rate is achievable in the weak interference regime. We define the weak interference regime as 
\begin{equation}
n_j^i+n_i^j\leq \min \{n_{12}^1,n_{22}^2\}.
\end{equation}
This means that the sum of both interference link strength values is smaller than the weakest direct link strength of both cells. In this regime, the system model can be decomposed into two sub channels (see section \ref{Proof LD-Theorems}) which simplifies the analysis. The purpose is to demonstrate the shortcomings of the LD model. Later analysis of the LTD model will treat the general interference case. 
The achievable sum-rate depends on the strength, or coarse channel gain, of the weaker user, which results in three sub-cases per cell $i$. If the coarse channel gain $n^i_{i2}$ of the weaker link signal $x^i_{i2}$ is smaller than the coarse interference channel gain $n_j^i$ in cell $i$, the IMAC sum-rate falls back to the IC sum-rate. Above this threshold, multiuser gain will increase along with $n^i_{i2}$ until it reaches the maximum multiuser gain. In the following we will briefly introduce the results for the last case, where the full multi-user gain can be achieved simultaneous in both cells.
\begin{theorem}
\label{ld-achiev}
The achievable sum-rate for the linear deterministic interfering multiple access channel (LD-IMAC) in the weak interference regime, is
\begin{equation*}
R_\Sigma \leq n^1_{12}+n^2_{22}-n^1_2-n^2_1 + \phi(n^1_2,\Delta_1) +\phi (n^2_1,\Delta_2)
\end{equation*}
with the function $\phi$ for $p,q \in \mathbb{N}_0$, following the notation of \cite{Buhler2011}, defined as
\begin{equation*}
\phi(p,q):=
\begin{cases}
q+\frac{l(p,q)q}{2} & \text{if } l(p,q)\text{ is even,}
\\
p-\frac{(l(p,q)-1)q}{2} & \text{if } l(p,q)\text{ is odd},
\end{cases}
\label{phi}
\end{equation*}
where $l(p,q):=\lfloor\frac{p}{q}\rfloor\ \mbox{for}\ q>0\ \mbox{and}\ l(p,0)=0$. \end{theorem}
$\phi$ is essentially composed of the multiuser gain and the difference of the coarse channel gain of both user in the corresponding cell. The proof of the theorem can be found in section \ref{proof ld-achievable}. To provide a result about the optimality of this sum-rate, the next lemma will give an upper bound for the aforementioned model.

\begin{theorem}[]
\label{ld-bound}
The sum rate for the linear deterministic interfering multiple access channel (LD-IMAC) in the weak interference regime can bounded from above by 
\begin{equation*}
R_\Sigma \leq n^1_{11}+n^2_{21}-\frac{n^2_1}{2}-\frac{n^1_2}{2}.
\end{equation*}
\end{theorem}
We remark, that the LD model can be regarded as a special case of the LTD model, where the lower triangular matrix is the identity matrix. 
The upper bound for the weak interference part of the LTD-IMAC therefore also provides an upper bound  for the LD-model. The proof can be found in section \ref{UB_LTDM}.
A graphical comparison by means of GDoF shows, that the achievable sum-rate of Theorem \ref{ld-achiev} reaches the upper bound at the $\alpha$ values at which the scheme has an even number of layers (see Fig.~\ref{ld-achiev-uB}).
However, in between those points the upper bound cannot be reached. The reason for this behaviour is the structure of the scheme which originates from the approximation model. Later, we will show that the lower triangular deterministic model eliminates this problem.
By utilizing layered lattice coding schemes as in \cite{Sridharan08}, \cite{Suvarup2011} and \cite{Bresler2010} we can transfer the LD achievable sum-rate to the Gaussian channel, which yields the following theorem.
\begin{theorem}
\label{ld-gauss-achiev}
The achievable sum-rate for the Gaussian interfering multiple access channel (G-IMAC) in the weak interference regime, is 
\begin{IEEEeqnarray*}{rCl}
R_{\Sigma}&>&\log \mbox{SNR}_{11}^{\beta_1} -\log \mbox{SNR}_{11}^{\alpha_1}+\log \mbox{SNR}_{21}^{\beta_2} -\log \mbox{SNR}_{21}^{\alpha_2}\\
&&+\:\phi(\log \mbox{SNR}_{21}^{\alpha_2},\log \mbox{SNR}_{11}^{(1-\beta_1)})\\
&&+\:\phi(\log \mbox{SNR}_{11}^{\alpha_1},\log \mbox{SNR}_{21}^{(1-\beta_2)})\\
&&-\:6-2.5(\lfloor L_2 \rfloor+ \lfloor L_1 \rfloor)
\end{IEEEeqnarray*}
\end{theorem}
with $\phi$ defined as in (\ref{phi}). Note that l(p,q) is equivalent to $\lfloor L_i \rfloor$, which basically counts the layers of lattice codes. The proof for the theorem is provided in section \ref{transfer from ld to gaussian channel}.
Observe that the sum-rate has the same structure as in Theorem \ref{ld-achiev} with the correspondence $n_{ik}^j=\lceil\log |h_{ik}^j|^2P\rceil $. 
We therefore need to prove an equivalent version of the bound in Theorem \ref{ld-bound} for the Gaussian case. By extending proof methods of \cite{Bresler2008} we can show the following result.
\begin{theorem}
\label{ld-gauss-bound}
The sum-rate for the Gaussian IMAC can be upper bounded by the corresponding LD-IMAC upper bound, within a constant number of bits. 
The G-IMAC in the weak interference regime is therefore bounded from above by
\begin{equation*}
R_\Sigma \leq n^1_{11}+n^2_{21}-\frac{n^2_1}{2}-\frac{n^1_2}{2}+c_1,
\end{equation*}
\label{mytheoremx}
\end{theorem}where $c_1$ is a constant. Therefore the two theorems \ref{ld-gauss-achiev} and \ref{ld-gauss-bound} show, that the transfer of the results from the linear deterministic scheme to the Gaussian IMAC is possible. However, the achievable scheme inherits the same structure as in the LD model and therefore also the problems with the step-like achievablity curve. These problems are related to the fact, that the LD schemes are constraint to orthogonal usage of bit-levels. Assuming independence of bit-levels would yield a tight upper bound. But we would therefore need to proof the optimality of an uniform input distribution. However, a result along these lines would be restricted to the linear deterministic model and does not extend to the general Gaussian case. Instead, one can use a more complex model, which enables schemes beyond orthogonal usage of bit-levels. The lower triangular deterministic model is such a model and we will now describe the results for this model.

\subsection{Approximate Constant-Gap Sum Capacity for the LTD-IMAC}
We start with a theorem about the constant-gap sum capacity for the deterministic LTD-IMAC in the weak interference regime with arbitrary channel gains.

\begin{theorem}
\label{weak_constant_gap}
For every $\delta\in(0,1]$, $n^k_{k1}$, $n^k_{k2}$, $n_l^k\in \mathbb{N}$ such that $n^k_{k1}\geq n_{k2}^k\geq n_l^k$ and
$n_2^1+n_1^2 \leq \min \{n_{11}^1,n_{21}^2\}$ with $k,l \in \{1,2\}$, $k\neq l$, there exists a set $B\subset (1,2]^{2\times 3}$
of Lebesgue measure at most $\delta$ such that for all channel gains $g_{ik}^j\in(1,2]^{2\times 3} \setminus B$, the sum capacity $C_{\text{LTD-IMAC}}$ of the lower triangular deterministic interfering multiple access channel satisfies
\begin{equation*}
D-2 \log (c/\delta)\leq C_{\text{LTD-IMAC}} \leq D
\end{equation*}

with $D:=\min\{D_1,D_2,D_3,D_4\}$
and
\begin{IEEEeqnarray*}{rCl}
D_1&:=&\max\{(n_{11}^1-n_1^2),n_{12}^1\}+\max\{(n_{21}^2-n_2^1),n_{22}^2\}\\
D_2&:=&\max\{(n_{11}^1-n_1^2),n_{12}^1\}+n_{21}^2-\tfrac{1}{2}n_2^1\\
D_3&:=&n_{11}^1-\tfrac{1}{2}n_1^2+\max\{(n_{21}^2-n_2^1),n_{22}^2\}\\
D_4&:=&n_{11}^1-\tfrac{1}{2}n_1^2+n_{21}^2-\tfrac{1}{2}n_2^1,
\end{IEEEeqnarray*}
for some constant c, independent of the channel gain.\end{theorem}
The proof is provided in section \ref{Proof for LTD-IMAC}. Recall that $n_2^1$ is the coarse interference strength of both users from cell 2 to cell 1, and equally, $n_1^2$ from cell 2 to cell 1.
This means that the assumption $n_2^1+n_1^2 \leq \min \{n_{11}^1,n_{21}^2\}$, states that the sum of the interference strengths is smaller than the minimum of $\{n_{11}^1,n_{21}^2\}$, which are the two coarse channel gains of the stronger users in each cell. For a symmetric setting, the assumption becomes $\alpha=\tfrac{n_i}{n_1}\leq \tfrac{1}{2}$ and corresponds to the weak interference regime for the IC-channel. The $\delta$ in the gap $2 \log (c/\delta)$ can be seen as a fixed trade-off factor. If $\delta$ is chosen to be large, the gap would become small. The achievable sum-rate would get closer to the bound. But $\delta$ is also the bound for the measure of the outage set $B$ and a large bound would mean, that the result would hold for a smaller set of channel gain configurations. On the other hand, a small $\delta$ would mean, that the result holds for a large set of channel gains but at the same time induce a bigger gap towards the upper bound. Due to the large number of cases for arbitrary channel gains and arbitrary interference regimes, we have limited the investigation of arbitrary regimes to symmetric channel gain configurations. The following theorem shows the result for the deterministic IMAC.

\begin{theorem}
\label{ltd-imac-thm}
For every $\delta\in(0,1]$, $n_1$, $n_2$, $n_i \in \mathbb{N}$ such that $n_1\geq n_2$, there exists a set $B\subset (1,2]^{2\times 3}$
of Lebesgue measure at most $\delta$ such that for all channel gains $g_{ik}^j\in(1,2]^{2\times 3} \setminus B$, the sum capacity $C_{\text{LTD-IMAC}}$ of the lower triangular deterministic symmetric interfering multiple access channel satisfies
\begin{equation*}
D-2 \log (c/\delta)\leq C_{\text{LTD-IMAC}} \leq D
\end{equation*}
with $D:=\min\{D_1,D_2,D_3,D_4,D_5\}$
and
\begin{IEEEeqnarray*}{rCl}
D_1&:=& 2\max((n_1-n_i)^+,n_i)+\min((n_1-n_i)^+,n_i),\\
D_2&:=& \tfrac{2}{3}(2\max(n_1,n_i)+(n_1-n_i)^+),\\
D_3&:=& 2n_1,\\
D_4&:=& \max(2n_2,2(n_1-n_i)^+,2n_i),\\
D_5&:=& \max(n_1,n_i)+\max(n_2,(n_1-n_i)^+).
\end{IEEEeqnarray*}
for some constant c, independent of the channel gain
\end{theorem}
This result shows the constant-gap result for the symmetrical LTD-IMAC for the whole interference regime. One can see that the symmetrical weak interference cases $D_1$ and $D_4$ are reflected in $D_4$ and $D_1$, respectively. The cases $D_2$, $D_3$, $D_5$ and parts of $D_4$ and $D_1$ represent additional bounds for the cases with interference $\alpha>\tfrac{1}{2}$. As in the weak interference case, one can see that the gap is constant and can be seen as a fixed trade-off factor between rate-gap and quantity of supported channel gains. In the next sub section we will show, that this result can be extended to the Gaussian IMAC.

\subsection{Constant-Gap Sum Capacity for the Gaussian IMAC}
For the Gaussian case, extensions of the methods developed in \cite{Niesen-Ali} show that the achievable schemes can be transferred to the Gaussian IMAC. In this process, the constant-gap gets larger buts stays constant in relation to the channel gain. Moreover, previously used techniques can show that the LTD-IMAC bounds can be used as a bound for the G-IMAC by introducing another constant-gap. This yields the following constant-gap capacity approximation of the G-IMAC.
\begin{theorem}
\label{General_Bound_G_IMAC}
For every $\delta\in(0,1]$, $n_1$, $n_2$, $n_i \in \mathbb{N}$ such that $n_1\geq n_2$, there exists a set $B\subset (1,2]^{2\times 4}$
of Lebesgue measure at most $\delta$ such that for all channel gains $h_{ik}^j\in(1,2]^{2\times 4} \setminus B$, the sum capacity $C_{\text{G-IMAC}}$ of the Gaussian symmetric interfering multiple access channel satisfies
\begin{equation*}
D-2 \log (c_2/\delta)\leq C_{\text{G-IMAC}} \leq D+c_3,
\end{equation*}
with $D:=\min\{D_1,D_2,D_3,D_4,D_5\}$
and
\begin{IEEEeqnarray*}{rCl}
D_1&:=& 2\max((n_1-n_i)^+,n_i)+\min((n_1-n_i)^+,n_i),\\
D_2&:=& \tfrac{2}{3}(2\max(n_1,n_i)+(n_1-n_i)^+),\\
D_3&:=& 2n_1,\\
D_4&:=& \max(2n_2,2(n_1-n_i)^+,2n_i),\\
D_5&:=& \max(n_1,n_i)+\max(n_2,(n_1-n_i)^+).
\end{IEEEeqnarray*}and $c_2,c_3$ are constants.
\end{theorem}

The proof for the achievability can be found in Section \ref{Proof G-IMAC}. It makes use of the scheme for the lower triangular deterministic model and uses a result from number theory, the Khintchine-Groshev Theorem, as well as techniques developed in \cite{Niesen-Ali} and new techniques tailored towards the G-IMAC model. The proof of the upper bound is in the Appendix~\ref{Proof_of_General_Bound_G_IMAC} and utilizes the upper bound of Theorem~\ref{weak_constant_gap}. Note that, as in the two theorems for the LTD channel model, we have a constant-gap result. This is because $c_2$ and $c_3$ are constants which are independent of the channel gain and $\delta$ is a trade-of factor in the same fashion as those in the previous results. This means that a bigger $\delta$ corresponds to a smaller gap but increases the outage set of channel gains for which the method does not work. In the following section we will go into the details of the analysis and provide the proofs for the stated theorems.
\section{Analysis of the LD-IMAC}
\label{Proof LD-Theorems}

\subsection{Achievable Scheme for Theorem~\ref{ld-achiev}}
\label{proof ld-achievable}

The achievability scheme for the IMAC is basically an extended version of the scheme already used for the MAC-P2P in \cite{Buhler2011}. 
Like in the MAC-P2P we split the system (\ref{MAC-MAC}) into two sub systems, $\mathcal{R}_{ach}^{(1)}$ and $\mathcal{R}_{ach}^{(2)}$, see Figure~\ref{system_model_split}. Unlike in \cite{Buhler2011}, both of our sub-systems are identically. The sum of the achievable rates of these two sub systems will constitute the overall sum capacity. The sub systems are given by the equations

\begin{IEEEeqnarray*}{rCl}
\mathbf{y}_1^{(1)} & = & \mathbf{S}^{q^{(1)}-(n^1_{11}-n^1_2)}\mathbf{x}_{11}^{(1)}\oplus \mathbf{S}^{q^{(1)}-(n^1_{12}-n^1_2)}\mathbf{x}_{12}^{(1)}\IEEEyessubnumber\\
\mathbf{y}_2^{(1)} & = & \mathbf{S}^{q^{(1)}-n^2_1}\mathbf{x}_{11}^{(1)}\oplus \mathbf{S}^{q^{(1)}-n^2_1}\mathbf{x}_{12}^{(1)}\oplus \mathbf{S}^{q^{(1)}-n^2_1}\mathbf{x}_{21}^{(1)}\IEEEyessubnumber\\
&&\oplus\: \mathbf{S}^{q^{(1)}-n^2_1}\mathbf{x}_{22}^{(1)} 
\end{IEEEeqnarray*}

for $\mathcal{R}_{ach}^{(1)}$ and 

\begin{IEEEeqnarray*}{rCl}
\mathbf{y}_2^{(2)} & = & \mathbf{S}^{q^{(2)}-(n^2_{21}-n^2_1)}\mathbf{x}_{21}^{(2)}\oplus \mathbf{S}^{q^{(2)}-(n^2_{22}-n^2_1)}\mathbf{x}_{22}^{(2)}\IEEEyessubnumber\\
\mathbf{y}_1^{(2)} & = & \mathbf{S}^{q^{(2)}-n^1_2}\mathbf{x}_{11}^{(2)}\oplus \mathbf{S}^{q^{(2)}-n^1_2}\mathbf{x}_{12}^{(2)}\oplus \mathbf{S}^{q^{(2)}-n^1_2}\mathbf{x}_{21}^{(2)}\IEEEyessubnumber\\
&&\oplus\: \mathbf{S}^{q^{(2)}-n^1_2}\mathbf{x}_{22}^{(2)} 
\end{IEEEeqnarray*}

for $\mathcal{R}_{ach}^{(2)}$.

\begin{figure}
\centering
\includegraphics[width=0.35 \textwidth]{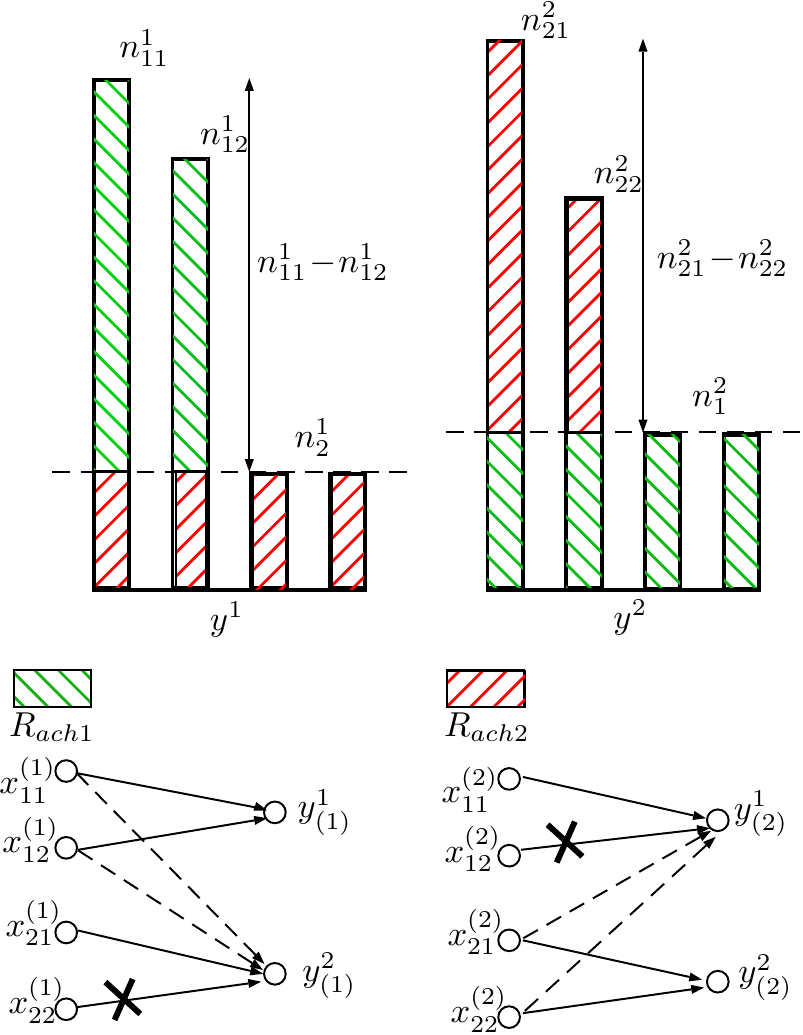}
	\caption{Illustration of the two subsystems. As one can see, letting the second user in each sub-system be silent, results in the one-sided interference MAC-P2P model. Hence, achievable schemes are applicable.}
    \label{system_model_split}
\end{figure}

Examining the resulting sub-systems, one can see that leaving one private part, at the side of the interference silent, results in a sub-system equal to $\mathcal{R}_{ach}^{(2)}$ in \cite{Buhler2011}. In particular we choose $\mathbf{X}_{22}^{(1)}:=\mathbf{0}$ and $\mathbf{X}_{12}^{(2)}:=\mathbf{0}$ resulting in 

\begin{IEEEeqnarray*}{rCl}
\mathbf{y}_1^{(1)} & = & \mathbf{S}^{q^{(1)}-(n^1_{11}-n^1_2)}\mathbf{x}_{11}^{(1)}\oplus \mathbf{S}^{q^{(1)}-(n^1_{12}-n^1_2)}\mathbf{x}_{12}^{(1)}\IEEEyessubnumber\\
\mathbf{y}_2^{(1)} & = & \mathbf{S}^{q^{(1)}-n^2_1}\mathbf{x}_{11}^{(1)}\oplus \mathbf{S}^{q^{(1)}-n^2_1}\mathbf{x}_{12}^{(1)}\oplus \mathbf{S}^{q^{(1)}-n^2_1}\mathbf{x}_{21}^{(1)}\IEEEyessubnumber
\end{IEEEeqnarray*}

for $\mathcal{R}_{ach}^{(1)}$ and 

\begin{IEEEeqnarray*}{rCl}
\mathbf{y}_2^{(2)} & = & \mathbf{S}^{q^{(2)}-(n^2_{21}-n^2_1)}\mathbf{x}_{21}^{(2)}\oplus \mathbf{S}^{q^{(2)}-(n^2_{22}-n^2_1)}\mathbf{x}_{22}^{(2)}\IEEEyessubnumber\\
\mathbf{y}_1^{(2)} & = & \mathbf{S}^{q^{(2)}-n^1_2}\mathbf{x}_{11}^{(2)}\oplus \mathbf{S}^{q^{(2)}-n^1_2}\mathbf{x}_{21}^{(2)}\oplus \mathbf{S}^{q^{(2)}-n^1_2}\mathbf{x}_{22}^{(2)} \IEEEyessubnumber
\end{IEEEeqnarray*}

for $\mathcal{R}_{ach}^{(2)}$.

The achievable sum rates for the systems are defined as 

\begin{equation}
R_\Sigma^{(1)} \leq n^2_1 + \zeta^{(1)} +\phi(n^1_2,\Delta_1)
\end{equation}
\begin{equation}
R_\Sigma^{(2)} \leq n^1_2 + \zeta^{(2)}+\phi (n^2_1,\Delta_2).
\end{equation}
Where $\zeta^{(1)}:=n^1_{12}-n^1_2-n^2_1$, $\zeta^{(2)}:=n^2_{22}-n_2^1-n_1^2$ and the function $\phi$ as in \ref{phi}.
\begin{figure}
\centering
\includegraphics[width=0.3 \textwidth]{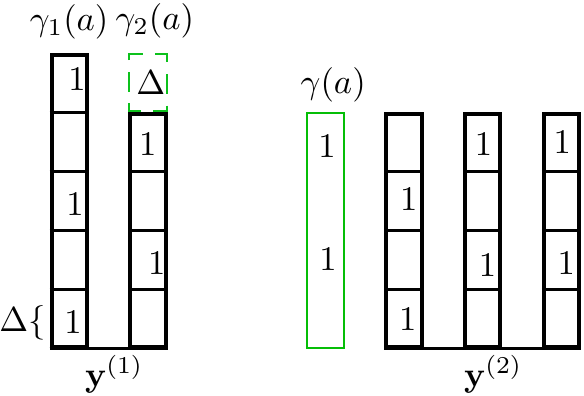}
	\caption{Illustration of the $\gamma$-functions.}
    \label{system model}
\end{figure}
It suffices to show the achievability of one sub-systems, the results for the other sub-systems follows by symmetry.
Consider the sum capacity $R_\Sigma^{(1)}$, let $\mathbf{a} \in \mathbb{F}_2^{n^2_1}$ specify the levels used for encoding the interference affected part of $\mathbf{x}_{21}$, where $a_i = 1$ if level $i$ is used and $a_i = 0$ otherwise. Define $\gamma(\mathbf{a}) := \mathbf{1}_{n^2_1} - \mathbf{a}$, $\gamma_1(\mathbf{a}) := (\gamma(\mathbf{a}); \mathbf{1}_{\Delta_1})$ and $\gamma_2(\mathbf{a}):= [ \mathbf{0}_{\Delta_1}; \gamma(\mathbf{a})]$. 
Then we can achieve 
\begin{equation}
R_\Sigma^{(1)}\leq |\gamma_1(\mathbf{a})| + |\gamma_2(\mathbf{a})| - \rho(|\mathbf{a}|) + |\mathbf{a}|
\end{equation}
where
\begin{equation} \label{eq:optrho}
\rho(x) :=\underset{\mathbf{a} \in \mathbb{F}_2^{n^2_1}: |\mathbf{a}| = x}{\text{min}}~ \gamma_1(\mathbf{a})^T \gamma_2(\mathbf{a})
\end{equation}
is indicating how many used bit-levels between the common signal parts of $\mathbf{x}_{11}$ and $\mathbf{x}_{12}$ are overlapping. Minimizing $\rho(x)$, with a per-definition non-overlapping 
$\gamma(\mathbf{a})$ gives a solution with maximal direct rate and minimal interference. 
An assignment vector $\mathbf{a}$ solving (\ref{eq:optrho}) for a given $x$ can be shown to be of the form described in the following. Let $l= n^2_1 \text{~div~} \Delta_1 $ and $Q = n^2_1 \text{~mod~} \Delta_1$, i.e., $n^2_1 = l\Delta_1 + Q$. We subdivide $\mathbf{a}$ into $l \Delta_1$ subsequences (blocks) of length $\Delta_1$ and one remainder block of length $Q$. We distribute ones over $\mathbf{a}$ until $x$ entries in $\mathbf{a}$ have been set to $1$: We start with the even-numbered blocks, followed by the remainder block. If $l$ is even, we finally distribute over the odd-numbered blocks. If $l$ is odd, we also fill the odd-numbered blocks, but in reversed (decreasing) order. To be precise, we define for the case that $l$ is even 
\begin{IEEEeqnarray*}{rCl}
\mathbf{A}_{\text{even}} &=& \left(\mathbf{0}_{1\times \frac{l}{2}};\mathbf{e}_{k}\right)_{k=1}^{l/2}  \otimes \mathbf{I}_{\Delta_1},\\
\mathbf{A}_{\text{odd}} &=& \left(\mathbf{e}_{k};\mathbf{0}_{1\times \frac{l}{2}}\right)_{k=1}^{l/2} \otimes \mathbf{I}_{\Delta_1}
\end{IEEEeqnarray*}and 
\begin{IEEEeqnarray*}{rCl}
\mathbf{A}_{\text{even}} &=& \left[\left(\mathbf{0}_{1\times \frac{(l-1)}{2}};\mathbf{e}_{k}\right)_{k=1}^{(l-1)/2};\mathbf{0}_{1 \times \frac{l-1}{2}} \right] \otimes \mathbf{I}_{\Delta_1},\\
\mathbf{A}_{\text{odd}} &=& \mathbf{M}_{l\Delta_1} \left(\left[\left(\mathbf{e}_{k};\mathbf{0}_{1\times \frac{(l+1)}{2}}\right)_{k=1}^{(l-1)/2};\mathbf{e}_{\frac{l+1}{2}} \right] \otimes \mathbf{I}_{\Delta}\right)
\end{IEEEeqnarray*}
for odd $l$. Here, $\otimes$ denotes the Kronecker product, $\mathbf{e}_k$ the unit row vector of appropriate size with $1$ at position $k$ and $\mathbf{M}_N = (\mathbf{e}_{N-k+1})_{k=1}^{N}$ is the flip matrix. Then the matrix
\begin{equation}
 \mathbf{P} = \left[\mathbf{A}_{\text{even}} |\mathbf{0}_{l\Delta\times Q} |\mathbf{A}_{\text{odd}}  \right],
\end{equation}
gives an optimal assignment vector $\mathbf{a}$ by setting $\mathbf{a} =\mathbf{P}[\mathbf{1}_{x};\mathbf{0}_{n^2_1 - x}]$.

Finally, the sum rate for the overall system can be obtained by adding the sub systems: $R_\Sigma^{(1)}+R_\Sigma^{(2)} = R_\Sigma$
 \begin{IEEEeqnarray*}{rCl}
R_\Sigma &\leq & n^2_1 + \zeta^{(1)} +\phi(n^1_2,\Delta_1)+n^1_2 + \zeta^{(2)}+\phi (n^2_1,\Delta_2)\\
& = & n_{12}+n_{22}-n^1_2-n^2_1 + \phi(n^1_2,\Delta_1) +\phi (n^2_1,\Delta_2).\IEEEyesnumber\label{IMAC-ach}
\end{IEEEeqnarray*}$\hfill\IEEEQEDclosed$
\begin{figure}
\centering
\includegraphics[scale=0.5]{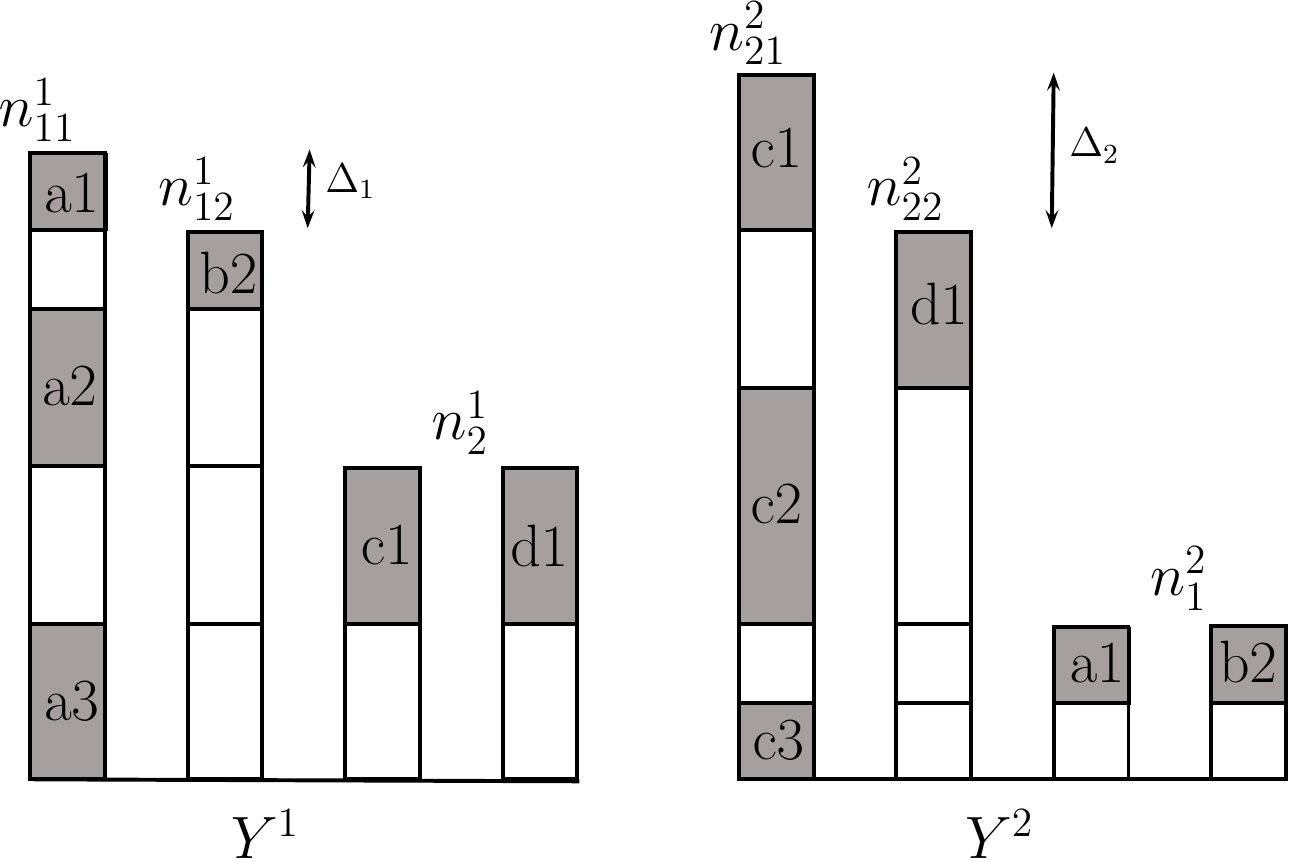}
\caption{An example for a scheme which achieves the upper bound is presented in the figure. The MAC-cell 1 has $n_{11}^1$ and $n_{12}^1$ bit levels in the direct paths and generates interference, at the other cell, of $n_1^2$ bit levels. Whereas the MAC cell 2 has $n_{21}^2$ and $n_{22}^2$ bit levels in the direct path and generates $n_2^1$ bit levels interference. It can be seen that the scheme utilizes the full bit-level range of the receivers except half of the incoming interference effected bit-levels and therefore reaches the upper bound  $R_{\Sigma}\leq n_{11}^1+n_{21}^2-\tfrac{1}{2}n_1^2-\tfrac{1}{2}n_2^1$.}
\label{MAC_MAC_SCHEMA}
\end{figure}

\subsection{Transfer from LD-IMAC to G-IMAC}
\label{transfer from ld to gaussian channel}
In this section we prove Theorem~\ref{ld-gauss-achiev}, and therefore show the achievable sum-rate for the G-IMAC based on lattice coding schemes with an inherited structure of the linear deterministic schemes.
\subsubsection{Interference Regime}

For this part, we assume without loss of generality that $h_{i1}^j\geq h_{i2}^j$ for $i=j$. This means that the first direct path is stronger or equal than the second direct path in each cell.
Also, remember that we have equal (coarse) interference strength at the receivers. With the specific modulation from \eqref{Gauss_Model_g}, we can assume that
\begin{equation}
 h_{i1}^j=h_{i2}^j \mbox{ for } i\neq j.
 \label{Equal_Interference}
\end{equation}
Furthermore we define two expressions, the signal-to-noise ratio and the interference-to-noise ratio as:
\begin{equation}
|h_{ik}^j|^2P=
\begin{cases}
\mbox{SNR}_{ik} & \text{if }  i=j
\\
\mbox{INR}_{i}^j & \text{if } i\neq j.
\end{cases}
\end{equation}
We also introduce two parameters $\alpha_i,\beta_i$ which combine these ratios with $\mbox{SNR}_{i2}\,=\,\mbox{SNR}_{i1}^{\beta_i}$ and $\mbox{INR}_{j}^i\,=\,\mbox{SNR}_{i1}^{\alpha_i}$. These parameters correspond to $\alpha, \beta$ which are used in the LDM channel model \cite{Fritschek2014a} and in GDoF considerations.
Now we can restrict the investigation to the weak interference regime defined through 
\begin{equation}
\mbox{INR}_1^2+\mbox{INR}_2^1 \leq \min\{ \mbox{SNR}_{12},\mbox{SNR}_{22}\}.
\label{weak-int-def}
\end{equation}
We assume that $\mbox{SNR}_{ik}>1$, otherwise the user can be left silent which results in a 1 bit penalty of the sum-rate. For convenience we use the standard terms: common and private signal, for the part which is seen at both cells and the part which is only received in the intended cell, respectively. Also note that the following techniques work for all channel parameters in the defined regime, except for a singularity at $\beta_i=1$ in which the additional gain for that corresponding cell will be zero. These techniques are therefore not limited to any kind of rational numbers set, as alignment comes natural due to the LDM schemes and the assumption \eqref{Equal_Interference}.

\subsubsection{Power Partitioning}
The power as observed at each receiver is partitioned into intervals. The number of partitions is depended on the channel structure, in particular the $\beta_i,\alpha_i$-parameters. These partitions play the role of intervals of bit levels in the LDM. The common signal part of cell $i$ is partitioned into 
\begin{equation}
\lfloor L_i \rfloor = \left\lfloor \frac{\log\mbox{SNR}_{j1}^{\alpha_j}}{\log \mbox{SNR}_{i1}^{(1-\beta_i)}}\right\rfloor
\end{equation} intervals. Moreover, there is an additional reminder block. The private signal part consists of an interference effected part, which is partitioned into $ \lfloor L_j \rfloor$ intervals with an additional reminder part, as well as an interference free part. We therefore have a total of $\l_{\max}\leq \lfloor L_i \rfloor+\lfloor L_j \rfloor+3$ power partitions, depending on the number of non-zero reminder parts and if there is an additional interference free private signal part.
Signal power is defined as
\begin{IEEEeqnarray*}{rCl}
&&\theta_{il}  =  \IEEEyesnumber\\
&&\begin{cases}
\mbox{SNR}_{i1}^{1-(l-1)(1-\beta_i)}-\mbox{SNR}_{i1}^{1-l(1-\beta_i)} & \text{for }  1 \leq l \leq \lfloor L_i \rfloor 
\\
\mbox{SNR}_{i1}^{1-\lfloor L_i \rfloor(1-\beta_i)}-\mbox{SNR}_{i1}^{1-(\lfloor L_i \rfloor+1)(1-\beta_i)} & \text{for } l=\lfloor L_i \rfloor+1\\
\mbox{SNR}_{i1}^{1-(\lfloor L_i \rfloor+1) (1-\beta_i)}-\mbox{SNR}_{i1}^{\alpha_i} & \text{for } l=\lfloor L_i \rfloor+2\\
\tfrac{\mbox{SNR}_{i1}^{\alpha_i}}{\mbox{SNR}_{j1}^{(l-\lfloor L_i \rfloor-2)(1-\beta_j)}}-\tfrac{\mbox{SNR}_{i1}^{\alpha_i}}{\mbox{SNR}_{j1}^{(l-\lfloor L_i \rfloor-1)(1-\beta_j)}} & \text{for } \lfloor L_i \rfloor+2 \leq l\\
& \leq \lfloor L_i \rfloor+\lfloor L_j \rfloor+2  \\
\mbox{SNR}_{i1}^{\alpha_i}\mbox{SNR}_{j1}^{-\lfloor L_j \rfloor(1-\beta_j)}-1 & \text{for } l=l_{\max}
\end{cases}
\label{theta_k_fractions}
\end{IEEEeqnarray*}
with $l$ indicating the specific partition (Fig. \ref{Powersplit}). Note that this is for the cases where $\lfloor L_i \rfloor=\text{odd}$. The additional remainder term at $l\,=\,\lfloor L_i \rfloor+1$ vanishes for $\lfloor L_i \rfloor=\text{even}$, because it can be merged with the subsequent  partition. In that case, all following powerlevels need to be changed accordingly. Each user $k$ of cell $i$ decomposes its signal into a sum of independent sub-signals \begin{equation}
\mathbf{x}_{ik}=\sum\limits_{l=1}^{l_{\max}} \mathbf{x}_{ik}(l).
\end{equation} 
This decomposition can be seen as a message split, where every message is separately encoded by a lattice code. Note that user 2 can only send at power levels $l>1$, due to its power constrains.

\subsubsection{Layered Nested Lattice Codes}
Instead of the  Loeliger-type \cite{Loeliger} lattice codes, used in \cite{Fritschek2014b}, we will use n-dimensional nested lattice codes introduced in \cite{UrezZamir} which can achieve capacity in the AWGN single-user channel. This results in a slightly lower gap in the overall rate terms. A lattice $\Lambda$ is a discrete subgroup of $\mathbb{R}^n$ which is closed under real addition and reflection. Moreover, denote the nearest neighbour quantizer by 
\begin{equation*}
Q_{\Lambda}(\mathbf{x}):=\arg\min_{\mathbf{t}\in \Lambda} ||\mathbf{x}-\mathbf{t}||.
\end{equation*} The fundamental Voronoi region $\mathcal{V}(\Lambda)$ of a lattice $\Lambda$ consists of all points which get mapped or quantized to the zero vector. The modulo operation is defined as
\begin{equation*}
[\mathbf{x}]\mod\Lambda := \mathbf{x}-Q_{\Lambda}(\mathbf{x}).
\end{equation*} A nested lattice code is composed of a pair of lattices $(\Lambda_{\text{fine}},\Lambda_{\text{coarse}})$ , where $\mathcal{V} (\Lambda_{\text{coarse}})$ is the fundamental Voronoi region of the coarse lattice and operates as a shaping region for the corresponding fine lattice $\Lambda_{\text{fine}}$. It is therefore required that $\Lambda_{\text{coarse}} \subset \Lambda_{\text{fine}}$. Such a code has a corresponding rate $R$ equal to the log of the nesting ratio. A part of the split message is now mapped to the corresponding codeword $\mathbf{u}_{ik}(l)\in \Lambda_{\text{fine},l-1} \cap \mathcal{V} (\Lambda_{\text{coarse},l})$, which is a point of the fine lattice inside the fundamental Voronoi region of the coarse lattice. Note that $\Lambda_{l_{max}} \subset \cdots\subset \Lambda_1$. 
The code is chosen such that it has a power of $\theta_{il}$. The codeword $\mathbf{x}_{ik}(l)$ is now given as 
\begin{equation}
\mathbf{x}_{ik}(l)=[\mathbf{u}_{ik}-\mathbf{d}_{ik}] \mod  \Lambda_l,
\end{equation} 
where we dither (shift) with $\mathbf{d}_{ik}\sim \mbox{Unif} (\mathcal{V}(\Lambda_l))$ and reduce the result modulo-$\Lambda_l$. Transmitter $ik$ now sends  a scaled $\mathbf{x}_{ik}$ over the channel, such that the power per sub-signal $\mathbf{x}_{ik}(l)$ is $\frac{\theta_{il}}{|h_{ik}^i|^2}$ and receivers see a power of $\theta_{il}$. Due to the partitioning construction, the $\mathbf{x}_{ik}$ satisfy the power restriction of $P$ for user 1,
\begin{equation}
\sum\limits_{l=1}^{l_{\max}} \frac{\theta_{il}}{|h_{i1}^i|^2} \leq \frac{\mbox{SNR}_{i1}}{|h_{i1}^i|^2} = P
\end{equation}
and user 2
\begin{equation}
\sum\limits_{l=2}^{l_{\max}} \frac{\theta_{il}}{|h_{i2}^i|^2} \leq \frac{\mbox{SNR}_{i2}}{|h_{i2}^i|^2} = P
\end{equation} in each cell i.

Moreover, aligning sub-signals use the same code (with independent shifts).

In \cite{UrezZamir} it was shown that nested lattice codes can achieve the capacity of the AWGN single-user channel with vanishing error probability. Viewing each of our power intervals as a channel, we therefore have that
\begin{equation}
R(l)\leq \log \left(1+\frac{\theta_{il}}{N_i(l)}\right),
\label{Decodingbound}
\end{equation}
where $N_i(l)$ denotes the noise variance per dimension of the sub-sequent levels. If a sum of $K$ lattice points gets aligned at a power level $\theta_{il}$ and we want to decode the lattice point corresponding to this sum, then the achievable rate is given by
\begin{equation}
R(l)\leq \log \left(\tfrac{1}{K}+\frac{\theta_{il}}{N_i(l)}\right),
\label{Decodingbound2}
\end{equation}
which is obtained through minimization of the denominator in \cite[Theorem 1]{ComputeUForward} and MMSE scaled decoding.

\subsubsection{Alignment}
Due to the construction, shifted codewords $\mathbf{x}_{i1}(l)$ and $\mathbf{x}_{i2}(l+1)$ are received on separate power levels at the intended receiver and align on the same power level at the unintended receiver. As an example we look at the codewords $\mathbf{x}_{11}(2)$ and $\mathbf{x}_{12}(3)$. They are transmitted from the first and second transmitter of cell 1, with a scaled power of $\tfrac{\theta_{12}}{|h_{11}^1|^2}$ and $\tfrac{\theta_{13}}{|h_{12}^1|^2}$, respectively. This means that receiver 1 sees them with power 
\begin{equation}
E(||h_{11}^1|\mathbf{x}_{11}(2)|^2)=\tfrac{\theta_{12}|h_{11}^1|^2}{|h_{11}^1|^2}=\theta_{12}.
\end{equation} 
assuming that level 2 and 3 are within the common signal part we have that $\theta_{12}=(|h_{11}^1|^2P)^{1-1(1-\beta_1)}-(|h_{11}^1|^2P)^{1-2(1-\beta_1)}$ and
\begin{equation}
E(||h_{12}^1|\mathbf{x}_{12}(3)|^2)=\tfrac{\theta_{13}|h_{12}^1|^2}{|h_{12}^1|^2}=\theta_{13}.
\end{equation} 
with $\theta_{13}=(|h_{11}^1|^2P)^{1-2(1-\beta_1)}-(|h_{11}^1|^2P)^{1-3(1-\beta_1)}$.
Clearly both codewords are received on different levels as long as $\beta_1\neq 1$. For the unintended receiver 2, these codewords are received with power 
\begin{equation}
E(||h_{11}^2|\mathbf{x}_{11}(2)|^2)=\tfrac{\theta_{12}|h_{11}^2|^2}{|h_{11}^1|^2}.
\label{received_power1}
\end{equation} 
and 
\begin{equation}
E(||h_{12}^2|\mathbf{x}_{12}(2)|^2)=\tfrac{\theta_{13}|h_{12}^2|^2}{|h_{12}^1|^2}.
\label{received_power2}
\end{equation} 
Now we need to show that \eqref{received_power1} and \eqref{received_power2} are the same. For \eqref{received_power1} we have that 
\begin{IEEEeqnarray*}{rCl}
\tfrac{\theta_{12}|h_{11}^2|^2}{|h_{11}^1|^2} &=& \left(\mbox{SNR}_{11}^{1-1(1-\beta_1)}-\mbox{SNR}_{11}^{1-2(1-\beta_1)}\right)\tfrac{|h_{11}^2|^2}{|h_{11}^1|^2}\\
&=&\left(\tfrac{\mbox{SNR}_{11}}{\mbox{SNR}_{11}^{(1-\beta_1)}}-\tfrac{\mbox{SNR}_{11}}{\mbox{SNR}_{11}^{2(1-\beta_1)}}\right)\tfrac{\mbox{INR}_1^2}{\mbox{SNR}_{11}}\\
&=&\tfrac{\mbox{INR}_1^2}{\mbox{SNR}_{11}^{(1-\beta_1)}}-\tfrac{\mbox{INR}_1^2}{\mbox{SNR}_{11}^{2(1-\beta_1)}}.
\end{IEEEeqnarray*} For \eqref{received_power2} we have that 

\begin{IEEEeqnarray*}{rCl}
\tfrac{\theta_{13}|h_{11}^2|^2}{|h_{11}^1|^2} &=& \left(\mbox{SNR}_{11}^{1-2(1-\beta_1)}-\mbox{SNR}_{11}^{1-3(1-\beta_1)}\right)\tfrac{|h_{11}^2|^2}{|h_{12}^1|^2}\\
&=&\left(\tfrac{\mbox{SNR}_{11}}{\mbox{SNR}_{11}^{2(1-\beta_1)}}-\tfrac{\mbox{SNR}_{11}}{\mbox{SNR}_{11}^{3(1-\beta_1)}}\right)\tfrac{\mbox{INR}_1^2}{\mbox{SNR}_{12}}\\
&=&\left(\tfrac{\mbox{SNR}_{11}^{\beta_1+(1-\beta_1)}}{\mbox{SNR}_{11}^{2(1-\beta_1)}}-\tfrac{\mbox{SNR}_{11}^{\beta_1+(1-\beta_1)}}{\mbox{SNR}_{11}^{3(1-\beta_1)}}\right)\tfrac{\mbox{INR}_1^2}{\mbox{SNR}_{12}}\\
&=&\left(\tfrac{\mbox{SNR}_{11}^{\beta_1}}{\mbox{SNR}_{11}^{(1-\beta_1)}}-\tfrac{\mbox{SNR}_{11}^{\beta_1}}{\mbox{SNR}_{11}^{2(1-\beta_1)}}\right)\tfrac{\mbox{INR}_1^2}{\mbox{SNR}_{12}}\\
&=&\tfrac{\mbox{INR}_1^2}{\mbox{SNR}_{11}^{(1-\beta_1)}}-\tfrac{\mbox{INR}_1^2}{\mbox{SNR}_{11}^{2(1-\beta_1)}}
\end{IEEEeqnarray*}

where we used that $\mbox{SNR}_{11}^{\beta_1}=\mbox{SNR}_{12}$ as defined.

\subsubsection{Decoding Procedure}

Decoding occurs per level, treating subsequent levels as noise. Due to the use of nested lattice codes, a sub-signal $\mathbf{u}_{ik}\mod \Lambda_l$ gets decoded, from which the original sub-signal $\mathbf{x}_{ik}$ can be reconstructed. The reconstructed signal then gets subtracted from the total received signal, leaving the noise part. The noise part constitutes the next level and the process continues. In case of an interference-affected level, only the sum of both sub-signals gets decoded $[\mathbf{u}_{ik}(l)+\mathbf{u}_{ik}(l+1)] \mod  \Lambda_l$. From the sum, the original sum can be reconstructed and subtracted from the received signal. It is therefore a successive decoding scheme, which was proven to work for nested lattice codes in \cite{Nazer2012} and recently applied in \cite{Chaaban16}. Therefore, each level is treated as a Gaussian point-to-point channel, and decodability is assured providing that the lattice rate is chosen appropriately according to \eqref{Decodingbound} or \eqref{Decodingbound2}, depending on the specific case. With a signal power of $\theta_{il}$, it only remains to specify the total noise of each level, consisting of the Gaussian noise at the receiver and the signal power of all subsequent levels, including the interference. The total achievable sum-rate is then given as the sum of all $R_{ik}(l)$ of the used direct levels. 

%
%

\subsubsection{Example of the symmetric restricted IMAC}
\begin{figure}
\centering
\includegraphics[scale=.7]{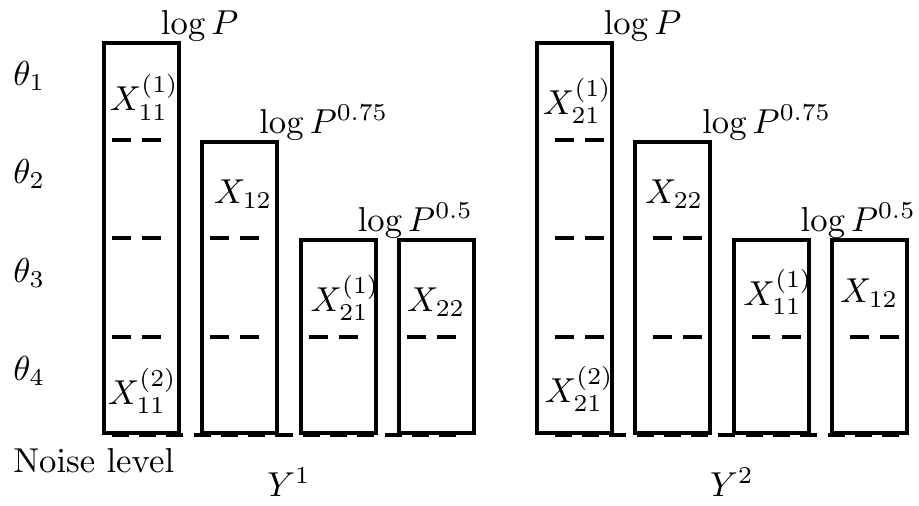}
\caption{Example of the symmetric restricted IMAC: Illustration of power partitioning, with the resulting 4 signal power levels and level use for coding}
\label{Powersplit}
\vspace{-1em}
\end{figure}
In this example we consider the fairly restricted symmetric IMAC channel, where $\alpha=0.5$ and $\beta=0.75$. Note that this example is technically not in the weak interference regime \eqref{weak-int-def} any more. However, the alignment strategies of this section can be applied to channels with the relaxed condition $\mbox{INR}_1^2+\mbox{INR}_2^1 \leq \min\{ \mbox{SNR}_{11},\mbox{SNR}_{21}\}$ by careful handling of certain cases. This condition becomes $\alpha\leq 0.5$ in our symmetrical example. The symmetry assumptions yields that $\mbox{SNR}_{11}=\mbox{SNR}_{21}=P$, $\mbox{SNR}_{12}=\mbox{SNR}_{22}=P^{\beta}$ and $\mbox{INR}_1^2=\mbox{INR}_2^1=P^{\alpha}$. Therefore we have that $L_i=\lfloor L_i \rfloor=\lfloor \tfrac{\alpha}{1-\beta } \rfloor=2$. The scheme \eqref{theta_k_fractions} yields $l_{\max}=4$ levels with $\theta_1=P-P^\beta$, $\theta_2=P^\beta-P^\alpha$, $\theta_3=P^\alpha-P^{0.25}$ and $\theta_4=P^{0.25}-1$. Moreover, we have the following noise powers per level, $N(1)=1+(P^{0.75}-1)+\theta_3$, $N(2)=1+(P^{0.5}-1)+\theta_3$, $N(3)=1+(P^{0.25}-1)$ and $N(4)=1$. We therefore have the following decoding bounds for both cells:
 \begin{IEEEeqnarray*}{rCl}
 \{\mathbf{x}_{11}(1),\mathbf{x}_{21}(1)\}:\qquad R(1) &\leq& \min \{r_1,r_3\}\\
 \{\mathbf{x}_{12}(2),\mathbf{x}_{22}(2)\}:\qquad R(2) &\leq& \min \{r_2,r_3\}\\
 \{\mathbf{x}_{11}(4),\mathbf{x}_{21}(4)\}:\qquad R(4)&\leq& r_4
\end{IEEEeqnarray*} with 
\begin{IEEEeqnarray*}{rCl}
r_1&= &\log \left(1+\frac{P-P^{0.75}}{P^{0.75}+\theta_3} \right),\\
r_2&= &\log \left(1+\frac{P^{0.75}-P^{0.5}}{P^{0.5}+\theta_3} \right),\\
r_3&= &\log \left(\tfrac{1}{2}+\frac{P^{0.5}-P^{0.25}}{P^{0.25}} \right)^+,\\
r_4&= &\log \left(1+\frac{P^{0.25}-1}{1} \right).
\end{IEEEeqnarray*} where the minima are necessary, because $\mathbf{x}_{11}(1),\mathbf{x}_{12}(2)$ and $\mathbf{x}_{21}(1),\mathbf{x}_{22}(2)$ need to be decodable at both receivers. Moreover, neither cell can send on level 3, since all interference is aligned at that level. 
The total achievable rate is the summation over all levels
\begin{equation}
R_\Sigma = 2R(1)+2R(2)+2R(4).
\end{equation}

Moreover, we can show that 
\begin{IEEEeqnarray*}{rCl}
r_1 &=& \log \left(1+\frac{P-P^{0.75}}{P^{0.75}+\theta_3} \right)\\
&>& \log \left(\frac{P+\theta_3}{P^{0.75}+\theta_3} \right)\\
&>& \log \left(\frac{P}{P^{0.75}+\theta_3} \right)\\
&>&\log \left(\frac{P}{2P^{0.75}} \right)\\
&=&\log P^{0.25}-1,
\end{IEEEeqnarray*}where the last inequality follows from the fact that $P^{0.75}>\theta_3$ because of $P>1$. Similarly, we can show that 
\begin{IEEEeqnarray*}{rCl}
r_2&> &\log P^{0.25}-1,\\
r_3&> &\log P^{0.25}-2,\\
r_4&= &\log P^{0.25}.
\end{IEEEeqnarray*} 

 The total achievable rate is therefore 
\begin{IEEEeqnarray*}{rCl}
R_\Sigma &=& 2R(1)+2R(2)+2(4)\\
&=&4\log P^{0.25}-8+2\log P^{0.25}\\
&=& 2\log P-\log P^{0.5}-8
\end{IEEEeqnarray*}
With the definition of $\alpha=0.5$, $n_1=\lceil\log P\rceil$ and $n_i=\lceil\log P^{0.5}\rceil$  one can see, that the proposed scheme can achieve the upper bound within a constant gap of $8$ Bits.

\subsubsection{Achievable Rate: The general (weak) interference case}
Henceforth, we define $i\neq j$ for $i,j\in \{1,2\}$ in all equations.
In the general case, $\beta_i$ and $\alpha_i$ can be any value in the defined regime and therefore any number of levels can be needed.
The power splitting is done as in the example where signal power is given by (\ref{theta_k_fractions}).
The choice of codeword decomposition and level usage is dependent on the underlying LDM scheme. We use the power partitions of $\theta_{il}$ in the same way as the $\Delta_i$ blocks in the LDM. However, instead of filling the $l-$th $\Delta_i$-block with bits, we use the specific power partition interval $\theta_{il}$ to transmit a sub-signal lattice codeword.
As in the previous example, the sub-signal codewords can be decoded providing a rate of \eqref{Decodingbound} and \eqref{Decodingbound2}, depending on the number of aligning signals.
It remains to specify the effective noise per level. The general noise structure is 
\begin{IEEEeqnarray*}{rCl}
N_i(l)&=&1+\sum_{\mbox{used levels}}\theta_{l+1},\\
\label{noise_term}
\end{IEEEeqnarray*}
we the levels of aligned interference are counted twice.
We have to distinguish three different rate-term structures for cell $i$, see Fig.~\ref{rate-term-structures}. The first one is the common signal part $R_{\text{IC}_i}$. Here we need decoding bounds $R_{\text{IC}_i}(l)$ per power-partition, which are also received as interference and therefore need to obey two decoding conditions. One decoding bound stems from the direct path utilizing bound \eqref{Decodingbound}. The other decoding bound stems from the fact that our successive decoding scheme needs to decode the sum of two interfering bit-levels at every odd bit-level and therefore needs to obey \eqref{Decodingbound2}. We show in the Appendix \ref{BoundOnInterferenceRateTerm}, that we can achieve

\begin{figure}
\centering
\includegraphics[scale=1]{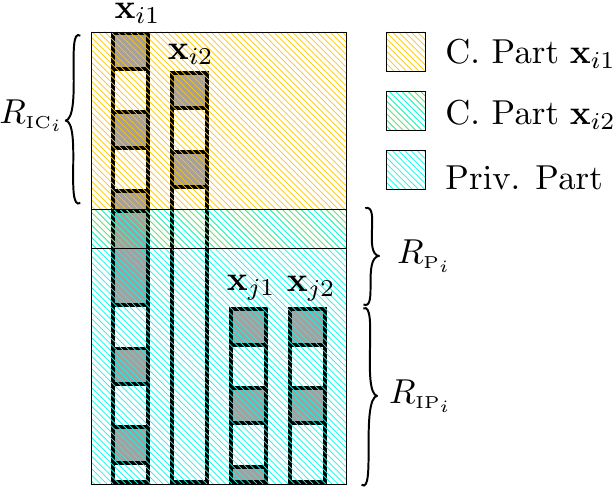}
\caption{Illustration of the three rate term structures. $R_{\text{IC}_i}$ is the alignment structure of the common part. Note that the common part of $\mathbf{x}_{i2}$ is down shifted by exactly one power partition. Moreover, we illustrate the private signal parts, with and without interference.}
\label{rate-term-structures}
\end{figure}

\begin{equation}
\bar{R}_{\text{IC}_i}(l)>\log  \mbox{SNR}_{i1}^{(1-\beta_i)}-2
\end{equation}
which is the minimum of both bounds for the common part per power partition.
The second term is the private part $R_{\text{P}_i}$ which is not interference affected, therefore outside the alignment structures. We need to distinguish the two cases $\lfloor L_i \rfloor=\text{odd}$ and $\lfloor L_i \rfloor=\text{even}$. For $\lfloor L_i \rfloor=\text{odd}$, we have a remainder term with power $\mbox{SNR}_{i1}^{1-\lfloor L_i \rfloor(1-\beta_i)}-\mbox{SNR}_{i1}^{1-(\lfloor L_i \rfloor+1)(1-\beta_i)}$, and the regular private term with power $\mbox{SNR}_{i1}^{1-(\lfloor L_i \rfloor+1)(1-\beta_i)}-\mbox{SNR}_{i1}^{1-\alpha_i(1-\beta_i)}$. For $\lfloor L_i \rfloor=\text{even}$, we have a remainder term of power $\mbox{SNR}_{i1}^{1-\lfloor L_i \rfloor(1-\beta_i)}-\mbox{SNR}_{i1}^{1- L_i (1-\beta_i)}$ and a private term of power $\mbox{SNR}_{i1}^{1-L_i (1-\beta_i)}-\mbox{SNR}_{i1}^{1-\alpha_i(1-\beta_i)}$, which is depicted in Fig.~\ref{rate-term-structures}. Note that for $\lfloor L_i \rfloor=L_i=\text{even}$, the remainder term is zero.
Let us choose the case $\lfloor L_i \rfloor=\text{odd}$ as an example. Here we have the remainder term and the private part term. The remainder term has the same analysis as the terms of $R_{\text{IC}_i}$ in Appendix \ref{BoundOnInterferenceRateTerm}. Considering the new power partition, we get a rate of
\begin{equation}
R_{\text{P1}_i}>\log \mbox{SNR}_{i1}^{1-\lfloor L_i \rfloor(1-\beta_i)}-\log \mbox{SNR}_{i1}^{1-(\lfloor L_i \rfloor+1)(1-\beta_i)} -2.
\end{equation}
For the private part, we only have to consider the direct rate, and do not need the minimisation over both decoding bounds which results in a smaller bit-gap
\begin{equation}
R_{\text{P2}_i}>\log \mbox{SNR}_{i1}^{1-(\lfloor L_i \rfloor+1)(1-\beta_i)}-\log \mbox{SNR}_{i1}^{\alpha_i} -1.
\end{equation}
The total rate of the private part is therefore 
\begin{equation}
R_{\text{P}_i}>\log \mbox{SNR}_{i1}^{1-\lfloor L_i \rfloor(1-\beta_i)}-\log \mbox{SNR}_{i1}^{\alpha_i} -3.
\end{equation}
Note that this is also achievable for the cases with $\lfloor L_i \rfloor=\text{odd}$. The only difference is the location of the power split. Note that in our weak interference regime \eqref{weak-int-def}, we have that $1-\lfloor L_i \rfloor(1-\beta_i)\geq \alpha_i$, since $1-L_i (1-\beta_i)\geq \alpha_i$ which can be shown in the following way. Plugging in the definition of $ L_i $, we get
\begin{IEEEeqnarray*}{rCl}
1-\frac{\log\mbox{SNR}_{j1}^{\alpha_j}}{\log \mbox{SNR}_{i1}^{(1-\beta_i)}}(1-\beta_i) &\geq & \alpha_i\\
1-\frac{\log\mbox{SNR}_{j1}^{\alpha_j}}{\log \mbox{SNR}_{i1}} &\geq & \alpha_i\\
\log \mbox{SNR}_{i1}-\log\mbox{SNR}_{j1}^{\alpha_j}&\geq & \alpha_i \log \mbox{SNR}_{i1}\\
\log \mbox{SNR}_{i1}&\geq & \log \mbox{SNR}_{i1}^{\alpha_i}+\log\mbox{SNR}_{j1}^{\alpha_j}\\
\log \mbox{SNR}_{i1}&\geq & \log \mbox{INR}_{2}^1+\log\mbox{SNR}_{1}^{2}
\end{IEEEeqnarray*}
which is true, since  $\min\{ \mbox{SNR}_{12},\mbox{SNR}_{22}\}\leq \mbox{SNR}_{i1}$.

The third rate term structure is the private rate part, which is affected by the interference. 
For every power partition, we can achieve a rate
\begin{IEEEeqnarray*}{rCl}
\bar{R}_{\bar{\text{IP}}_i}(l)
&>& \log  \mbox{SNR}_{j1}^{(1-\beta_j)}-1,
\end{IEEEeqnarray*}
 which is shown in the Appendix \ref{BoundOnInterferenceRateTerm}.

Furthermore, if $L_j \neq \lfloor L_j \rfloor=\text{even}$, the private part rate $R_{\bar{\text{IP}}_i}$ has a remainder term allocated at the lowest power level. The lowest level has a noise term of 1 and a power of $\mbox{SNR}_{i1}^{\alpha_i}\mbox{SNR}_{j1}^{-\lfloor L_j \rfloor(1-\beta_j)}-1$. We only need to decode the direct rate term and therefore get a decoding bound of
\begin{equation}
R_{\text{IP}_{i,{\text{rem}}}}\leq \log \mbox{SNR}_{i1}^{\alpha_i}\mbox{SNR}_{j1}^{-\lfloor L_j \rfloor(1-\beta_j)}.
\end{equation}

The total achievable sum rate is the summation over all three rate term structures $\bar{R}_{\text{IC}_i}$, $\bar{R}_{\text{IP}_i}$, and $R_{\text{P}_i}$ including the remainder parts. Moreover, we need to sum over all individual partitions. This means that 
\begin{equation}
\bar{R}_{\text{IC}_i}>\sum\limits_l^{\lfloor L_i \rfloor}\mbox{SNR}_{i1}^{(1-\beta_i)}-2,
\end{equation}
and 
\begin{equation}
\bar{R}_{\text{IP}_i}>\sum\limits_{\substack{l=1\\l  \text{ odd}}}^{\lfloor L_j \rfloor}\log \mbox{SNR}_{j1}^{(1-\beta_j)}-1.
\end{equation}
We show the proof exemplary for the case that $L_i=\lfloor L_i \rfloor=even$. Here we have an achievable sum-rate of:
\begin{IEEEeqnarray*}{rCl}
R_{\Sigma}&=& \sum\limits_i^{2}  \bar{R}_{\text{IC}_i}+\bar{R}_{\text{P}_i}+\bar{R}_{\text{IP}_i}\\
&>& \sum\limits_l^{\lfloor L_2 \rfloor} (\log  \mbox{SNR}_{21}^{(1-\beta_2)}-2)+\sum\limits_{\substack{l=1\\l  \text{ even}}}^{\lfloor L_2 \rfloor}(\log \mbox{SNR}_{21}^{(1-\beta_2)}-1)\\
&&+\:\sum\limits_l^{\lfloor L_1 \rfloor} (\log  \mbox{SNR}_{11}^{(1-\beta_1)}-2)+\sum\limits_{\substack{l=1\\l  \text{ even}}}^{\lfloor L_1 \rfloor}(\log \mbox{SNR}_{11}^{(1-\beta_1)}-1)\\
&&+\:\log \mbox{SNR}_{11}^{1-\lfloor L_1 \rfloor(1-\beta_1)}-\log \mbox{SNR}_{11}^{\alpha_1} -3\\
&&+\:\log \mbox{SNR}_{21}^{1-\lfloor L_2 \rfloor(1-\beta_2)}-\log \mbox{SNR}_{21}^{\alpha_2} -3\\
&=&\lfloor L_2 \rfloor\log   \mbox{SNR}_{21}^{(1-\beta_2)}+\frac{\lfloor L_2 \rfloor}{2}\log \mbox{SNR}_{21}^{(1-\beta_2)}\\
&&+\:\lfloor L_1 \rfloor\log  \mbox{SNR}_{11}^{(1-\beta_1)}+\frac{\lfloor L_1 \rfloor}{2}\log \mbox{SNR}_{11}^{(1-\beta_1)}\\
&&-\: 2.5(\lfloor L_1 \rfloor+\lfloor L_2 \rfloor)\\
&&+\:\log \mbox{SNR}_{11}^{1-\lfloor L_1 \rfloor(1-\beta_1)}-\log \mbox{SNR}_{11}^{\alpha_1} -3\\
&&+\:\log \mbox{SNR}_{21}^{1-\lfloor L_2 \rfloor(1-\beta_2)}-\log \mbox{SNR}_{21}^{\alpha_2} -3\\
&=&\log \mbox{SNR}_{11}+\frac{\lfloor L_2 \rfloor}{2}\log \mbox{SNR}_{21}^{(1-\beta_2)}\\
&&+\log \mbox{SNR}_{21}+\frac{\lfloor L_1 \rfloor}{2}\log \mbox{SNR}_{11}^{(1-\beta_1)}\\
&&-\:\log \mbox{SNR}_{11}^{\alpha_1}-\log \mbox{SNR}_{21}^{\alpha_2} -6-2.5(\lfloor L_2 \rfloor+ \lfloor L_1 \rfloor)\\
&=&\log \mbox{SNR}_{11}^{\beta_1} -\log \mbox{SNR}_{11}^{\alpha_1}+\log \mbox{SNR}_{21}^{\beta_2} -\log \mbox{SNR}_{21}^{\alpha_2}\\
&&+\:\phi(\log \mbox{SNR}_{21}^{\alpha_2},\log \mbox{SNR}_{11}^{(1-\beta_1)})\\
&&+\:\phi(\log \mbox{SNR}_{11}^{\alpha_1},\log \mbox{SNR}_{21}^{(1-\beta_2)})\\
&&-\:6-2.5(\lfloor L_2 \rfloor+ \lfloor L_1 \rfloor),
\end{IEEEeqnarray*}
where the last step follows from the definition of $\phi(p,q)$ in \eqref{phi}. We have that $\phi(p,q)=q+\frac{l(p,q)q}{2}$ for $l(p,q)=\text{even}$. Plugging in $p=\log \mbox{SNR}_{i1}^{\alpha_i}$ and $q=\log \mbox{SNR}_{j1}^{(1-\beta_j)}$, shows that $l(\log \mbox{SNR}_{i1}^{\alpha_i},\log \mbox{SNR}_{j1}^{(1-\beta_j)})=\lfloor L_j \rfloor$ and the result follows.$\hfill\IEEEQEDclosed$

One can see that the achievable rate of the Gaussian channel is within a constant-gap of $6+2.5(\lfloor L_2 \rfloor+ \lfloor L_1 \rfloor)$bits of the LDM rate using the correspondence $n_{ik}^j=\lceil\log |h_{ik}^j|^2 P\rceil$.

\section{Analysis of the LTD-IMAC}
In this section we prove the two theorems, Theorem~\ref{weak_constant_gap} and Theorem~\ref{ltd-imac-thm}. In particular, we will first show the achievability in Section~\ref{LTD-achiev.} and then in Section~\ref{UB_LTDM} the upper bound, for the theorems.
\label{Proof for LTD-IMAC}
\subsection{Achievable Schemes}

\label{LTD-achiev.}
We start with two lemmas, which provide the bases for the achievable rate in the LTD model in the general symmetric setting and in the weak interference general setting, respectively. Afterwards we discuss two lemmas which are necessary to prove that the codewords in the schemes can be decoded. 
\begin{lemma}[{\it general interference, symmetry}]
For every $\delta\in (0,1]$ and $n_1$, $n_2$, $n_i \in \mathbb{N}$ such that $n_1 \geq n_2$ there exists a set $B\subset (1,2]^{2 \times 3}$ of Lebesgue measure $\mu(B)\leq \delta$ such that for all channel gains $g_{ik}^j \in (1,2]^{2 \times 3} \setminus B$ the achievable sum rate for the IMAC system model is:
\begin{equation}
R_\Sigma \geq \min\{R_{\text{ach},1},R_{\text{ach},2},R_{\text{ach},3},R_{\text{ach},4},R_{\text{ach},5}\}-2 \log (c/\delta)
\end{equation}
with 
\begin{IEEEeqnarray*}{rCl}
R_{\text{ach},1}&:=& 2\max((n_1-n_i)^+,n_i)\IEEEyessubnumber\\
&&+\:\min((n_1-n_i)^+,n_i),\\
R_{\text{ach},2}&:=& \tfrac{2}{3}(2\max(n_1,n_i)+(n_1-n_i)^+),\IEEEyessubnumber\\
R_{\text{ach},3}&:=& 2n_1,\IEEEyessubnumber\\
R_{\text{ach},4}&:=& \max(2n_2,2(n_1-n_i)^+,2n_i),\IEEEyessubnumber\\
R_{\text{ach},5}&:=& \max(n_1,n_i)+\max(n_2,(n_1-n_i)^+).\IEEEyessubnumber
\end{IEEEeqnarray*}

\label{ltd-upBound1}
\end{lemma}

\begin{lemma}[{\it weak interference}]
For every $\delta\in (0,1]$ and $n_1$, $n_2$, $n_i \in \mathbb{N}$ such that $n_1 \geq n_2$ there exists a set $B\subset (1,2]^{2 \times 3}$ of Lebesgue measure $\mu(B)\leq \delta$ such that for all channel gains $g_{ik}^j \in (1,2]^{2 \times 3} \setminus B$ the achievable sum rate for the IMAC system model is:
\begin{equation}
R_\Sigma \geq \min\{R_{\text{ach},1},R_{\text{ach},2},R_{\text{ach},3},R_{\text{ach},4}\}-2 \log (c/\delta)
\end{equation}
with 
\begin{IEEEeqnarray*}{rCl}
R_{\text{ach},1}&:=&\max\{(n_{11}^1-n_1^2),n_{12}^1\}\IEEEyessubnumber\label{weak-ltd-a}\\
&&+\:\max\{(n_{21}^2-n_2^1),n_{22}^2\}\\
R_{\text{ach},2}&:=&\max\{(n_{11}^1-n_1^2),n_{12}^1\}+n_{21}^2-\tfrac{1}{2}n_2^1\IEEEyessubnumber\label{weak-ltd-b}\\
R_{\text{ach},3}&:=&n_{11}^1-\tfrac{1}{2}n_1^2+\max\{(n_{21}^2-n_2^1),n_{22}^2\}\IEEEyessubnumber\label{weak-ltd-c}\\
R_{\text{ach},4}&:=&n_{11}^1-\tfrac{1}{2}n_1^2+n_{21}^2-\tfrac{1}{2}n_2^1\IEEEyessubnumber\label{weak-ltd-d},
\end{IEEEeqnarray*}
\label{ltd-upBound2}
\end{lemma}


\begin{proof} To prove the lemmas we first need to show under which conditions a rate allocation scheme yields linear independence of the used $\mathbf{H}$ columns and therefore allows successful decoding. For the interference range $\alpha \geq \tfrac{1}{2}$, we can use the lemma 11 of \cite{Niesen-Ali}, with a re-labelling such that it fits our case. The following lemma shows the modified version.

\begin{lemma}[{\cite[Lemma~11 (modified)]{Niesen-Ali}}]
\label{general-ach-ltd-lem}
Let $\delta\in (0,1]$ and $n_1$, $n_2$, $n_i \in \mathbb{N}$ such that $n_1\geq n_2 \geq n_i $, and $2n_i \geq n_1$,  and let $ R^c_{k1}$, $R^{p1}_{k1}$, $R^{p2}_{k1}$, $R^c_{k2} \in \mathbb{N}$ with $k,l\in\{1,2\}$ and $l\neq k$  satisfy,
\begin{IEEEeqnarray*}{rCl}
R^c_{k1}+R^c_{k2} +\max\{R^c_{l1},R^c_{l2}\} + R^{p1}_{k1}&\leq &  n_{1}-\log\left(\tfrac{32}{\delta}\right)\\
R^c_{k2} +\max\{R^c_{l1},R^c_{l2}\} + R^{p1}_{k1}&\leq&  n_{2}-\log\left(\tfrac{32}{\delta}\right)\\
\max\{R^c_{l1},R^c_{l2}\}+R^{p1}_{k1}&\leq&  n_i
\end{IEEEeqnarray*}
Then, the bit allocation [...] for the (modulated) deterministic X-channel allows successful decoding at both receivers for all channel gains ($g_{ik}^j\in(1,2]^{2\times 3}$) except for a set $B\subset (1,2]^{2 \times 3}$ of Lebesgue measure $\mu(B)\leq \delta$.

\end{lemma}

For the range of $\alpha < \tfrac{1}{2}$, our achievable scheme gets a second private signal part between the common and the private signal of the stronger user (see for example Fig.~\ref{rateallo}, $R_{11}^{p2}$). Therefore, Lemma~\ref{general-ach-ltd-lem} is not applicable anymore. However, due to the special structure of the achievable scheme, we can modify the proof to show a similar result. 
\begin{figure}
\centering
\includegraphics[scale=.7]{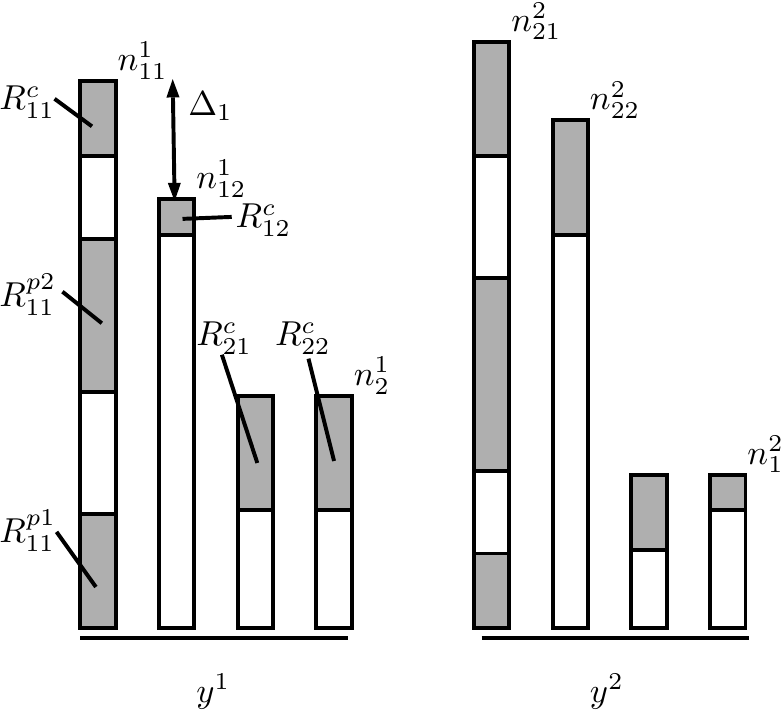}
\caption{Illustration of the rate allocations and the design rules for the achievable scheme in the range $\alpha < \tfrac{1}{2}$. At $y^1$, we have an example for the case I.2: $\tfrac{1}{2}n_i < \Delta < n_i$ and at $y^2$ we have an example for the case I.3: $\Delta \leq \tfrac{1}{2}n_i$. }
\label{rateallo}
\end{figure} 

 
\begin{lemma}
\label{ltd-weak-decoding-lemma}
Let $\delta\in (0,1]$ and $n^k_{k1}$, $n^k_{k2}$, $n^k_l \in \mathbb{N}$ such that $n^k_{k1}\geq n^k_{k2} > n^k_l$ and $n_2^1+n_1^2 < \min\{n_{11}^1,n_{21}^2\}$,  and let $ R^c_{k1}$, $R^{p1}_{k1}$, $R^{p2}_{k1}$, $R^c_{k2}$, $R^c_{21}  \in \mathbb{N}$, with $k,l \in \{1,2\}, k\neq l$ satisfy,

\begin{IEEEeqnarray*}{rCl}
R^c_{k1}+R^{p1}_{k1}+R^{p2}_{k1}+R^c_{k2}+\max\{R^c_{l2},R^c_{l1}\}&\leq&  n^k_{k1}-\log\left(\tfrac{8}{\delta}\right)\\
R^{p1}_{k1}+R^{p2}_{k1}+R^c_{k2}+\max\{R^c_{l1},R^c_{l2}\}&\leq& n^k_{k2}\\
R^{p2}_{k1}+R^{p1}_{k1}+\max\{R^c_{l1},R^c_{l2}\}&\leq&  (n^k_{k1}-n^l_k)\\
R^{p1}_{k1} + \max\{R^c_{l1},R^c_{l2}\} &\leq&  n^k_l
\end{IEEEeqnarray*} 

Then, a bit allocation, chosen such that the conditions are satisfied, allows successful decoding at both receivers for all channel gains except for an outage set $B\subset (1,2]^{2 \times 3}$ of Lebesgue measure $\mu(B)\leq \delta$.

\end{lemma}
 \begin{IEEEproof}
 \label{groshev-det-weak}
The proof for Lemma \ref{general-ach-ltd-lem} can be found in \cite{Niesen-Ali} and the proof for Lemma \ref{ltd-weak-decoding-lemma} is in the same fashion but exploits the non-overlapping coding structure between $R_{k1}^{p1}, \max\{R^c_{l1},R^c_{l2}\}$ and $R_{11}^{p2}$ in the weak interference case, see Appendix~\ref{Proof of Lemma groshev-det-weak}. A similar method is used in Section \ref{Proof G-IMAC}, where the Gaussian equivalent of this modification is used.
 \end{IEEEproof}
 
The last two lemmas tell us that for all rates which obey the stated conditions, the spanned subspaces are independent except for a small set of measure $\mu(B)\leq \delta$. This means that there exists a unique solution and the signals can be decoded. Hence, any proposed scheme needs to be checked if it obeys these conditions and therefore allows for successful decoding. The conditions of the following schemes are checked in the appendix. It then remains to show that the proposed schemes achieve the rates in the theorem.

We have to choose different schemes for the cases
\begin{IEEEeqnarray*}{rCl}
\text{I}&:&\ \alpha \in [0,\tfrac{1}{2}],\ \text{II}:\ \alpha \in (\tfrac{1}{2},\tfrac{3}{5}),\ \text{III}:\ \alpha \in [\tfrac{3}{5},1]\\
\text{IV}&:&\ \alpha \in (1,\tfrac{3}{2}],\ \text{V}:\ \alpha \in (\tfrac{3}{2},\infty).
\end{IEEEeqnarray*}

Common and private signal parts are indicated with the superscript $c$ and $p$ respectively, such that $\mathbf{\bar{x}}=[\mathbf{\bar{x}}^c;\mathbf{\bar{x}}^p]$. The private parts of the signal can be used to communicate solely to the intended receiver, without affecting the other cell. We dedicate the $R_{ik}^c$ most significant bits, and the $R_{i1}^{p_1}$ least significant bits of $\mathbf{\bar{x}}_{i1}$ to carry information. For the weak interference cases $\alpha<\tfrac{1}{2}$, the private part of $\mathbf{\bar{x}}$ has another bit-level allocation. There we dedicate the $R_{i1}^{p_2}$ most significant bits of the private part $\mathbf{\bar{x}}^p$ to carry information, see Figure~\ref{rateallo}. We now have to choose the scheme, and therefore how many bit-levels we give to each of the allocation rates, which we just introduced. We start with the weak interference case.

{\bf Case I: ($0 \leq \alpha \leq \tfrac{1}{2}$):} 
We use the fact, that the model can be split into two sub-models, similar to the LD-Model (see Figure~\ref{system_model_split}).
This means that scheme differences for the bit-levels above the interference-level of one cell, do not influence the other cell. Hence, we can consider them separately and without loss of generality restrict the case analysis to symmetric cases. However, this is only possible for the weak interference model, which is why we consider the symmetrical model for higher interference regimes.
For $k,l\in\{1,2\}$ and $k\neq l$ we set 
\begin{IEEEeqnarray*}{rCl}
R^c_{k1}&:=&\left\lfloor\tfrac{1}{2}n_k^l\right\rfloor, R^c_{k2}:=\min\{\left\lceil\tfrac{1}{2}n_k^l\right\rceil,(n_{k2}^k-(n_{k1}^k-n_k^l))^+\}\\
R^{p1}_{k1}&:=&\left\lfloor\tfrac{1}{2}n_l^k\right\rfloor, R^{p2}_{k1}:=n_{k1}^k-n_k^l-n_l^k.
\end{IEEEeqnarray*}

Multiuser gain in the IMAC model is dependent of the strength of the second user in each cell. This is also true in the weak interference case, where we differentiate between three sub-cases (of nine in total), dependent on $n_{k2}^k$.

Case I.1:\ 
The first sub-case is when $n_{k2}^k \leq n_{k1}^k-n_k^l$. In this case, the second user is useless for the channel, since the private rate $R_{i1}^{p_2}$ of user 1 can serve the same purpose. There is no multi-user gain in this regime and the achievable rate is limited to the IC sum rate. 

\begin{IEEEeqnarray*}{rCl}
R_\Sigma &=&R^c_{11}+R^{p1}_{11}+R^{p2}_{11}+R^c_{12}+R^c_{21}+R^{p1}_{21}+R^{p2}_{21}+R^c_{22}\\
&=& 2\left\lfloor\tfrac{1}{2}n_1^2\right\rfloor +2\left\lfloor\tfrac{1}{2}n_2^1\right\rfloor +n_{11}^1+n_{21}^2-2n_1^2-2n_2^1\\
&\geq & (n_{11}^1-n_1^2)+(n_{21}^2-n_2^1)-4.
\end{IEEEeqnarray*}
We remark that the achievable rate for one cell and this scheme is $(n_{k1}^k-\tfrac{1}{2}n_k^l-\tfrac{1}{2}n_l^k)$. This might be against the intuition from the IC model, where it is $((n_{k1}^k-n_l^k)$ resulting from a scheme similar to treating interference as noise. However, considering the sum rate of both cells in our scheme, we get back to the IC-rate. The following cases, where the second user of one or both cells is stronger, results in an {\it additional} rate part on top of the IC sum-rate. In the symmetrical case, we see that the conditions lead to a minimum of $R_{\text{ach},4}$ in Lemma \ref{ltd-upBound1}, which can be reached by the allocation.

Case I.2:\ 
The second sub-case is when $n_{k1}^k-n_k^l < n_{k2}^k \leq n_{k1}^k-\tfrac{1}{2}n_k^l$. In this range, the upper part of the second user becomes useful as the top bit-level reaches above the part $R_{i1}^{p_2}$, and makes bit-level alignment possible. This part has a rising multiuser gain and is no longer limited to the IC sum rate. 
\begin{IEEEeqnarray*}{rCl}
R_\Sigma &=&R^c_{11}+R^{p1}_{11}+R^{p2}_{11}+R^c_{12}+R^c_{21}+R^{p1}_{21}+R^{p2}_{21}+R^c_{22}\\
&=& 2\left\lfloor\tfrac{1}{2}n_1^2\right\rfloor +2\left\lfloor\tfrac{1}{2}n_2^1\right\rfloor +n_{11}^1+n_{21}^2-2n_1^2-2n_2^1\\
&&+\:(n_{12}^1-(n_{11}^1-n_1^2))+(n_{22}^2-(n_{21}^2-n_2^1))\\
&\geq & n_{22}^2+n_{12}^1-4
\end{IEEEeqnarray*}
The $n_{k2}^k-(n_{k1}^k-n_k^l)$ terms represent the additional rate of the second user and thus the multiuser gain. We achieve the active bound $R_{ach,4}$ in the symmetrical setting.

Case I.3:\  
The last sub-case is for $n_{k1}^k-\tfrac{1}{2}n_k^l < n_{k2}^k \leq n_{k1}^k$, here the second user can be fully utilized to provide one half of the interference strength at the opposite cell. The multiuser gain becomes half of the interference and the achievable sum rate reaches the main upper bound. 
\begin{IEEEeqnarray*}{rCl}
R_\Sigma &=&R^c_{11}+R^{p1}_{11}+R^{p2}_{11}+R^c_{12}+R^c_{21}+R^{p1}_{21}+R^{p2}_{21}+R^c_{22}\\
&=& 2\left\lfloor\tfrac{1}{2}n_1^2\right\rfloor +2\left\lfloor\tfrac{1}{2}n_2^1\right\rfloor +n_{11}^1+n_{21}^2-2n_1^2-2n_2^1\\
&&+\: \tfrac{1}{2}(\left\lceil n_1^2 \right\rceil + \left\lceil n_2^1\right\rceil )\\
&=& n_{11}^1+n_{21}^2-\left\lceil\tfrac{1}{2}n_1^2\right\rceil-\left\lceil\tfrac{1}{2}n_2^1\right\rceil\\
&\geq & (n_{11}^1-\tfrac{1}{2}n_1^2)+(n_{21}^2-\tfrac{1}{2}n_2^1)-2.
\end{IEEEeqnarray*}
Here we see that the multiuser gain is $\tfrac{1}{2}n_k^l$ for cell $k$. Intuitively this means, that we can use the aligned part as additional rate. 

The combination of the cell rates (I.1:) $n_{11}^1-n_1^2$, (I.2:) $n_{k2}^k$, (I.3:) $n_{11}^1-\tfrac{1}{2}n_1^2$ leads to nine different cases. If a cell satisfies condition (I.1:) $n_{k2}^k \leq n_{k1}^k-n_k^l$, the $n_{k1}^k-n_k^l$ terms in the $\max$ of \eqref{weak-ltd-a}-\eqref{weak-ltd-c} are active. Furthermore the $\max$ terms are smaller than $n_{k1}^k-n_k^l$ and the bound is reached. For condition (I.2:) $n_{k1}^k-n_k^l < n_{k2}^k \leq n_{k1}^k-\tfrac{1}{2}n_k^l$ yields activation of $n_{k2}^k$ inside the $\max$ terms, and on the other hand, $n_{k2}^k$ is weaker than $n_{k1}^k-\tfrac{1}{2}n_k^l$ proving the remaining bounds. The decoding conditions of lemma~\ref{ltd-weak-decoding-lemma} are checked in the appendix. Considering a bit-gap of at most $2\log(16/\delta)$ from the decoding conditions and at most $4$-bit from the use of fractional terms, results in a total gap of at most $2\log(c/\delta)$ with $c=64$. This proves lemma~\ref{ltd-upBound2}. $\hfill\IEEEQEDclosed$

We continue with the proof for lemma~\ref{ltd-upBound1}. From now on we confine the analysis to the symmetrical case, meaning that $n^1_{11}=n^2_{21}:=n_1$, $n^1_{12}=n^2_{22}:=n_2$, $n_1^2=n_2^1=n_i$. Sum rates for the weak interference case of the symmetric model are already shown as part of lemma~\ref{ltd-upBound2}, we therefore go on with the remaining regimes.

{\bf Case II ($\tfrac{1}{2} <\alpha < \tfrac{3}{5}$):}
\begin{figure}
\centering
\includegraphics[scale=1.3]{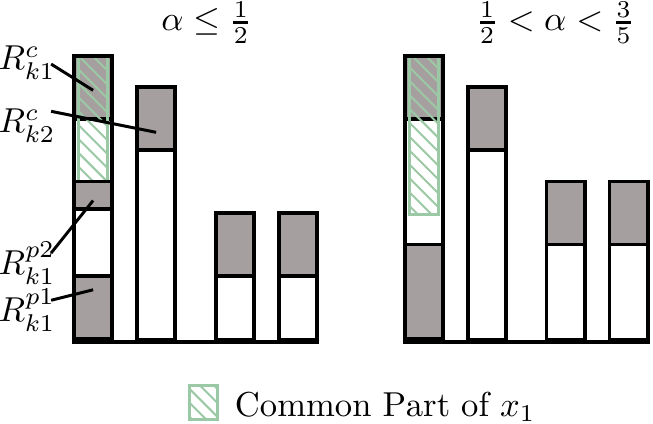}
\caption{Illustration of the rate allocations and the design rules for the achievable scheme in the range $\tfrac{1}{2} <\alpha < \tfrac{3}{5}$ for an exemplary cell with $\alpha=\tfrac{4}{7}$. With a comparison to the scheme of the previous regime to illustrate the problem with the decoding bound on $n_1$.}
\label{rateallo_3_5}
\end{figure}
For $k\in\{1,2\}$ we set
\begin{IEEEeqnarray*}{rCl}
R^c_{k1}&:=&\left\lfloor\tfrac{1}{2}(n_1-n_i)\right\rfloor,\ R^c_{k2}:=\min\{\left\lfloor\tfrac{1}{2}(n_1-n_i)\right\rfloor,(n_2-n_i)\}\\
R^{p1}_{k1}&:=&n_i-\left\lfloor\tfrac{1}{2}(n_1-n_i)\right\rfloor.
\end{IEEEeqnarray*}
One can see the reason for this rate allocation scheme intuitively considering the interference at both cells. First of all, at $\alpha > \tfrac{1}{2}$, the interference is strong enough to ''reach'' into the common part of the signal $\mathbf{\bar{x}}_{k1}$ in cell $k$ (see Fig.~\ref{rateallo_3_5}). Therefore, the additional private part bit allocation $R_{k1}^{p2}$ of the weak interference case is zero and we can use the decoding lemma~\ref{general-ach-ltd-lem} from now on. Note that lemma~\ref{general-ach-ltd-lem}, and therefore a bigger gap, is needed because of overlapping signal parts. Moreover, as a result of this new regime, the signal $\mathbf{\bar{x}}_{k2}$ in each cell $k$ can support multiuser gain as long as $n_2>n_i$. The previous rate allocation scheme ($\left\lfloor\tfrac{1}{2}n_i\right\rfloor$ for common parts) fails to satisfy the first equation of lemma~\ref{general-ach-ltd-lem}, since there are less than $n_i$ interference unaffected bit-levels for both users in a cell. The new allocation needs to fit into $(n_1-n_i)$ bit-levels and an obvious choice is $\tfrac{1}{2}(n_1-n_i)$. This way we can use the whole bit-levels available and at the same time maximize the alignment at the other cell. The private part allocates the available rest, satisfying the third equation of lemma~\ref{general-ach-ltd-lem}. See Figure \ref{rateallo_3_5} for an illustration. As in the weak interference case, we need to differentiate between sub-cases, depending on the strength of the signals $\mathbf{\bar{x}}_{k2}$.

Case II.1:\
The first sub-case is when $n_2\geq n_i+\left\lfloor\tfrac{1}{2}(n_1-n_i)\right\rfloor$.
Therefore $\left\lfloor\tfrac{1}{2}(n_1-n_i)\right\rfloor\leq (n_2-n_i)$ and $R^c_{k2}:=\left\lfloor\tfrac{1}{2}(n_1-n_i)\right\rfloor$. The second direct signal has enough bit-levels to support the full multiuser gain. The sum-rate becomes
\begin{IEEEeqnarray*}{rCl}
R_\Sigma= 4\left\lfloor\tfrac{1}{2}(n_1-n_i)\right\rfloor+2(n_i-\left\lfloor\tfrac{1}{2}(n_1-n_i)\right\rfloor)\geq n_1+n_i-2.
\end{IEEEeqnarray*}
Regarding Lemma~\ref{ltd-upBound1}, the sub-case implies that $n_2\geq n_i$ and it therefore follows from $n_i>(n_1-n_i)$, that $n_2>(n_1-n_i)$. Moreover, one can see from $n_2\geq n_i+\left\lfloor\tfrac{1}{2}(n_1-n_i)\right\rfloor$ and $\alpha < \tfrac{3}{5}$, that $R_{\text{ach},1}$ gets activated and is the minimum of all bounds.

Case II.2:\ On the contrary, for $n_2< n_i+\left\lfloor\tfrac{1}{2}(n_1-n_i)\right\rfloor$, $R^c_{k2}:=(n_2-n_i)$ because the second direct signal cannot support the full previous rate allocation without violating the second decoding condition of lemma \ref{general-ach-ltd-lem}. The sum-rate is 
\begin{IEEEeqnarray*}{rCl}
R_\Sigma=2(n_2-n_i)+2n_i=2n_2.
\end{IEEEeqnarray*}
And for $n_2<n_i$ the IMAC collapses to the IC-model and yields the known $2n_i$ as sum rate. Note that $R_{ach,4}$ gets activated, and is also achieved, in the last two cases.

{\bf Case III ($\tfrac{3}{5} \leq \alpha \leq 1$):} 
\begin{figure}
\centering
\includegraphics[scale=1]{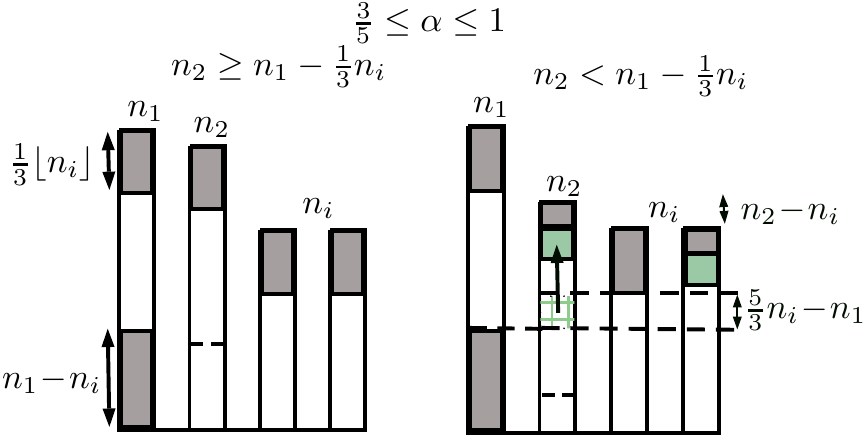}
\caption{Illustration of the rate allocations for the achievable scheme in the range $\tfrac{3}{5} \leq \alpha \leq 1$ for an exemplary cell with $\alpha=\tfrac{2}{3}$. We show the two achievable schemes for $n_2\geq n_1-\tfrac{1}{3}n_i$ and $n_2< n_1-\tfrac{1}{3}n_i$.}
\label{rateallo_3_1}
\end{figure} 
For case 3, we have the special situation, that the $\alpha$-range needs to be subdivided to account for different sub-ranges. Still, due to the full multi-user gain sum-rate being the same over the whole range $\tfrac{3}{5} \leq \alpha \leq 1$ justifies analysing those as part of one case, see fig.~\ref{Case3_Illustration}.

\begin{figure}
\centering
\includegraphics[scale=.38]{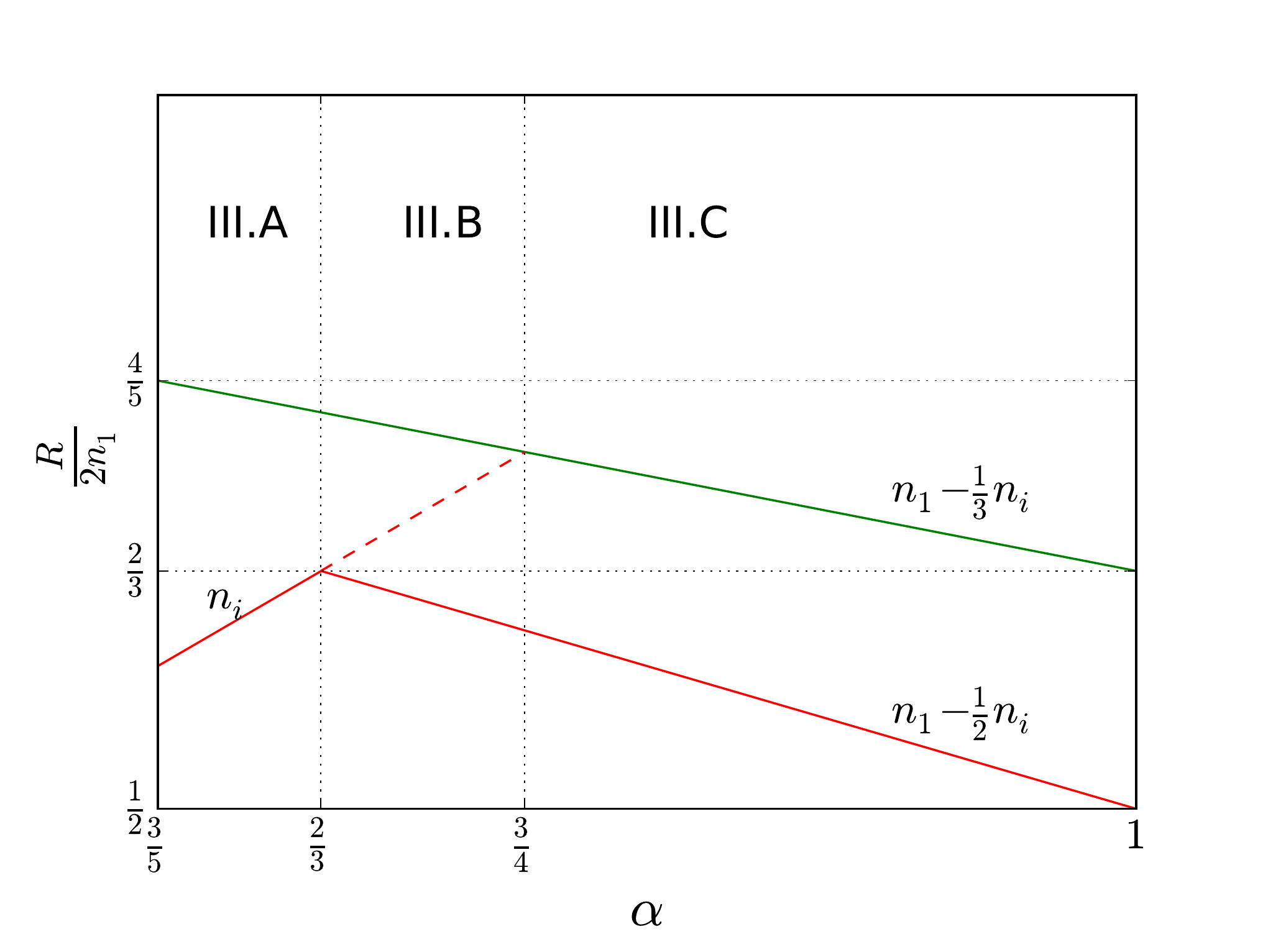}
\caption{Illustration of the full multi-user gain sum-rate (green) and the IC sum-rate (red) in the three $\alpha$ sub-ranges of case III. For varying strength of the weaker users signal $\mathbf{\bar{x}}_{k2}$ in each cell $k$, the achieved sum-rate can lie in between both curves. 
}
\label{Case3_Illustration}
\end{figure} 

{\bf Case III.A ($\tfrac{3}{5} \leq \alpha \leq \tfrac{2}{3}$):} 
For $k\in\{1,2\}$ we set 
\begin{IEEEeqnarray*}{rCl}
R^c_{k1}&:=&\left\lfloor\tfrac{1}{3}n_i\right\rfloor,\\
R^c_{k2}&:=&\min\{\left\lfloor\tfrac{1}{3}n_i\right\rfloor,\left\lfloor((n_2-n_i)^+ +\tfrac{5}{3}n_i-n_1)\right\rfloor\},\\
R^{p1}_{k1}&:=&n_1-n_i.
\end{IEEEeqnarray*}
The change of the previous rate allocation is necessary because the common part drops in the range where we allocated the private part in the previous case. Therefore, the allocation of the private part needs to be changed accordingly. Due to this change, we have $n_i$ bit-levels to allocate two common part allocations plus the aligned interference of the other cell. Due to the symmetry, an obvious choice is to use $\tfrac{1}{3}n_i$ for every common part allocation to maximise bit-level usage without violating the constraints. This scheme ensures that condition 1 from lemma~\ref{general-ach-ltd-lem} holds. We need to make sure, that also condition 2 and 3 hold. Condition 3 holds independently of $n_2$ due to the scheme design. For condition 2 the user of $\mathbf{\bar{x}}_{k2}$ needs to be strong enough to support the allocation of $\tfrac{1}{3}n_i$ bit-levels and therefore $n_2\geq n_1-\tfrac{1}{3}n_i$, which is the condition for the sub-case III.A.1. For $n_1-\tfrac{1}{3}n_i > n_2 >n_i$ we reach the sub-case III.A.2, where the signal $\mathbf{\bar{x}}_{k2}$ can support an allocation of $(n_2-n_i)^+ +\tfrac{5}{3}n_i-n_1$ bit-levels. For $n_2\leq n_i$ the scheme cannot support multi-user gain any more, the LTD-IMAC falls back to the LTD-IC and one strategy is to leave the weak user of signal $\mathbf{\bar{x}}_{k2}$ silent and use IC techniques. We therefore have two sub-cases, where only for $n_2\geq n_1-\tfrac{1}{3}n_i$ (III.A.1) full multi-user gain can be achieved and at $n_2= n_1-\tfrac{1}{3}n_i$ is a transition to case III.A.2, which still yields multi-user gain as long as $n_2>n_i$. We remark that as long as $n_2$ is larger than the private part $n_1-n_i$, it can be used to achieve the same sum-rate as IC-techniques. For achieving $n_i$ in the range $n_2<n_i$, $R_{k1}^c$ would need a larger allocation, which will be used in case III.B.


{\bf Sum Rate:}
The sum rate for the case III.A.1 is
\begin{IEEEeqnarray*}{rCl}
R_\Sigma= 4\left\lfloor\tfrac{1}{3}n_i\right\rfloor+2(n_1-n_i)\geq 2n_1-\tfrac{2}{3}n_i-4,
\end{IEEEeqnarray*}
where $R_{ach,2}$ is active and for the case III.A.2 $(n_2>n_i)$ it is
\begin{IEEEeqnarray*}{rCl}
R_\Sigma&=&2(\tfrac{1}{3}n_i+(n_1-n_i)+((n_2-n_i)^+ \\
&&+\:\tfrac{5}{3}n_i-n_1)-2\geq 2n_2-2,
\end{IEEEeqnarray*}
where the rate term $R_{ach,4}$ is active.

{\bf Case III.B ($\tfrac{2}{3} < \alpha \leq \tfrac{3}{4}$):} For $k\in\{1,2\}$ we set 
\begin{IEEEeqnarray*}{rCl}
R^c_{k1}&:=&\left\lfloor\tfrac{1}{3}n_i\right\rfloor+ \min \{\left\lfloor (\tfrac{2}{3}n_i-n_2)^+\right\rfloor ,\\
&&\: \left\lfloor (n_1-\tfrac{4}{3}n_i)^++\tfrac{1}{2}(2n_i-n_2-n_1)^+\right\rfloor \},\\
R^c_{k2}&:=&\min\{\left\lfloor\tfrac{1}{3}n_i\right\rfloor,\left\lfloor((n_2-n_i)^+ \right. \\
&&+\:\left.(\tfrac{5}{3}n_i-n_1)^+)\right\rfloor , (n_2-(n_1-n_i))^+\},\\
R^{p1}_{k1}&:=&n_1-n_i.
\end{IEEEeqnarray*}

\begin{figure}
\centering
\includegraphics[scale=1.2]{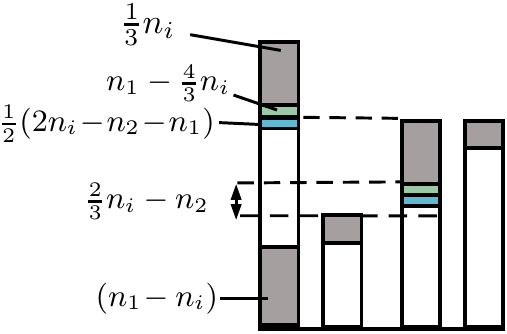}
\caption{Illustration of the scheme for case III.B.4. The goal is to allocate as many as possible bits of $\mathbf{\bar{x}}_{k1}$, to fill the free (allocatable) space $\tfrac{2}{3}n_i-n_2$ between the interference and $R_{k2}^c,R_{k1}^{p1}$. The obvious action is to allocate more bit-levels at $\mathbf{\bar{x}}_{k1}$, however, there are just $n_1-\frac{4}{3}n_i$ freely allocatable  (i.e. which do not overlap with other signals) bits. The remaining $2n_i-n_2-n_1$ bits overlap with the aligned interference and one can allocate only half of the available bit-levels in order to stay within decoding conditions.}
\label{Case3_Illustration3}
\end{figure}

As in case III.A, only for $n_2\geq n_1-\tfrac{1}{3}n_i$ (III.B.1) full multi-user gain can be achieved and at $n_2= n_1-\tfrac{1}{3}n_i$ is a transition to case III.B.2. However, due to the fact that $n_i$ gets bigger than $n_1-\tfrac{1}{2}n_i$ for $\alpha>\tfrac{2}{3}$ (see fig. \ref{Case3_Illustration}), the scheme still yields multi-user gain as long as $n_2> n_1-n_i$. This means that we still have multi-user gain as long as the signal bit vector $\mathbf{\bar{x}}_{k2}$ is larger than the private part of $\mathbf{\bar{x}}_{k1}$. We therefore have two additional sub-cases for $n_1-n_i<n_2<n_i$.
The first sub-case III.B.3 is for $\tfrac{2}{3}n_i\leq n_2<n_i$, where we use the same strategy as in III.B.2, but can only support a free allocation below $n_i$ at $\mathbf{\bar{x}}_{k2}$. This means that the part $(n_2-n_i)^+$ will be zero, resulting in a sum-rate of $2n_i$ instead of $2n_2$. The second sub-case III.B.4 is in the range $n_1-n_i<n_2<\tfrac{2}{3}n_i$. For this sub-case, the signal $\mathbf{\bar{x}}_{k2}$ cannot support the $\tfrac{5}{3}n_i-n_1$ bit-level allocation anymore. The allocation $n_2-(n_1-n_i)$ gets active, which uses a maximum of bit-levels without overlapping with the private part of $\mathbf{\bar{x}}_{k1}$. The low number of bit levels of $\mathbf{\bar{x}}_{k2}$ results in a gap of $(\tfrac{2}{3}n_i-n_2)^+$ bits between the interference signal allocation and the bit allocations of $R_{k2}^c$ and $R_{k1}^{p1}$. One therefore needs to allocate more bit levels at $\mathbf{\bar{x}}_{k1}$ to compensate. However, due to $\tfrac{3}{4}\geq\alpha >\tfrac{2}{3}$ we have that $n_1-n_i< 2n_i-n_1\leq \tfrac{2n_i}{3}$. Therefore, $n_2$ can fall into a range ($n_2<2n_i-n_1$), where the missing bit-levels $(\tfrac{2}{3}n_i-n_2)^+$ are more than the freely allocatable space above $n_i$, which is $n_1-\tfrac{4}{3}n_i$ bit-levels. Therefore, only $n_1-\tfrac{4}{3}n_i$ bit-levels can be allocated without penalty, and another $\tfrac{1}{2}(2n_i-n_2-n_1)^+$ for half of the remaining space (see fig.~\ref{Case3_Illustration3}).
Note that case III.B.3 and II.B.4 is in the range $n_2<n_i$. One therefore needs to switch the signal $\mathbf{\bar{x}}_{k2}$ with the interfering signal in the decoding lemma, which yields the third condition that $R^c_{k2}+R^{p1}_{k1}\leq  n_2$ and in condition 2, $n_2$ gets replaced with $n_i$. 


{\bf Sum Rate:}
The sum rate for the case III.B.1 is
\begin{IEEEeqnarray*}{rCl}
R_\Sigma= 4\left\lfloor\tfrac{1}{3}n_i\right\rfloor+2(n_1-n_i)\geq 2n_1-\tfrac{2}{3}n_i-4,
\end{IEEEeqnarray*}
where $R_{ach,2}$ is active. For the case III.B.2 and $n_2\geq n_i$ the sum rate is
\begin{IEEEeqnarray*}{rCl}
R_\Sigma &=&2(\tfrac{1}{3}n_i+(n_1-n_i)+((n_2-n_i)^+ +\tfrac{5}{3}n_i-n_1)\\
&&-\:2\geq 2n_2-2,
\end{IEEEeqnarray*}
and for $\tfrac{2}{3}n_i \leq n_2< n_i$ (case III.B.3) it is
\begin{IEEEeqnarray*}{rCl}
R_\Sigma=2(\tfrac{1}{3}n_i+(n_1-n_i)+(\tfrac{5}{3}n_i-n_1))-2\geq 2n_i-2,
\end{IEEEeqnarray*}
where the term $R_{ach,4}$ is active in the last two cases.

For case III.B.4 $(n_1-n_i) < n_2 < \tfrac{2}{3}n_i$,

as long as $n_2>2n_i-n_1$ we have that $R_{k1}^c=\left\lfloor\tfrac{1}{3}n_i\right\rfloor+ \left\lfloor (\tfrac{2}{3}n_i-n_2)^+\right\rfloor$, while $R_{k2}^c=(n_2-(n_1-n_i))^+$  and the sum rate is $2n_i-2$. If $n_2\leq 2n_i-n_1 $, the allocation of $R_{k1}^c$ changes and the sum rate becomes
\begin{IEEEeqnarray*}{rCl}
R_\Sigma&=&2(\tfrac{1}{3}n_i+\left\lfloor (n_1-\tfrac{4}{3}n_i)+\tfrac{1}{2}(2n_i-n_2-n_1)\right\rfloor\\
&&+\:(n_1-n_i)+(n_2-(n_1-n_i))^+-2\geq n_2+n_1-2,
\end{IEEEeqnarray*}
where the rate term $R_{ach,5}$ gets active.

{\bf Case III.C ($\tfrac{3}{4} \leq \alpha \leq 1$):} For $k\in\{1,2\}$ we set 
\begin{IEEEeqnarray*}{rCl}
R^c_{k1}&:=&\left\lfloor\tfrac{1}{3}n_i\right\rfloor+\left\lfloor\tfrac{1}{2}\left[\tfrac{1}{3}n_i-(n_2-(n_1-n_i))^+\right]^+\right\rfloor,\\
R^c_{k2}&:=&\min\{\left\lfloor\tfrac{1}{3}n_i\right\rfloor, (n_2-(n_1-n_i))^+\},\\
R^{p1}_{k1}&:=&n_1-n_i.
\end{IEEEeqnarray*}

\begin{figure}
\centering
\includegraphics[scale=1.2]{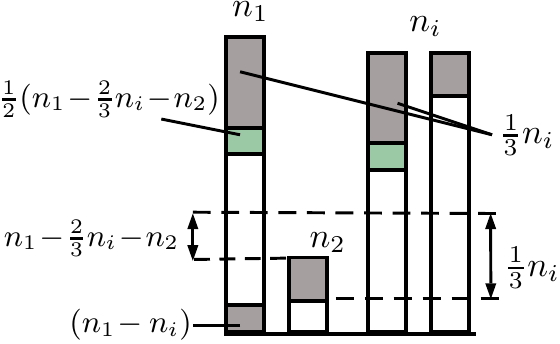}
\caption{Illustration of the scheme for the case III.C.2. In this case, $n_2$ is too weak to support the full $\tfrac{n_i}{3}$ bit-levels. It can support a maximum of $n_2-(n_1-n_i)$ bit-levels. The remaining free $n_1-\tfrac{2n_i}{3}-n_2$ bit-levels can be added at $\mathbf{\bar{x}}_{k1}$. However, since every additional bit generates additional interference (both green parts), we can only allocate half of the remaining bits to stay within decoding conditions.}
\label{Case3_Illustration2}
\end{figure}

For $\alpha\geq \tfrac{3}{4}$ we have $n_1-\tfrac{1}{3}n_i\leq n_i$ which means that $R_{k1}^c$ drops below interference level and therefore overlaps with $R_{l1}^c$ and $R_{l2}^c$, the aligned interference signals of the other cell. Additionally, more than  $\tfrac{1}{3}n_i$ free bit-levels are available between the private part $R^{p1}_{k1}$ and the aligned interference ($R_{l1}^c$ and $R_{l2}^c$). This means, that full multi-user gain can be supported even if $n_2< n_1-\tfrac{1}{3}n_i$ for as long as $n_2\geq n_1-\frac{2}{3}n_i$ (III.C.1). If $n_2< n_1-\frac{2}{3}n_i$, $\mathbf{\bar{x}}_{k2}$ is not strong enough to support full multi-user gain, and the allocation $R^c_{k2}:=n_2-(n_1-n_i)$ gets active, resulting in case III.C.2. As in the previous sub-case III.B, $\mathbf{\bar{x}}_{k1}$ gets active and can allocate half of the available difference of $(n_1-n_i)+\tfrac{1}{3}n_i-n_2$ bit-levels. Note that one cannot use all of the bit-levels between $n_2$ and $n_i-\tfrac{1}{3}n_i$ which are $\tfrac{2}{3}n_i-n_2$, because  it would violate decoding condition 1 for $\alpha>\tfrac{3}{4}$. As in the previous case, the decoding lemma needs to be adjusted for $n_2<n_i$. There is multi-user gain in this case, as long as $n_2> n_1-n_i$, below the LTD-IMAC falls back to the LTD-IC.

{\bf Sum Rate:}
The sum rate for the case III.C.1 is
\begin{IEEEeqnarray*}{rCl}
R_\Sigma= 4\left\lfloor\tfrac{1}{3}n_i\right\rfloor+2(n_1-n_i)\geq 2n_1-\tfrac{2}{3}n_i-4,
\end{IEEEeqnarray*}
where $R_{ach,2}$ is active. For the case III.C.2 
\begin{IEEEeqnarray*}{rCl}
R_\Sigma &=&2(\tfrac{1}{3}n_i+(n_1-n_i)+(n_2-(n_1-n_i))^+\\
&&+ \tfrac{1}{2}(\tfrac{1}{3}n_i-(n_2-(n_1-n_i))^+)-2\\
&\geq& n_2+n_1-2,
\end{IEEEeqnarray*} the scheme reaches $R_{ach,5}$.

{\bf Case IV ($1 < \alpha < \tfrac{3}{2}$):} 
For the range $\alpha>1$ the decoding lemma needs to be adjusted. A simple reassignment of the corresponding rates solves this problem. We can set
\begin{IEEEeqnarray*}{rCl}
R^c_{k1}&:=&\left\lfloor\tfrac{1}{3}n_i\right\rfloor+\left\lfloor\tfrac{1}{2}(\tfrac{1}{3}n_i-n_2)^+\right\rfloor, R^c_{k2}:=\min\{\left\lfloor\tfrac{1}{3}n_i\right\rfloor,n_2\}.
\end{IEEEeqnarray*}
Since the interfering signal is stronger than both users, the private part of $\mathbf{\bar{x}}_{k1}$ vanishes completely. We have to analyse two sub-cases:
Case IV.1:\ This sub-case is in the range $n_2\geq \left\lfloor\tfrac{1}{3}n_i\right\rfloor$. Here, $R^c_{k2}:=\left\lfloor\tfrac{1}{3}n_i\right\rfloor$, and the multi-user gain can be completely supported by the weaker users signal $\mathbf{\bar{x}}_{k2}$.
Case IV.2:\ For $n_2<\left\lfloor\tfrac{1}{3}n_i\right\rfloor$, the multi-user gain cannot be fully supported. We have $R^c_{k2}:=n_2$ and for $n_2=0$, the channel falls back to the IC sum-rate.

{\bf Sum Rate:}
The sum rate for the cases IV.1 and IV.2 is $R_\Sigma\geq \frac{4}{3}n_i-4$ and $R_\Sigma\geq n_i+n_2-2$, respectively. In case IV.1, $R_{ach,2}$ is active, and for case IV.2, $R_{ach,5}$ is active.

{\bf Case V ($\tfrac{3}{2} \leq \alpha < \infty$):} At $\alpha=\tfrac{3}{2}$, we have that $n_1=\tfrac{2}{3}n_i$ and therefore reached the maximum allocatable rate point which satisfies all decoding conditions with the previous rate allocation. Since $n_i$ keeps growing, the strategy is to use half of $n_1$ as allocation and we set  
\begin{IEEEeqnarray*}{rCl}
R^c_{k1}&:=&\max\{\left\lfloor\tfrac{1}{2}n_1\right\rfloor,\left\lfloor\tfrac{1}{2}n_i-\tfrac{1}{2}n_2\right\rfloor,n_1-n_2\},\\
R^c_{k2}&:=&\min\{\left\lfloor\tfrac{1}{2}n_1\right\rfloor,n_2\}.
\end{IEEEeqnarray*}
Once again, multi-user gain is dependent on the strength of $n_2$. For $n_2\geq\left\lfloor\tfrac{1}{2}n_1\right\rfloor$ (case V.I), the second user can fully support the gain but for $n_2<\left\lfloor\tfrac{1}{2}n_1\right\rfloor$ just $R^c_{k2}:=n_2$. For this case, we have $R^c_{k1}:=n_1-n_2$ as long as $2n_1<n_i+n_2$ and get the full sum-rate. Otherwise, full multi-user gain cannot be supported and we have $R^c_{k1}:=\left\lfloor\tfrac{1}{2}n_i-\tfrac{1}{2}n_2\right\rfloor$ and reach the case V.2.

{\bf Sum Rate:}
The sum rate for the cases V.1 and V.2 is $R_\Sigma\geq 2n_1-4$ and $R_\Sigma\geq n_i+n_2-2$, respectively. 

As in \cite{Niesen-Ali}, we have ignored the $\log(32/\delta)$ terms from the decoding lemma, and add a reduction in the overall sum-rate of $2\log(32/\delta)$. Together with a bit-gap of at most 4, we get the overall gap $2\log(128/\delta)$.
\end{proof}
\subsection{Upper Bounds}
\label{UB_LTDM}
\begin{theorem}[]

The sum rate for the symmetric LTD-IMAC system model can be bounded from above by 

\begin{equation}
R_\Sigma \leq \min\{D_1,D_2,D_3,D_4,D_5\}
\end{equation}
with
\begin{IEEEeqnarray}{rCl}
D_1&:=& 2\max((n_1-n_i)^+,n_i)+\min((n_1-n_i)^+,n_i),\IEEEyessubnumber\label{ub:D1}\\
D_2&:=& \tfrac{2}{3}(2\max(n_1,n_i)+(n_1-n_i)^+),\IEEEyessubnumber\label{ub:D2}\\
D_3&:=& 2n_1,\IEEEyessubnumber\label{ub:D3}\\
D_4&:=& \max(2n_2,2(n_1-n_i)^+,2n_i),\IEEEyessubnumber\label{ub:D4}\\
D_5&:=& \max(n_1,n_i)+\max(n_2,(n_1-n_i)^+)\IEEEyessubnumber\label{ub:D5}.
\end{IEEEeqnarray}

\label{upBound}
\end{theorem}
\begin{proof}

Considering Fano's inequality and the Data Processing inequality one can establish the following bounds:
\begin{IEEEeqnarray*}{rCl}
\IEEEeqnarraymulticol{3}{l}{
n(R^1_{11}+R^1_{21}+R^2_{21}+R^2_{22})}\\
& \leq & I(\mathbf{\bar{x}}_{11}^n,\mathbf{\bar{x}}_{12}^n;(\mathbf{y}^1)^n)+I(\mathbf{\bar{x}}_{21}^n,\mathbf{\bar{x}}_{22}^n; (\mathbf{y}^2)^n) +n(\epsilon_{n,12}+\epsilon_{n,34})\\
& = & H((\mathbf{y}^1)^n)-H((\mathbf{y}^1)^n| \mathbf{\bar{x}}_{11}^n,\mathbf{\bar{x}}_{12}^n)+H((\mathbf{y}^2)^n)\\
&&-\: H((\mathbf{y}^2)^n|\mathbf{\bar{x}}_{21}^n,\mathbf{\bar{x}}_{22}^n)+n(\epsilon_{n,12}+\epsilon_{n,34}).
\end{IEEEeqnarray*}

We denote $(\mathbf{y}^j)^{n,\uparrow}:=(\mathbf{y}^j)^{n}_{[\eta+1:n_1]}$ and $(\mathbf{y}^j)^{n,\downarrow}:=(\mathbf{y}^j)^{n}_{[1:\eta]}$ with $\eta:=\max((n_1-n_i)^+,n_i)$ and show that 
\begin{IEEEeqnarray*}{rCl}
\IEEEeqnarraymulticol{3}{l}{
2 n(R_\Sigma - \epsilon_{n,\Sigma})}\\
 &\leq& 2 H((\mathbf{y}^1)^n)-2 H ((\mathbf{y}^1)^n| \mathbf{\bar{x}}_{11}^n,\mathbf{\bar{x}}_{12}^n) +2 H(\mathbf{y}_{12}^n)\\
 &&-\: 2H((\mathbf{y}^2)^n|\mathbf{\bar{x}}_{21}^n,\mathbf{\bar{x}}_{22}^n)\\
& \leq & 2 H((\mathbf{y}^1)^{n,\downarrow})+2 H((\mathbf{y}^1)^{n,\uparrow})-2H(\mathbf{H}_2^1 ((\mathbf{\bar{x}}_{21}^c)^n\oplus (\mathbf{\bar{x}}_{22}^c)^n))\\
&&+\: 2 H((\mathbf{y}^2)^{n,\downarrow})+2 H((\mathbf{y}^2)^{n,\uparrow})-2H(\mathbf{H}_1^2((\mathbf{\bar{x}}_{11}^c)^n\oplus (\mathbf{\bar{x}}_{12}^c)^n))\\ 
& \overset{(a)}{\leq} & 4n\max((n_1-n_i)^+,n_i) + H((\mathbf{y}^1)^{n,\uparrow})+ H((\mathbf{y}^2)^{n,\uparrow})\\ 
& \leq & 4n\max((n_1-n_i)^+,n_i) + 2n\min((n_1-n_i)^+,n_i)
\end{IEEEeqnarray*}
where (a) follows from $H((\mathbf{y}^j)^{n,\downarrow})\leq n\eta$ and \begin{IEEEeqnarray*}{rCl}
\IEEEeqnarraymulticol{3}{l}{
H((\mathbf{y}^{1})^{n,\uparrow})-2H(\mathbf{H}_1^2((\mathbf{\bar{x}}^c_{11})^n \oplus (\mathbf{\bar{x}}^c_{12})^n))}\\
& \leq & H((\mathbf{y}^{1})^{n,\uparrow}) -H(\mathbf{H}_1^2((\mathbf{\bar{x}}^c_{11})^n \oplus (\mathbf{\bar{x}}^c_{12})^n)|(\mathbf{\bar{x}}^c_{11})^n)\\
&&-\:H(\mathbf{H}_1^2((\mathbf{\bar{x}}^c_{11})^n \oplus (\mathbf{\bar{x}}^c_{12})^n)|(\mathbf{\bar{x}}^c_{12})^n) \\
& = & H((\mathbf{y}^{1})^{n,\uparrow})  -H((\mathbf{\bar{x}}^c_{11})^n)-H((\mathbf{\bar{x}}^c_{12})^n)\\
& \leq & 0.
\end{IEEEeqnarray*}
Here we used in the second last step, that the $\mathbf{H}$ matrices are lower uni-triangular matrices. This means that they are invertible, hence are bijective mappings. By symmetry is also holds that 
\begin{equation}
H((\mathbf{y}^{2})^{n,\uparrow})-2H(\mathbf{H}_2^1((\mathbf{\bar{x}}^c_{21})^n \oplus (\mathbf{\bar{x}}^c_{22})^n))\leq 0.
\end{equation}
Dividing both sides by $2n$ and taking $n\rightarrow \infty$ yields the upper bound \eqref{ub:D1}.

We now establish the bound \eqref{ub:D2}. Therefore we show that 
\begin{IEEEeqnarray*}{rCl}
\IEEEeqnarraymulticol{3}{l}{
3n(R^1_{11}+R^1_{21}+R^2_{21}+R^2_{22}-\epsilon)}\\
& \leq & 3I(\mathbf{\bar{x}}_{11}^n,\mathbf{\bar{x}}_{12}^n;(\mathbf{y}^1)^n)+3I(\mathbf{\bar{x}}_{21}^n,\mathbf{\bar{x}}_{22}^n; (\mathbf{y}^2)^n)\\
& = & 2H((\mathbf{y}^1)^n)-2H((\mathbf{y}^1)^n| \mathbf{\bar{x}}_{11}^n,\mathbf{\bar{x}}_{12}^n)+2H((\mathbf{y}^2)^n)\\
&&-\: 2H((\mathbf{y}^2)^n|\mathbf{\bar{x}}_{21}^n,\mathbf{\bar{x}}_{22}^n)+ H(\mathbf{\bar{x}}_{11}^n\oplus \mathbf{\bar{x}}_{12}^n)+H(\mathbf{\bar{x}}_{21}^n\oplus \mathbf{\bar{x}}_{22}^n).\\
& \leq & 2H((\mathbf{y}^1)^n)+2H((\mathbf{y}^2)^n)+2n(n_1-n_i)^+
\end{IEEEeqnarray*} 
with \begin{IEEEeqnarray*}{rCl}
\IEEEeqnarraymulticol{3}{l}{
H(\mathbf{\bar{x}}_{i1}^n\oplus \mathbf{\bar{x}}_{i2}^n)-2H(\mathbf{H}_i^{i^c}((\mathbf{\bar{x}}^c_{i1})^n \oplus (\mathbf{\bar{x}}^c_{i2})^n))}\\
& \leq & H(\mathbf{\bar{x}}_{i1}^n\oplus \mathbf{\bar{x}}_{i2}^n) -H(\mathbf{H}_i^{i^c}((\mathbf{\bar{x}}^c_{i1})^n \oplus (\mathbf{\bar{x}}^c_{i2})^n)|(\mathbf{\bar{x}}^c_{i1})^n)\\
&&-\:H(\mathbf{H}_i^{i^c}((\mathbf{\bar{x}}^c_{i1})^n \oplus (\mathbf{\bar{x}}^c_{i2})^n)|(\mathbf{\bar{x}}^c_{i2})^n) \\
& = & H(\mathbf{\bar{x}}_{i1}^n\oplus \mathbf{\bar{x}}_{i2}^n) -H((\mathbf{\bar{x}}^c_{i1})^n)-H((\mathbf{\bar{x}}^c_{i2})^n)\\
&\leq & H((\mathbf{\bar{x}}_{i1}^n\oplus \mathbf{\bar{x}}_{i2}^n)_{[1:(n_1-n_i)^+]})\\
& \leq & n(n_1-n_i)^+
\end{IEEEeqnarray*}
for $i\in \{1,2\}$ in all\footnote{Note that the index $i$ at the signals identifies the MAC-cell and has nothing to do with the $i$ at $n_i$ which stands for {\it interference.}} $\mathbf{\bar{x}}_{i1}^n$ and $\mathbf{\bar{x}}_{i2}^n$ and therefore
\begin{IEEEeqnarray*}{rCl}
\IEEEeqnarraymulticol{3}{l}{
n(R_\Sigma-\epsilon)}\\
&\leq & \tfrac{2n}{3}(2\max(n_1,n_i)+(n_1-n_i)^+),
\end{IEEEeqnarray*} 
which shows the upper bound \eqref{ub:D2} for $n\rightarrow \infty$.

We now establish the bound \eqref{ub:D3}.
We can show that
\begin{IEEEeqnarray*}{rCl}
\IEEEeqnarraymulticol{3}{l}{
n(R^1_{11}+R^1_{21}+R^2_{21}+R^2_{22}-\epsilon)}\\
& \leq & I(\mathbf{\bar{x}}_{11}^n,\mathbf{\bar{x}}_{12}^n;(\mathbf{y}^1)^n,\mathbf{\bar{x}}_{21}^n,\mathbf{\bar{x}}_{22}^n)+I(\mathbf{\bar{x}}_{21}^n,\mathbf{\bar{x}}_{22}^n; (\mathbf{y}^2)^n,\mathbf{\bar{x}}_{11}^n,\mathbf{\bar{x}}_{12}^n) \\
& = & I(\mathbf{\bar{x}}_{11}^n,\mathbf{\bar{x}}_{12}^n;(\mathbf{y}^1)^n|\mathbf{\bar{x}}_{21}^n,\mathbf{\bar{x}}_{22}^n)+I(\mathbf{\bar{x}}_{21}^n,\mathbf{\bar{x}}_{22}^n; (\mathbf{y}^2)^n|\mathbf{\bar{x}}_{11}^n,\mathbf{\bar{x}}_{12}^n) \\
& = & H((\mathbf{y}^1)^n|\mathbf{\bar{x}}_{21}^n,\mathbf{\bar{x}}_{22}^n)+H((\mathbf{y}^2)^n|\mathbf{\bar{x}}_{11}^n,\mathbf{\bar{x}}_{12}^n)\\
& \leq & 2nn_1.
\end{IEEEeqnarray*}
Dividing by $n$ shows the upper bound \eqref{ub:D3} for $n\rightarrow \infty$.

Now, we show the upper bound \eqref{ub:D4}.
One can show that:
\begin{IEEEeqnarray*}{rCl}
\IEEEeqnarraymulticol{3}{l}{
n(R^1_{11}+R^1_{21}+R^2_{21}+R^2_{22}-\epsilon_{n,\Sigma})}\\
& \leq & H((\mathbf{y}^1)^n)-H((\mathbf{y}^1)^n| \mathbf{\bar{x}}_{11}^n,\mathbf{\bar{x}}_{12}^n)\\
&&+\: H((\mathbf{y}^2)^n)-H((\mathbf{y}^2)^n|\mathbf{\bar{x}}_{21}^n,\mathbf{\bar{x}}_{22}^n)\\
& = & H((\mathbf{y}^1)^n)-H(\mathbf{H}_2^1 ((\mathbf{\bar{x}}_{21}^c)^n \oplus (\mathbf{\bar{x}}_{22}^c)^n))\\
&&+\: H((\mathbf{y}^2)^n)-H(\mathbf{H}_1^2((\mathbf{\bar{x}}_{11}^c)^n\oplus (\mathbf{\bar{x}}_{12}^c)^n))\\
& \leq & H((\mathbf{y}^1)^n)-H(\mathbf{H}_2^1 (\mathbf{\bar{x}}_{21}^c)^n )\\
&&+\: H((\mathbf{y}^2)^n)-H(\mathbf{H}_1^2(\mathbf{\bar{x}}_{11}^c)^n)\\
& \leq & 2n\max(n_2,(n_1-n_i)^+,n_i)
\end{IEEEeqnarray*}
Dividing by $n$ and letting $n\rightarrow \infty$ shows the bound \eqref{ub:D4}.

We now establish the bound \eqref{ub:D5}.
We can show that
\begin{IEEEeqnarray*}{rCl}
\IEEEeqnarraymulticol{3}{l}{
 n(R_\Sigma- \epsilon_{n,\Sigma})}\\
 &\leq&  H((\mathbf{y}^1)^n)- H ((\mathbf{y}^1)^n| \mathbf{\bar{x}}_{11}^n,\mathbf{\bar{x}}_{12}^n) + H((\mathbf{y}^{2})^n)\\
 &&-\: H((\mathbf{y}^2)^n|\mathbf{\bar{x}}_{21}^n,\mathbf{\bar{x}}_{22}^n)\\
 &\leq&  H((\mathbf{y}^1)^n)- H ((\mathbf{y}^1)^n| \mathbf{\bar{x}}_{11}^n,\mathbf{\bar{x}}_{12}^n,\mathbf{\bar{x}}_{22}^n) + H((\mathbf{y}^{2})^n)\\
 &&-\: H((\mathbf{y}^2)^n|\mathbf{\bar{x}}_{21}^n,\mathbf{\bar{x}}_{22}^n)\\
  &\leq &  H((\mathbf{y}^1)^n)- H (\mathbf{\bar{x}}_{21}^n) + H((\mathbf{y}^{2})^n)\\
 &&-\: H(\mathbf{\bar{x}}_{11}^n\oplus \mathbf{\bar{x}}_{12}^n)\\
   &\leq&  H((\mathbf{y}^1)^n)+ n\max(n_2,(n_1-n_i)^+)\\
& \leq & n\max(n_1,n_i)+n\max(n_2,(n_1-n_i)^+).
\end{IEEEeqnarray*}
Dividing both sides by $n$ and taking $n\rightarrow \infty$ yields the desired upper bound. 

\end{proof}
\begin{corollary}[]
If the weaker user is strong enough to support full multi-user gain, the sum rate for the symmetric LTD-IMAC system model can be bounded from above by 
\begin{equation}
R_\Sigma \leq \begin{cases}
2(n_1-\tfrac{1}{2}n_i) & \text{for  } 0 \leq \alpha \leq \tfrac{1}{2}\\
2(\tfrac{1}{2}n_1+\frac{1}{2}n_i) & \text{for } \tfrac{1}{2} \leq \alpha \leq \tfrac{3}{5}\\
2(n_1-\tfrac{1}{3}n_i) & \text{for } \tfrac{3}{5} \leq \alpha \leq 1\\
\tfrac{4}{3}n_i & \text{for } 1 \leq \alpha \leq \tfrac{3}{2}\\
2n_1 & \text{for } \tfrac{3}{2} \leq \alpha \leq \infty.
\end{cases}
\end{equation}
\label{upBoundweak}
\end{corollary}
\begin{proof}
Follows immediately from Theorem \ref{UB_LTDM} by investigation of the active bounds in the corresponding regimes.
\end{proof}

\begin{remark}
The upper bounds for the weak interference asymmetric cases can be found in \cite{Fritschek2015a}. Upper bounds for weak symmetric cases correspond to the specific achievable schemes. 
Since the model can be split in the weak interference regime, the symmetric upper bounds can be split and mixed as well and show the asymmetric bounds as well. 
\end{remark}


\section{Transfer from LTD-IMAC to G-IMAC}
\label{Proof G-IMAC}
\subsection{Achievability for the G-IMAC}
In this part, we prove the achievability for the Gaussian IMAC. We will show, that the lower-triangular scheme directly guides the constant-gap capacity achieving Gaussian scheme. The analysis is done in the same fashion as in \cite{Niesen-Ali}. We assume perfect knowledge of the channel gains $h^j_{ik}$. However, the same techniques as in \cite{Niesen-Ali} can be applied to show, that a $\max$ $n^j_{ik}$-bit quantisation of $h^j_{ik}$ is sufficient for the achievability. 
Remember that the input signals are constructed such that 
\begin{IEEEeqnarray}{rCl}
x_{11}&:=&h_{12}^2 u_{11}\IEEEyessubnumber\\
x_{12}&:=&h_{11}^2 u_{12}\IEEEyessubnumber\\
x_{21}&:=&h_{22}^1 u_{21}\IEEEyessubnumber\\
x_{22}&:=&h_{21}^1 u_{22}\IEEEyessubnumber.
\end{IEEEeqnarray}

Observe, that this results in \eqref{Gauss_Model_g}, where the new channel gains are $g_{ik}^j\in(1,4]$.
Due to this modulation, the first two bits of $u_{ik}$ are set to zero i.e.,
\begin{equation}
u_{ik}:=\sum_{j=3}^{n_{i1}} [u_{ik}]_j 2^{-j},
\end{equation}
where $[u_{ik}]_j\in \{0,1\}$ represents the bits of the binary expansion of $u_{ik}$. This ensures that $|u_{ik}|\leq \tfrac{1}{4}$ and therefore 
$|g^j_{ik}u_{ik}|\leq 1$.
Moreover, the $u_{ik}$-inputs are chosen such that the corresponding $[u_{ik}]_j$-bits obey the design criteria of the LTD scheme in \ref{LTD-achiev.}. In particular, the places where the binary vectors of the LTD model are forced to be zero, need to be zero in the Gaussian scheme as well.
Therefore, the $u_{i1}$-inputs will be also decomposed into common and private signal parts i.e.,

\begin{equation}
u_{i1}=u_{i1}^P+u_{i1}^C.
\end{equation}
Where the private part $u_{i1}^P:=u_{i1}^{P_1}+u_{i1}^{P_2}$ is decomposed again into two parts for the weak interference regime i.e., $\tfrac{n_i}{n_1}\leq \tfrac{1}{2}$.

In the following, we analyse the receiver one exemplary. By symmetry, all results for receiver one apply for receiver two as well.
The channel equation for receiver one is
\begin{IEEEeqnarray}{rCl}
y^1 &=& g_{11}^1 2^{n^1_{11}}u_{11}+g_{12}^1 2^{n_{12}^1} u_{12}\\
&&+\: g_{2}^1 2^{n_2^1} (u_{21}^C+  u_{22}^C) +(g_{2}^12^{n_2^1}u_{21}^P+ z^1).
\end{IEEEeqnarray}
Observe that the private part $u_{21}^P$ of transmitter two is grouped with the Gaussian noise. The purpose of the private part is, that it is received below the noise floor of the unintended receiver. By the specific structure of the LTD scheme, the private part is given by 
\begin{equation}
u_{i1}^P:=\sum_{j=n_{2}^1+3}^{n_{i1}} [u_{i1}]_j 2^{-j},
\end{equation}
where the two first bits are set to zero. Therefore we have $|2^{n_2^1}u_{21}^P|\leq \tfrac{1}{4}$ and $|g_{2}^1 2^{n_2^1}u_{21}^P|\leq 1$. It follows, that the last parts of the received signal
$(g_{2}^12^{n_2^1}u_{21}^P+ z^1)$ can be treated as noise. Moreover, one can see that the two interfering signals from receiver two $u_{21}^C$ and $u_{22}^C$ are received with the same channel gain and therefore align at receiver one.
Once again following the notation of \cite{Niesen-Ali}, we can write the received signal parts above noise as
\begin{IEEEeqnarray}{rCl}
s_{11}&:=&2^{n^1_{11}}u_{11}\IEEEyessubnumber\\
s_{12}&:=&2^{n_{12}^1} u_{12}\IEEEyessubnumber\\
s_{2}^1&:=& 2^{n_2^1} (u_{21}^C+  u_{22}^C)\IEEEyessubnumber.
\end{IEEEeqnarray}
The decoder at receiver one tries to find estimates $\hat{s}_{11}$, $\hat{s}_{12}$, $\hat{s}_{2}^1$ for the received signal parts. Therefore, it is only interested in the sum of both interfering signals. The goal is to find the specific $\hat{s}_{11}$, $\hat{s}_{12}$, $\hat{s}_{2}^1$ which minimize the distance to the received signal, i.e.
\begin{equation}
|y_1-g_{11}^1\hat{s}_{11}-g_{12}^1\hat{s}_{12}-g_{2}^1\hat{s}_{2}^1|.
\end{equation}
An error can occur, if the distance between $y_1$ and any other triple $(\hat{s}_{11}$, $\hat{s}_{12}$, $\hat{s}_{2}^1)$ has a smaller distance than the noise. Therefore we need to investigate the minimum distance between the received signal parts and any other triple, which is
\begin{IEEEeqnarray}{rCl}
d&:=&\min_{(s_{11},s_{12},s_{2}^1)\neq (\hat{s}_{11},\hat{s}_{12},\hat{s}_{2}^1)} \\
&&|g_{11}^1(s_{11}-\hat{s}_{11})-g_{12}^1(s_{12}-\hat{s}_{12})-g_{2}^1(s_{2}^1-\hat{s}_{2}^1)|.\IEEEnonumber
\end{IEEEeqnarray}
Observe that the structure of the problem is exactly the same as in \cite[p.~4869]{Niesen-Ali}, although the network model  is different\footnote{In particular the coarse channel gain of the aligning signals, and therefore the form of $s_2^1$}.
For the interference regime where $\alpha\geq \tfrac{1}{2}$, we can use a Lemma which was proved in \cite{Niesen-Ali}.

\begin{lemma}[{\cite[Lemma~9]{Niesen-Ali}}]
Let $\delta\in (0,1]$ and $n_1$, $n_2$, $n_i \in \mathbb{N}$, $n_1 \geq n_2$ and $\tfrac{n_i}{n_1}\geq \tfrac{1}{2}$. Assume $R^P_{11}$, $R^{C}_{11}$, $R_{12}$, $R_{21}$, $R^P_{22}$, $R^{C}_{22}$ $\in \mathbb{Z}_+$ satisfy,
\begin{IEEEeqnarray*}{rCl}
R^{C}_{11}+R^C_{12} +R_{21}^C + R^P_{11}&\leq &  n_{1}-\log\left(\tfrac{c}{\delta}\right)\\
R_{12}^C+R_{21}^C + R^P_{11}&\leq&  n_{2}-\log\left(\tfrac{c}{\delta}\right)\\
R_{21}^C+R^P_{11}&\leq&  n_i-6
\end{IEEEeqnarray*} 
and
\begin{IEEEeqnarray*}{rCl}
R^{C}_{21}+R_{22}^C+R^{C}_{11}+ R^P_{22}&\leq&  n_1-\log\left(\tfrac{c}{\delta}\right)\\
R_{22}^C+R^{C}_{11}+ R^P_{22}&\leq&  n_2-\log\left(\tfrac{c}{\delta}\right)\\
R^{C}_{11}+ R^P_{22}&\leq&  n_i-6
\end{IEEEeqnarray*} 
where $c:=13104$ and $R_{ik}$ is the rate of the signal $u_{ik}$. Then, the bit allocation of the LTD-IMAC applied to the Gaussian IMAC results in a minimum constellation distance $d\geq32$ at each receiver for all channel gains ($h_{mk}\in(1,2]^{2\times 2}$) except for a set $B\subset (1,2]^{2 \times 2}$ of Lebesgue measure $\mu(B)\leq \delta$.
\label{decoding_lemma_g_strong}
\end{lemma}

Since the structure of the conditions is the same as in the LTD-IMAC case (Lemma~\ref{general-ach-ltd-lem}), we see that the G-IMAC can achieve the sum rate of Lemma~\ref{ltd-upBound1} (with appropriate adjustment of the constants) with small probability of error.
However, note that we cannot apply this Lemma to our weak interference case. This is due to the private part, which is composed of the two private signals $u^{P_1}_{11},u^{P_2}_{11}$. We provide the proof for an adjusted Lemma in Appendix~\ref{Proof_Gaussian_Weak_Case}. 
Moreover, Lemma~\ref{decoding_lemma_g_strong} only provides conditions for the decoding error to be small. As in \cite{Niesen-Ali}, it can be shown that an outer code over the modulated channel results in a vanishing error probability. Due to the structural similarities between the X-channel and the IMAC, all results in \cite{Niesen-Ali} regarding vanishing error probability also apply for the IMAC. In particular, for an outer code with rate $R'_{ik}$ and a modulation rate of $R_{ik}$ it holds that
\begin{IEEEeqnarray}{rCl}
R'_{ik}&=& I(u_{ik};\hat{s}_{11}, \hat{s}_{12}, \hat{s}_{2}^1)\IEEEnonumber\\
&=&I(s_{ik};\hat{s}_{11}, \hat{s}_{12}, \hat{s}_{2}^1)\IEEEnonumber\\
&\geq& R_{ik}-1.5,
\label{outer_modulation}
\end{IEEEeqnarray} which uses the bound
\begin{equation*}
H(s_{11}, s_{12}, s_{2}^1|\hat{s}_{11}, \hat{s}_{12}, \hat{s}_{2}^1)\leq 1.5,
\end{equation*} proven in \cite{Niesen-Ali}, under usage of the conclusions of Lemma~\ref{decoding_lemma_g_strong}.
This can be shown for each signal. The sum of the lower bounds then shows, that an outer code can achieve the previous sum rate with a vanishing error probability within a constant gap of six bits.

\section{On The Difference between the LD Model and the LTD Model}

Let us look at an specific example for both models. Assume we have a symmetrical channel, such that
\begin{IEEEeqnarray*}{rCl}
 n_{11}^1&=&n_{21}^2=n_1=11\times\Delta \text{ bits}\\
 n_{12}^1&=&n_{22}^2=n_2=10\times \Delta \text{ bits}\\
  n_{21}^1&=&n_{22}^1=n_{11}^2=n_{12}^2=n_i=5\times\Delta \text{ bits}.
\end{IEEEeqnarray*}
The parameters are $\alpha_1=\alpha_2=\tfrac{5}{11}$ and $\beta_1=\beta_2=\tfrac{10}{11}$. Figure~\ref{Fig17} shows the achievable strategy for the LD model. We partition the common part of the signal in $\Delta$-bit partitions such that we have maximal direct bit-rate while minimizing the aligning interference. This results in an achievable bit-rate of $8\Delta$ (sum rate of $16\Delta$) bits using an orthogonal alignment scheme, where bit-levels are used independently. Lets assume the optimal input distribution in the IMAC setting is uniform. If $X\sim \text{Unif[0,1]}$ we have that the binary expansion $X=\sum_{n=1}^\infty x_n 2^{-n}$ yields i.i.d. Bern$(\tfrac{1}{2})$ $x_n$. Which would mean that the bits of the binary received vector $\mathbf{y}$ are also i.i.d. Bern($\tfrac{1}{2}$). Hence, we could show the following converse bound
\begin{IEEEeqnarray*}{rCl}
n(\sum R-\epsilon)&\leq& I(\mathbf{x}_{11}^n,\mathbf{x}_{12}^n;\mathbf{y}_1^n)+I(\mathbf{x}_{21}^n,\mathbf{x}_{22}^n;\mathbf{y}_2^n)\\
&\leq & H(\mathbf{y}_1^n)-H(\mathbf{y}_1^n|\mathbf{x}_{11}^n,\mathbf{x}_{12}^n)\\
&&+H(\mathbf{y}_2^n)-H(\mathbf{y}_2^n|\mathbf{x}_{21}^n,\mathbf{x}_{22}^n)\\
&\leq & H(\mathbf{y}_1^{\uparrow,n})+H(\mathbf{y}_1^{\downarrow,n})-H(\mathbf{y}_1^n|\mathbf{x}_{11}^n,\mathbf{x}_{12}^n)\\
&&+H(\mathbf{y}_2^{\uparrow,n})+H(\mathbf{y}_2^{\downarrow,n})-H(\mathbf{y}_2^n|\mathbf{x}_{21}^n,\mathbf{x}_{22}^n).
\end{IEEEeqnarray*}
Now we can upper bound $H(\mathbf{y}_1^{\downarrow,n})$ and $H(\mathbf{y}_2^{\downarrow,n})$ by $n6\Delta$ bits.
Moreover, we can split 
\begin{IEEEeqnarray*}{rCl}
H(\mathbf{y}_1^{\uparrow,n})-H(\mathbf{y}_2^n|\mathbf{x}_{21}^n,\mathbf{x}_{22}^n)&=&\\
\sum_{\Delta} H(\mathbf{y}_{\Delta,1}^{\uparrow,n})-H(\mathbf{y}_{\Delta,2}|\mathbf{x}_{21,\Delta}^n,\mathbf{x}_{22,\Delta}^n)
\end{IEEEeqnarray*}
in $\Delta$ partitions, which is possible due to the bits of $\mathbf{y}$ being independent as assumed. Now one can condition the term $H(\mathbf{y}_{\Delta,2}|\mathbf{x}_{21,\Delta}^n,\mathbf{x}_{22,\Delta}^n)$ alternately on $\mathbf{x}_{11,\Delta}^n$ and $\mathbf{x}_{12,\Delta}^n$, starting with the latter, which results in 
\begin{equation*}
H(\mathbf{y}_1^{\uparrow,n})-H(\mathbf{y}_2^n|\mathbf{x}_{21}^n,\mathbf{x}_{22}^n)\leq n2\Delta.
\end{equation*}
The same can be done with $H(\mathbf{y}_2^{\uparrow,n})-H(\mathbf{y}_1^n|\mathbf{x}_{11}^n,\mathbf{x}_{12}^n)$ resulting in an overall bound of \begin{equation*}
R_{\Sigma}\leq 16\Delta \text{ bits}.
\end{equation*}
However, without any $\mathbf{y}$ distribution assumption we can get the bound of theorem~\ref{ld-bound}
\begin{equation}
R_\Sigma \leq n^1_{11}+n^2_{21}-\frac{n^2_1}{2}-\frac{n^1_2}{2}
\end{equation}
and therefore $R_\Sigma \leq 17\Delta$ bits. Those $17\Delta$ bits are achievable by using the LTD model and an achieving strategy for our example is given in Fig.~\ref{Fig18}. The difference in the LTD model is, that a dependence of the bit-levels at the receiver is introduced by taking the binary expansion of the channel gain into account.
\begin{figure}\centering
\includegraphics[scale=1]{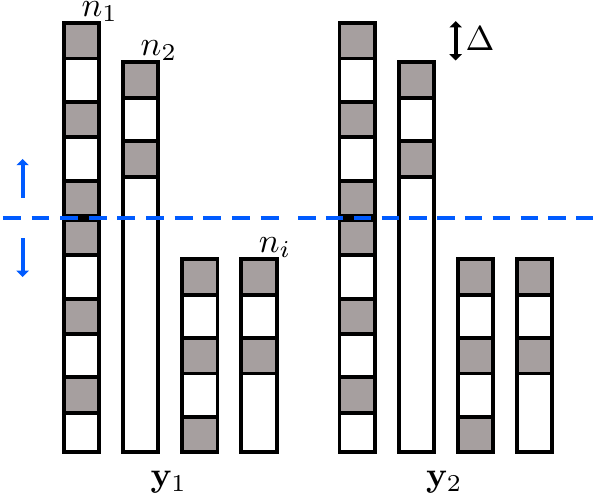}
\caption{Achievable strategy for the LDM in the exemplary setting. We have illustrated the parts of $\mathbf{y}_1^{\uparrow,n}$, $\mathbf{y}_2^{\uparrow,n}$ and $\mathbf{y}_1^{\downarrow,n}$, $\mathbf{y}_2^{\downarrow,n}$ with blue arrows. Every grey box corresponds to a $\Delta-$bit allocation.}
\label{Fig17}
\end{figure}

\begin{figure}\centering
\includegraphics[scale=1]{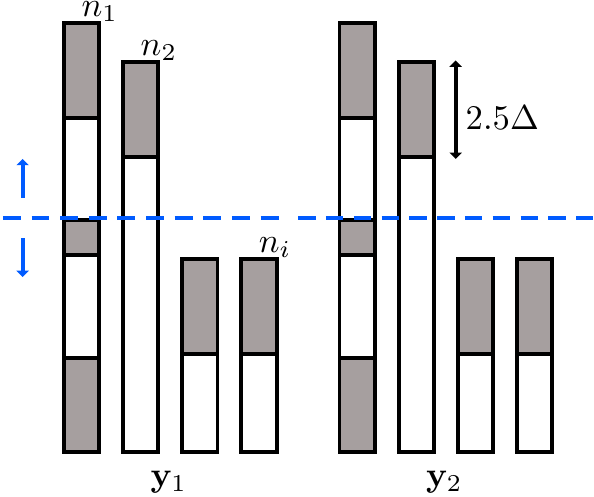}
\caption{Achievable strategy for the LTDM in the exemplary setting. We have illustrated the parts of $\mathbf{y}_1^{\uparrow,n}$, $\mathbf{y}_2^{\uparrow,n}$ and $\mathbf{y}_1^{\downarrow,n}$, $\mathbf{y}_2^{\downarrow,n}$ with blue arrows. The aligning blocks have $2.5 \Delta$ bits. The block in the middle, under the blue split, has $\Delta$ bits.}
\label{Fig18}
\end{figure}

\section{Conclusions}
\begin{figure}
\centering
\includegraphics[scale=0.38]{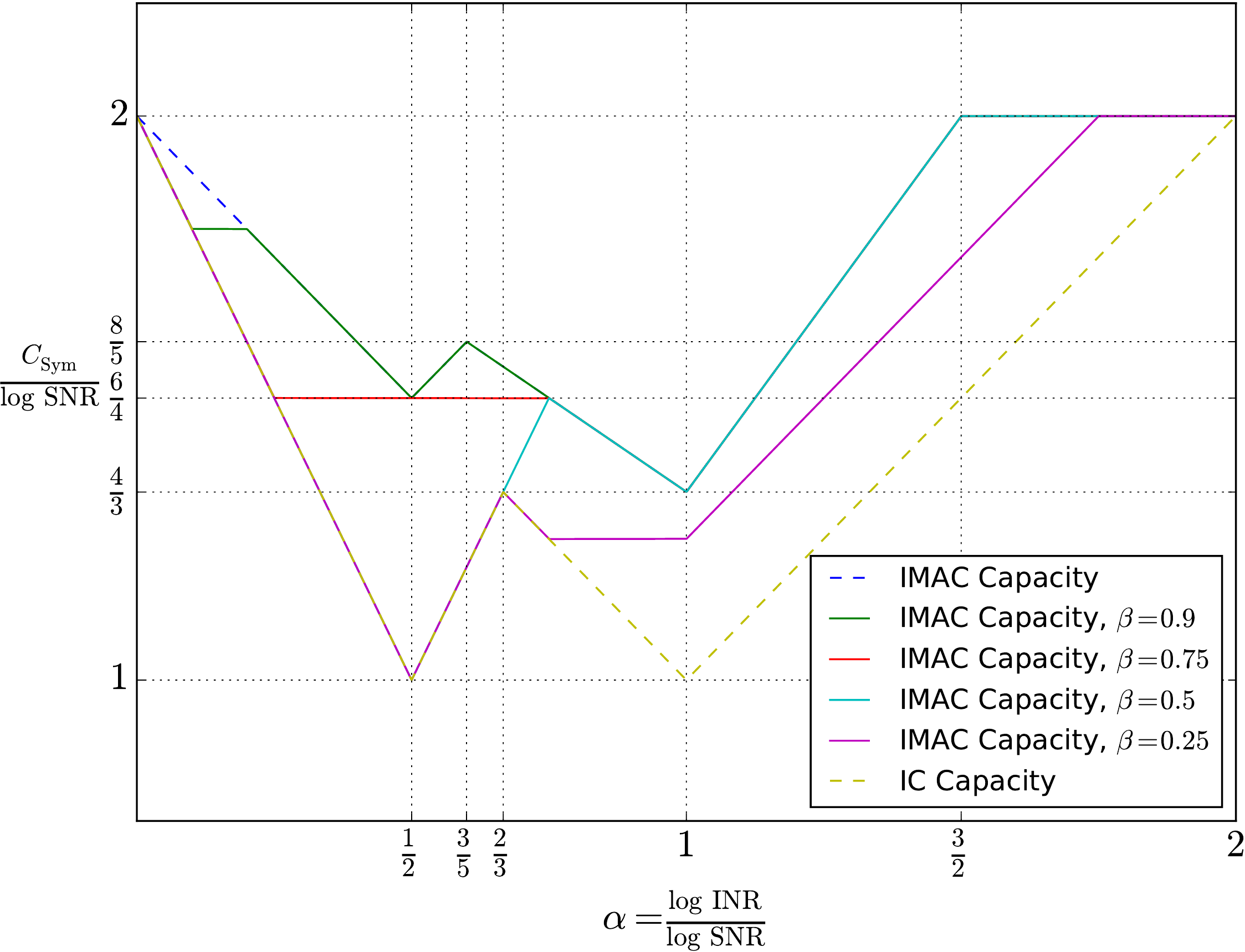}
\caption{The GDoF ``W'' curve for the Gaussian IMAC, for various $\beta$, which are the defined as the channel strength ratio between the two direct links $\mbox{SNR}_{i2}\,=\,\mbox{SNR}_{i1}^{\beta}$, or equivalently $\beta=\tfrac{n_2}{n_1}$ if fine gains are neglected. All curves are based on a symmetric channel. One can see that $\beta$ controls the multi-user gain dependent on $\alpha$. Moreover, one can see that for certain parameter ranges, the capacity is independent of differences in $\alpha$, which results in horizontal lines. The first two (from the left) show $2\beta$, the third shows $1+\beta$.}
\label{last_fig}
\end{figure}

In this paper we have investigated the Gaussian interfering multiple access channel (G-IMAC). We used the linear deterministic model (LDM) of the IMAC (coined LD-IMAC) as a first approximation to gain new insights. These insights show that with the help of interference alignment techniques in the signal-scale, upto half of the interference can be aligned and the other half can be therefore made available for communication. Moreover, we have shown upper bounds which coincide with the achievable rates for certain interference-to-signal ratios (see Fig.~\ref{ld-achiev-uB}) and are within a certain gap of all other points. We conjectured that the gap is due to the over simplification of the LDM. Subsequently, lattice codes were used to convert the bit-level alignment schemes from the LD-IMAC to the G-IMAC. However, the gap towards the upper bounds consisted due to the fact that the lattice codes schemes were modelled after the bit-alignment strategies of the LD-IMAC, therefore inheriting the suboptimal structure. To overcome this gap, an approximation was needed which lies between the standard G-IMAC and the LD-IMAC. The new approximation model of \cite{Niesen-Ali}, the lower triangular deterministic model (LTDM), was a promising approach. Instead of setting the fine channel gain to one (as in the LDM case), this new model integrated the channel gain by another binary expansion. Taking the fine channel gain into account enables a new class of achievable schemes which are not limited to orthogonal bit alignment. It turned out, that this was also the limiting factor in the previous LDM achievable schemes. We have shown that the IMAC approximated by the LTDM (LTD-IMAC) can achieve the previous upper bounds in the whole interference regime, within a {\it constant} gap. Moreover, techniques from \cite{Niesen-Ali} could be modified in a way to show a constant gap capacity approximation of the G-IMAC, thereby porting the LTDM schemes to the Gaussian model. As a by-product this shows, that the GDoF of the IMAC is indeed as pictured in Fig.~\ref{last_fig}. This shows that in the IMAC considerable gains can be achieved via signal scale alignment methods. Above $\tfrac{3}{2}$ interference-to-signal ratio, the harmful effects of interference can be completely cancelled. Note that we only treated the case of equal interference strength at the receivers. This assumption is only for simplification of the analysis. However, one needs a difference in the ratio of both direct link gains in comparison to the interference link gains in order to enable the signal scale alignment and get multi-user gain. We call this difference the shift-property.
\appendices
\section{Proof of Theorem \ref{mytheoremx} and \ref{General_Bound_G_IMAC}}
\label{Proof_of_General_Bound_G_IMAC}

\begin{proof}
A part of this proof follows closely the method of \cite{Bresler2008} for the Gaussian 2-user IC and applies it to the IMAC setting. The mutual information terms of the Gaussian channel will be transformed into a deterministic form in two steps. Each step will introduce a bit penalty which contributes to the overall gap between both bounds. Step 1 will transform the average power constraint to a peak power constraint. Step 2 truncates the signals at noise level and removes the noise. 

\paragraph{Step 1: average power constraint to peak power constraint.}

The channel equation for the Gaussian IMAC is
\begin{IEEEeqnarray*}{rCl}
y^1 &=& h_{11}^1 2^{n^1_{11}}x_{11}+h_{12}^1 2^{n_{12}^1} x_{12}\\
&& +\: h_{21}^1 2^{n_{21}^1} x_{21}+ h_{22}^1 2^{n_{22}^1} x_{22} +z^1\\
y^2 &=& h_{11}^2  2^{n_{11}^2} x_{11}+h_{12}^2 2^{n_{12}^2} x_{12}\\
&&+\: h_{21}^2 2^{n_{21}^2} x_{21}+h_{22}^2 2^{n_{22}^2}x_{22}+z^2,
\end{IEEEeqnarray*}

Recall that we assume without loss of generality that the Gaussian IMAC inputs have a unit average power constraint. We have to split the input signals in two parts, where one part is not exceeding the unit peak power constraint. We can write the binary expansion of the input signals as

\begin{equation*}
x_{ik}=\sum_{b=-\infty}^\infty [x_{ik}]_b2^{-b}.
\end{equation*}
Now we can write the part which exceeds the peak power constraint as
\begin{equation*}
\hat{x}_{ik}:=sign(x_{ik})\sum_{b=-\infty}^0 [x_{ik}]_b2^{-b}
\end{equation*}
and let the remaining part be
\begin{equation*}
\bar{x}_{ik}:=x_{ik}-\hat{x}_{ik}=sign(x_{ik})\sum_{b=1}^\infty [x_{ik}]_b2^{-b}.
\end{equation*}
Therefore, we can rewrite the channel equations such that all signals have a peak power constraint.
\begin{IEEEeqnarray*}{rCl}
\bar{y}^1 &=& h_{11}^1 2^{n^1_{11}}\bar{x}_{11}+h_{12}^1 2^{n_{12}^1} \bar{x}_{12} + h_{21}^1 2^{n_{21}^1} \bar{x}_{21}\\
&&+\: h_{22}^1 2^{n_{22}^1} \bar{x}_{22} +z^1\\
\bar{y}^2 &=& h_{11}^2  2^{n_{11}^2} \bar{x}_{11}+h_{12}^2 2^{n_{12}^2} \bar{x}_{12}+ h_{21}^2 2^{n_{21}^2} \bar{x}_{21}\\
&&+\:h_{22}^2 2^{n_{22}^2}\bar{x}_{22}+z^2,
\end{IEEEeqnarray*}
and let 
\begin{IEEEeqnarray*}{rCl}
\hat{y}^1 &=& y^1-\bar{y}^1\\
\hat{y}^2 &=& y^2-\bar{y}^2,
\end{IEEEeqnarray*}
be the receiver-side which exceeds the constraints. Each transmitter encodes a codeword $x_{ik}^n(w_{ik})$ based on the message $w_{ik}$ with message rate $R_{ik}$ for cell $i$ and user $k$. We can therefore write the mutual information term for cell $i$ in the following way 

\begin{IEEEeqnarray*}{rCl}
\IEEEeqnarraymulticol{3}{l}{
I(w_{i1},w_{i2};(y^i)^n)}\\
&\overset{(a)}{\leq}& I(w_{i1},w_{i2};(\bar{y}^i)^n,(\hat{y}^i)^n)\\
&=& I(w_{i1},w_{i2};(\bar{y}^i)^n)+I(w_{i1},w_{i2};(\hat{y}^i)^n|(\bar{y}^i)^n)\\
&\leq &I(w_{i1},w_{i2};(\bar{y}^i)^n) + H((\hat{y}^i)^n)\\
&\overset{(b)}{\leq} & I(w_{i1},w_{i2};(\bar{y}^i)^n) + \sum_{i=1}^2 H(\hat{x}_{i1}^n)+ H(\hat{x}_{i2}^n)\\
&\overset{(c)}{\leq} & I(w_{i1},w_{i2};(\bar{y}^i)^n) + 8n\\
\end{IEEEeqnarray*}
where (a) follows with the data processing inequality on $(y^i)^n=(\bar{y}^i)^n+(\hat{y}^i)^n=f((\bar{y}^i)^n,(\hat{y}^i)^n)$, (b) follows since $(\hat{y}^i)^n=\hat{x}_{i1}^n+\hat{x}_{i2}^n$ and (c) is a consequence of Lemma 6 in \cite{Bresler2008}.

\paragraph{Step 2: truncate at noise level and remove noise}
Recall that the input is pure made of peak-power constraint signals due to step 1. Moreover, the channel gain is represented as $h2^n$, where $h\in (1,2]$, $n\in\mathbb{N}$. 
\begin{equation*}
x_{ik}=sign(x_{ik})\sum_{b=1}^\infty [x_{ik}]_b2^{-b}.
\end{equation*}
We will now split the terms $h_{ik}2^{n_{ik}^j}x_{ik}$ in parts above noise level and below noise level. The part above noise level can be written as
\begin{equation*}
h_{ik}2^{n_{ik}^j}sign(x_{ik})\sum_{b=1}^{n_{ik}^j} [x_{ik}]_b2^{-b}.
\end{equation*}
The part below noise level can be bounded from above by
\begin{equation*}
|h_{ik}2^{n_{ik}^j}\sum_{b=n_{ik}^j+1}^{\infty} [x_{ik}]_b2^{-b}|\leq 2^{n_{ik}^j+1}2^{-n_{ik}^j}\leq 2.
\label{bound_by_two}
\end{equation*}

Therefore, we can write the truncated output equations without noise as

\begin{equation*}
\tilde{y}^j=\sum_i\sum_k\left\lfloor2^{n_{ik}^j}\sum_{b=1}^{n_{ik}^j} [x_{ik}]_b2^{-b} \right\rfloor,
\end{equation*}
where $i,k\in\{1,2\}$.
Following the proof method, we define

\begin{IEEEeqnarray*}{rCl}
\epsilon^j &:=&\hat{y}^j-\tilde{y}^j=\sum_i\sum_k \left[ h_{ik}2^{n_{ik}^j}\sum_{b=n_{ik}^j+1}^{\infty} [x_{ik}]_b2^{-b} \right.\\
 &&\:+ \left. \vphantom{h_{ik}2^{n_{ik}^j}\sum_{b=n_{ik}^j+1}^{\infty} [x_{ik}]_b2^{-b}} (h_{ik}^j-1)\sum_{b=1}^{n_{ik}^j} [x_{ik}]_b2^{-b}+ \text{frac}\left( 2^{n_{ik}^j}\sum_{b=1}^{n_{ik}^j} [x_{ik}]_b2^{-b}\right)\right]\\
 &&+\:z^j\\
&=& \sum_i\sum_k x_{ik}^*+z^j,
\end{IEEEeqnarray*}
where the first part in the sum accounts for the terms below noise level, the second part includes the fractional channel gain $h$ and the last part is for the fractional terms due to the floor function. 

Now we can transform the mutual information terms into the final linear deterministic mutual information terms

\begin{IEEEeqnarray*}{rCl}
\IEEEeqnarraymulticol{3}{l}{
I(w_{i1},w_{i2};(\bar{y}^i)^n)}\\
&\leq & I(w_{i1},w_{i2};(\tilde{y}^j)^n,(\epsilon^j)^n)\\
&=&I(w_{i1},w_{i2};(\tilde{y}^j)^n)+I(w_{i1},w_{i2};(\epsilon^j)^n|(\tilde{y}^j)^n)\\
&=&I(w_{i1},w_{i2};(\tilde{y}^j)^n)+h((\epsilon^j)^n|(\tilde{y}^j)^n)\\
&&-\:h((\epsilon^j)^n|(\tilde{y}^j)^n,w_{i1},w_{i2})\\
&\leq & I(w_{i1},w_{i2};(\tilde{y}^j)^n)+h((\epsilon^j)^n)-h(z^j)\\
&=& I(w_{i1},w_{i2};(\tilde{y}^j)^n)+I(x_{11}^*,x_{12}^*,x_{21}^*,x_{22}^*;(\epsilon^j)^n)\\
&\overset{(a)}{<}& I(w_{i1},w_{i2};(\tilde{y}^j)^n)+3.1n,
\end{IEEEeqnarray*}
where (a) is due to the fact that $(x_{11}^*,x_{12}^*,x_{21}^*,x_{22}^*) \mapsto \epsilon^j$ forms a 4-user MAC channel. 
With (\ref{bound_by_two}) and $|(h_{ik}^j-1)x_{ik}|\leq 1$, one can see that $|x_{ik}^*|\leq 4$ and therefore $\tfrac{1}{n}I(x_{11}^*,x_{12}^*,x_{21}^*,x_{22}^*;(\epsilon^j)^n) \leq \tfrac{1}{2}\log (1+4(16))+\epsilon_n<3.1n$ for $n \rightarrow \infty$.

Therefore, we have shown that we can bound the Gaussian IMAC mutual information terms from above with deterministic mutual information terms within a constant gap of 11.1 bits.
\begin{IEEEeqnarray*}{rCl}
\IEEEeqnarraymulticol{3}{l}{
I(w_{i1},w_{i2};(y^i)^n)}\\
&\overset{\text{Step }1}{\leq} &I(w_{i1},w_{i2};(\bar{y}^i)^n) + 8n\\
&\overset{\text{Step }2}{\leq}& I(w_{i1},w_{i2};(\tilde{y}^j)^n)+11.1n.
\end{IEEEeqnarray*}
Note that the LD-IMAC and the LTD-IMAC models are restricted to positive inputs. However, capacity is only a function of the magnitude of the channel gains, and additional negative inputs would only result in an additional constant number of bits. Moreover, addition over $\mathbb{F}_2$ also costs only a constant number of bits (the missing carry over) and can be neglected for high SNR regimes as well. We can therefore upper bound the mutual information of the Gaussian model, with the terms of the deterministic model plus a constant number of bits. Furthermore, we can upper bound the deterministic model terms with Theorem~\ref{UB_LTDM}, which yields an upper bound for the Gaussian model.

\begin{remark}
An alternative way to prove those bounds is to  directly bound the Gaussian mutual information terms by using smart genies and the fact that conditional differential entropy is maximised by Gaussian random variables. In particular, the genie must be chosen such that it mimics the cuts in the proof of Theorem~\ref{UB_LTDM}. The main challenges here are the first two bounds, $D_1$ and $D_2$. Both bounds rely on a cut of the received {\it sum} of signals $y$. We remark that a simple genie which provides for example the interference of the users in cell 1 at receiver 2 plus noise (i.e. $s_1=h_{11}^2  2^{n_{11}^2} x_{11}+h_{12}^2 2^{n_{12}^2} x_{12}+z^2$), to receiver 1 is not sufficient for a good bound. Instead, one needs to provide a genie, where the sum of both signals has the same relative shift as the received signals. Lets say we provide the following genie information to receiver 1: $s_1=a x_{11}+b x_{12}+z^2$, then $a$ and $b$ need to obey $h_{11}^1 2^{n^1_{11}}b=h_{12}^1 2^{n_{12}^1}a$. This property allows an elimination of an additional term inside the log variance term resulting from the Gaussian conditional entropy minus noise. Now one can scale $a$ and $b$ appropriately as in the proof of Theorem~\ref{UB_LTDM}. Since the genie information is now a shifted version of the interference terms, one needs to use the same trick as in the deterministic case to show the bound.
\end{remark}

\end{proof}
\section{Bound on Alignment Structure Rate Term}
\label{BoundOnInterferenceRateTerm}
We look into the lattice decoding bounds and develop lower bounds on the maximum achievable rates $R_{I_i}$ for the possible rate expressions inside the alignment structure for cell $i$ and bit-level $l$. The achievable rate for the alignment structure is divided into three parts. We have two common signal parts, which are also received at cell $j$ and one private signal part which is just received in cell $i$. The common signal parts need to obey two different decoding bounds, where one bound is for the decodability in cell $i$ and one is for the decodability in cell $j$.

\subsection{Common Signal Parts}
For the direct path we have the decoding bound \eqref{Decodingbound} with an equivalent noise
\begin{IEEEeqnarray*}{rCl}
N_i(l)&=&1+\sum_{\mbox{used levels}}\theta_{l+1}\\
&=& \mbox{SNR}_{i1}^{1-l(1-\beta_i)}+\sum_{\substack{m=1\\m  \text{ odd}}}^{\lfloor L_j \rfloor } \tfrac{\mbox{INR}_j^i}{\mbox{SNR}_{j1}^{(m-1)(1-\beta_j)}}-\tfrac{\mbox{INR}_j^i}{\mbox{SNR}_{j1}^{m(1-\beta_j)}},
\label{noise_term}
\end{IEEEeqnarray*}
because every partition, starting at level l, is used atleast once, the middle terms of the sum vanish and we have $\mbox{SNR}_{i1}^{1-l(1-\beta_i)}-1$ left. Moreover, due to the overlapping of the interference, we have a sum over the odd bit-levels of the interference affected part. Together with the general noise of 1, we get the expression above. The decoding bound is therefore given as
\begin{IEEEeqnarray*}{rCl}
R_{\text{IC}_i}(l)&\leq&\log \left(1+ \frac{\mbox{SNR}_{i1}^{1-(l-1)(1-\beta_i)}-\mbox{SNR}_{i1}^{1-l(1-\beta_i)}}{\mbox{SNR}_{i1}^{1-l(1-\beta_i)}+\nu}\right)\\
&=&\bar{R}_{\text{IC}_i^{'}}(l),
\end{IEEEeqnarray*}
where we denote $\nu=\sum_{\substack{m=1\\m  \text{ odd}}}^{\lfloor L_j \rfloor } \tfrac{\mbox{INR}_j^i}{\mbox{SNR}_{j1}^{(m-1)(1-\beta_j)}}-\tfrac{\mbox{INR}_j^i}{\mbox{SNR}_{j1}^{m(1-\beta_j)}}$.
Now we can lower bound $\bar{R}_{\text{I}_i}(l)$ and show that
\begin{IEEEeqnarray*}{rCl}
\bar{R}_{\text{IC}_i^{'}}(l)&=& \log \left(1+ \frac{\mbox{SNR}_{i1}^{1-(l-1)(1-\beta_i)}-\mbox{SNR}_{i1}^{1-l(1-\beta_i)}}{\mbox{SNR}_{i1}^{1-l(1-\beta_i)}+\nu}\right)\\
&=& \log \left(\frac{\mbox{SNR}_{i1}^{1-(l-1)(1-\beta_i)}+\nu}{\mbox{SNR}_{i1}^{1-l(1-\beta_i)}+\nu}\right)\\
&>& \log \left(\frac{\mbox{SNR}_{i1}^{1-(l-1)(1-\beta_i)}}{\mbox{SNR}_{i1}^{1-l(1-\beta_i)}+\nu}\right)\\
&>& \log \left(\frac{\mbox{SNR}_{i1}^{1-(l-1)(1-\beta_i)}}{2\mbox{SNR}_{i1}^{1-l(1-\beta_i)}}\right)\\
&=& \log \mbox{SNR}_{i1}^{1-\beta_i}-1,
\end{IEEEeqnarray*}
where we used that $\mbox{SNR}_{i1}^{1-l(1-\beta_i)}>\nu$.
The decoding bound is different in the interference path, because we need to decode the sum of two bit-levels. A sum of $K$ signals needs to obey the decoding bound \eqref{Decodingbound2} with an equivalent noise of 

\begin{IEEEeqnarray*}{rCl}
N_j(l)&=&1+\sum_{\mbox{used levels}}\theta_{l+1}\\
&=& \tfrac{\mbox{INR}_i^j}{\mbox{SNR}_{i1}^{l(1-\beta_i)}}+\sum_{\substack{m=l+1\\m  \text{ odd}}}^{\lfloor L_i \rfloor } \tfrac{\mbox{INR}_i^j}{\mbox{SNR}_{i1}^{(m-1)(1-\beta_i)}}-\tfrac{\mbox{INR}_i^j}{\mbox{SNR}_{i1}^{m(1-\beta_i)}}\\
&=& \tfrac{\mbox{INR}_i^j}{\mbox{SNR}_{i1}^{l(1-\beta_i)}}(1+\sum_{\substack{m=1\\m  \text{ odd}}}^{\lfloor L_i \rfloor } \mbox{SNR}_{i1}^{(1-m)(1-\beta_i)}-\mbox{SNR}_{i1}^{-m(1-\beta_i)})\\
&=& \tfrac{\mbox{INR}_i^j}{\mbox{SNR}_{i1}^{l(1-\beta_i)}}(1+\nu_2).
\label{noise_term}
\end{IEEEeqnarray*}
Note that an signal directed to level $l$ with an received power of $\theta_l$ at cell $i$ gets received at cell $j$ with power
\begin{IEEEeqnarray*}{rCl}
\theta_l\tfrac{|h_{ik}^j|^2}{|h_{ik}^i|^2}&=& \tfrac{\mbox{INR}_i^j}{\mbox{SNR}_{i1}^{(l-1)(1-\beta_i)}}-\tfrac{\mbox{INR}_i^j}{\mbox{SNR}_{i1}^{l(1-\beta_i)}}\\
&&=\tfrac{\mbox{INR}_i^j}{\mbox{SNR}_{i1}^{l(1-\beta_i)}}(\mbox{SNR}_{i1}^{(1-\beta_i)}-1).
\end{IEEEeqnarray*}
We therefore have the second decoding bound with

\begin{IEEEeqnarray*}{rCl}
R_{\text{CI}_i}(l)&\leq&\log \left(\frac{1}{2}+ \frac{\mbox{SNR}_{i1}^{(1-\beta_i)}-1}{1+\nu_2}\right)\\
&=&\bar{R}_{\text{IC}_i^{''}}(l).
\end{IEEEeqnarray*}

We can now lower bound the rate $\bar{R}_{\text{IC}_i''}(l)$ and show that

\begin{IEEEeqnarray*}{rCl}
\bar{R}_{\text{IC}_i^{''}}(l) &=& \log \left(\frac{1}{2}+ \frac{\mbox{SNR}_{i1}^{(1-\beta_i)}-1}{1+\nu_2}\right)\\
&>& \log \left(1+ \frac{\mbox{SNR}_{i1}^{(1-\beta_i)}-1}{1+\nu_2}\right)-1\\
&=& \log \left( \frac{\mbox{SNR}_{i1}^{(1-\beta_i)}+\nu_2}{1+\nu_2}\right)-1\\
&>& \log \left( \frac{\mbox{SNR}_{i1}^{(1-\beta_i)}}{1+\nu_2}\right)-1\\
&>& \log \left(\mbox{SNR}_{i1}^{(1-\beta_i)}\right)-\log(2)-1\\
&=& \log  \mbox{SNR}_{i1}^{(1-\beta_i)}-2
\end{IEEEeqnarray*}
where we used the fact that $\nu_2<1$ since $\mbox{SNR}_{i1}>1$ and the weak interference regime.

Note that we need to obey both decoding bounds such that our common signal parts are decodable at the legitimate receiver, and also as part of the lattice sum at the unintended receiver. The latter is important for the successive decoding scheme in which we need to subtract the interference sum to be able to decode the following bit-levels. We therefore use the minimum of both and the achievable rate is therefore 
\begin{equation}
\bar{R}_{\text{IC}_i}(l)=\min\{\bar{R}_{\text{IC}_i^{'}}(l),\bar{R}_{\text{IC}_i^{''}}(l)\}>\log  \mbox{SNR}_{i1}^{(1-\beta_i)}-2.
\end{equation}

\subsection{Private Signal Parts}
Here we just need to look into the direct path where we have the decoding bound \eqref{Decodingbound} with an equivalent noise
\begin{IEEEeqnarray*}{rCl}
N_i(l)&=&1+\sum_{\mbox{used levels}}\theta_{l+1}\\
&=& \tfrac{\mbox{INR}_j^i}{\mbox{SNR}_{j1}^{l(1-\beta_j)}}(1+\nu_3),
\label{noise_term}
\end{IEEEeqnarray*}
where 
\begin{equation}
\nu_3=\sum_{\substack{m=1\\m  \text{ odd}}}^{\lfloor L_j \rfloor } \mbox{SNR}_{j1}^{(1-m)(1-\beta_j)}-\mbox{SNR}_{j1}^{-m(1-\beta_j)}.
\end{equation}
Note that this is the equivalent of $N_j(l)$ in the bound for $R_{IC_i^{''}}$. Moreover, the bound is similar, except that we now just need to decode one codeword and therefore get the bound
\begin{IEEEeqnarray*}{rCl}
R_{\text{IP}_i}(l)&\leq&\log \left(1+ \frac{\mbox{SNR}_{j1}^{(1-\beta_j)}-1}{1+\nu_3}\right)\\
&=&\bar{R}_{\text{IP}_i}(l).
\end{IEEEeqnarray*}
Now we can lower bound $\bar{R}_{\text{I}_i}(l)$ by
\begin{IEEEeqnarray*}{rCl}
\bar{R}_{\text{IP}_i}(l)
&>& \log  \mbox{SNR}_{j1}^{(1-\beta_j)}-1,
\end{IEEEeqnarray*}
using the same steps as for $\bar{R}_{\text{IC}_i}$.
This rate depends on the ratio of the direct signals of cell $j$, since this ratio gives the size of the interference alignment blocks and therefore the size of the interference-free slots which cell $i$ can use for communication.

\remark{We did not discuss the case where $\lfloor Li \rfloor \neq Li$, which results in remainder terms for certain cases in the alignment structures. The transmission power is not affected, but the noise term would be slightly different. However, the noise term would only change in the $\nu$-terms but still obey our conditions. This means that all results above apply to these cases as well.}

\section{Proof of Lemma 4}
\label{Proof of Lemma groshev-det-weak}
\begin{figure}
\centering
\includegraphics[scale=1]{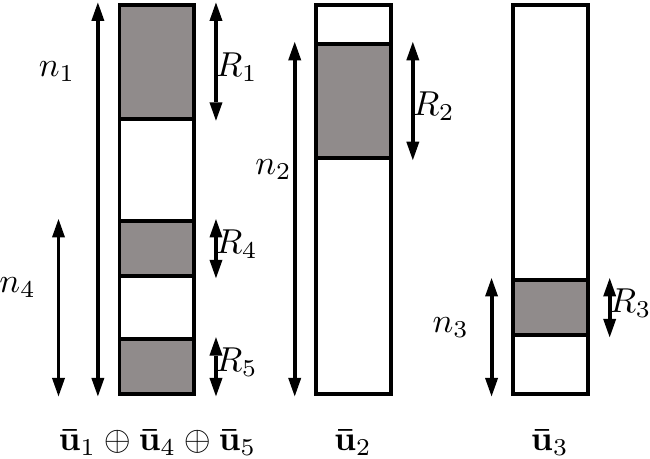}
\caption{Illustration of the received bit-vectors at one receiver. A blank region corresponds to bits which are set to zero, i.e. unused bits. And grey regions corresponds to used bit-levels. Observe that there are 5 used regions, which corresponds to the weak interference case and our proposed scheme.}
\label{analysis-groshev-det-weak-pic}
\end{figure}
We have three bit vectors $\mathbf{\bar{u}}_1^{'}=\mathbf{\bar{u}}_1\oplus \mathbf{\bar{u}}_4 \oplus\mathbf{\bar{u}}_5$, $\mathbf{\bar{u}}_2$ and $\mathbf{\bar{u}}_3$, which are multiplied by the matrices $\mathbf{\bar{G}}_1$, $\mathbf{\bar{G}}_2$ and $\mathbf{\bar{G}}_3$, respectively. Here, $\mathbf{\bar{u}}_1^{'}$ corresponds to the wanted received signal $x_{i1}$, $\mathbf{\bar{u}}_2$ to $x_{i2}$ and $\mathbf{\bar{u}}_3$ to the sum of both interference bit vectors. One can view the signals $\mathbf{\bar{u}}_4$ and $\mathbf{\bar{u}}_5$ as the private parts of $\mathbf{\bar{u}}_1^{'}$. We use the same framework as in \cite{Niesen-Ali} to analyse this setting, i.e we use the notation
\begin{IEEEeqnarray*}{rCl}
\mathcal{U}(n^-,n^+)&:=&\lbrace \mathbf{\bar{u}}\in \{0,1\}^{n_1}: \bar{u}_i=0\\
&\forall& i\in \{1,\ldots,n_1-n^-\} \cup \{n_1-n^++1,\ldots , n_1\} \rbrace
\end{IEEEeqnarray*} 
where $n^-$ and $n^+$ are nonnegative integers such that $n^-\geq n^+$. Such that we can consider bit vectors ($\mathbf{\bar{u}}_1$, $\mathbf{\bar{u}}_2$, $\mathbf{\bar{u}}_4$, $\mathbf{\bar{u}}_3$, $\mathbf{\bar{u}}_5$) in
\begin{IEEEeqnarray*}{rCl}
\mathcal{U}&:=&\mathcal{U}(n_1,n_1-R_1)\times \mathcal{U}(n_2,n_2-R_2)\times\\
&&\times \mathcal{U}(n_4,n_4-R_4)\times \mathcal{U}(n_3,n_3-R_3)\times\mathcal{U}(R_5,0),
\end{IEEEeqnarray*} with $n_1\geq n_2 \geq n_4 \geq n_3$, see Fig~\ref{analysis-groshev-det-weak-pic}.
Note that due to our scheme, we have a few assumptions on $R_{k^{'}}$, $k'\in \{1,\ldots,5\}$, namely that $\mathbf{\bar{u}}_4$, $\mathbf{\bar{u}}_3$ and $\mathbf{\bar{u}}_5$ do not overlap, and $\mathbf{\bar{u}}_1$ and $\mathbf{\bar{u}}_2$ do not overlap with those three signals, too. This yields the following equations
\begin{IEEEeqnarray*}{rCl}
\label{det_scheme_rate_conditions}
R_1+R_4+R_3+R_5 &\leq &  n_1\\
R_2+R_4+R_3+R_5 &\leq &  n_2\\
R_4+R_3+R_5&\leq & n_4\\
R_3+R_5&\leq &  n_3.
\end{IEEEeqnarray*}
We introduce the rate $R_\text{Fix}:=R_3+R_4+R_5$ to denote the sum of the bit-level rates which are fixed by the scheme, i.e. do not overlap. Therefore, we have that 
\begin{IEEEeqnarray*}{rCl}
R_1 &\leq &  n_1-R_\text{Fix}\\
R_2+n_1-n_2 &\leq &  n_1-R_\text{Fix}
\end{IEEEeqnarray*} and it follows that
$n(\mathbf{\bar{u}}_k)\leq n_1-R_\text{Fix}$, for $k\in \{1,2\}$, $\mathbf{\bar{u}}_k \neq 0$ and $n(\mathbf{\bar{u}})$ which gives the smallest index $i$, such that $\bar{u}_i=1$, and $n(\mathbf{0})=\infty$. This means we have that
\begin{equation}
\min\{ n(\mathbf{\bar{u}}_4), n(\mathbf{\bar{u}}_5),n(\mathbf{\bar{u}}_3) \}> n_1-R_\text{Fix}.
\label{privatebits}
\end{equation}
We can now remove the dependence of the outage set on the signals $\mathbf{\bar{u}}_4$, $\mathbf{\bar{u}}_5$, $\mathbf{\bar{u}}_3$. Note that we have $n(\mathbf{\bar{G}}\mathbf{\bar{u}})=n(\mathbf{\bar{u}})$, since the matrices $\mathbf{\bar{G}}$ are unit lower triangular and therefore down-shift (to bigger indices) the signal $\mathbf{\bar{u}}$\footnote{We have $=$ instead of $\leq$, due to the unit diagonal.}. We therefore have that
\begin{equation}
\mathbf{\bar{G}}_1(\mathbf{\bar{u}}_1 \oplus \mathbf{\bar{u}}_4 \oplus \mathbf{\bar{u}}_5) \oplus \mathbf{\bar{G}}_2 \mathbf{\bar{u}}_2 \oplus\mathbf{\bar{G}}_3 \mathbf{\bar{u}}_3= \mathbf{0}
\end{equation} only if
\begin{IEEEeqnarray*}{rCl}
n\left(\mathbf{\bar{G}}_1\mathbf{\bar{u}}_1 \oplus \mathbf{\bar{G}}_2 \mathbf{\bar{u}}_2 \right) &=& n \left( \mathbf{\bar{G}}_1(\mathbf{\bar{u}}_4 \oplus \mathbf{\bar{u}}_5) + \mathbf{\bar{G}}_3 \mathbf{\bar{u}}_3 \right)\\
&\geq & \min\{ n ( \mathbf{\bar{G}}_1(\mathbf{\bar{u}}_4 \oplus \mathbf{\bar{u}}_5)), n(\mathbf{\bar{G}}_3 \mathbf{\bar{u}}_3)\}\\
&=& \min \{n (\mathbf{\bar{u}}_4 \oplus \mathbf{\bar{u}}_5), n(\mathbf{\bar{u}}_3) \}\\
&=&\min \{ \min\{ n (\mathbf{\bar{u}}_4), n(\mathbf{\bar{u}}_5)\}, n(\mathbf{\bar{u}}_3) \}\\
&=&\min \{ n (\mathbf{\bar{u}}_4), n(\mathbf{\bar{u}}_5), n(\mathbf{\bar{u}}_3) \}\\
&>& n_1-R_\text{Fix}.
\end{IEEEeqnarray*}
Also, we have for $(\mathbf{\bar{u}}_4,\mathbf{\bar{u}}_3,\mathbf{\bar{u}}_5)\neq (\mathbf{0},\mathbf{0},\mathbf{0})$ that 
\begin{equation*}
\mathbf{\bar{G}}_1(\mathbf{\bar{u}}_1 \oplus \mathbf{\bar{u}}_4 \oplus \mathbf{\bar{u}}_5) \oplus \mathbf{\bar{G}}_2 \mathbf{\bar{u}}_2 \oplus\mathbf{\bar{G}}_3 \mathbf{\bar{u}}_3= \mathbf{0}
\end{equation*} can hold only if $(\mathbf{\bar{u}}_1,\mathbf{\bar{u}}_2)\neq (\mathbf{0},\mathbf{0})$, again due to the non-overlapping property of the signals $\mathbf{\bar{u}}_4,\mathbf{\bar{u}}_3$ and $\mathbf{\bar{u}}_5$.
Define the sets 
\begin{IEEEeqnarray*}{rCl}
&B'&(\mathbf{\bar{u}}_0,\mathbf{\bar{u}}_1,\mathbf{\bar{u}}_2 ):= \lbrace (g_0,g_1,g_2)\in (1,2]^3: \\
&&n\left(\mathbf{\bar{G}}_0\mathbf{\bar{u}}_0 \oplus\mathbf{\bar{G}}_1\mathbf{\bar{u}}_1 \oplus \mathbf{\bar{G}}_2 \mathbf{\bar{u}}_2 \right)
> n_1-R_{\text{Fix}} \rbrace
\end{IEEEeqnarray*} and 
\begin{IEEEeqnarray*}{rCl}
\mathcal{U}&:=&\mathcal{U}(n_0,n_0-R_0)\times \mathcal{U}(n_1,n_1-R_1)\times \mathcal{U}(n_2,n_2-R_2).
\end{IEEEeqnarray*} with the outage set 
\begin{equation*}
B\subset B' := \bigcup_{(\mathbf{\bar{u}}_0, \mathbf{\bar{u}}_1,\mathbf{\bar{u}}_2)\in \mathcal{U} \setminus \{(\mathbf{0},\mathbf{0},\mathbf{0})\} } B'(\mathbf{\bar{u}}_0,\mathbf{\bar{u}}_1,\mathbf{\bar{u}}_2 ).
\end{equation*}
Observe that this corresponds to our case for $\mathbf{\bar{u}}_0=\mathbf{0}$ and $R_0=0$.
It was shown in \cite[Proof of Lemma 11]{Niesen-Ali}, that the measure of the set 
\begin{equation*}
B\subset B' := \bigcup_{(\mathbf{\bar{u}}_0, \mathbf{\bar{u}}_1,\mathbf{\bar{u}}_2)\in \mathcal{U} \setminus \{(\mathbf{0},\mathbf{0},\mathbf{0})\} } B'(\mathbf{\bar{u}}_0,\mathbf{\bar{u}}_1,\mathbf{\bar{u}}_2 )
\end{equation*}can be bounded by 
\begin{equation}
\mu(B) \leq \delta
\end{equation}
if
\begin{IEEEeqnarray*}{rCl}
R_1+R_0+R_2+R_{\text{Fix}} &\leq &  n_1-\log(16/\delta)\\
R_0+R_2+R_{\text{Fix}} &\leq &  n_0-\log(16/\delta)\\
R_2+R_{\text{Fix}} &\leq & n_2.
\end{IEEEeqnarray*}
If we plug-in $\mathbf{\bar{u}}_0=\mathbf{0}$ and go through the proof, one can see that this reduces the number of cases from four to one, which results in a lower gap and just two conditions, namely
\begin{IEEEeqnarray*}{rCl}
R_1+R_2+R_{\text{Fix}} &\leq &  n_1-\log(4/\delta)\\
R_2+R_{\text{Fix}} &\leq & n_2.
\end{IEEEeqnarray*}
Together with our conditions for the fixed, non-overlapping parts, we have that 
\begin{IEEEeqnarray*}{rCl}
R_1+R_2+R_4+R_3+R_5 &\leq &  n_1-\log(4/\delta)\\
R_2+R_4+R_3+R_5 &\leq &  n_2\\
R_4+R_3+R_5&\leq & n_4\\
R_3+R_5&\leq &  n_3.
\end{IEEEeqnarray*}
Now we can set 
\begin{IEEEeqnarray*}{rCl}
&R_1:= R^c_{k1}\qquad\qquad\qquad &n_1:=n^k_{k1}\\
&R_2:= R^c_{k2}\qquad\qquad\qquad &n_2:=n^k_{k2}\\
&R_3:= \max\{R^c_{l2},R^c_{l1}\}\quad\  &n_4:=(n^k_{k1}-n^k_l)\\
&R_4:= R^{p2}_{k1}\qquad\qquad\qquad &n_3:=n^l_k\\
&R_5:= R^{p1}_{k1}.\qquad\qquad\qquad
\end{IEEEeqnarray*} and replace $\delta$ with $\delta/2$ to account for the measure of the overall outage set (over both receivers) which concludes the proof of Lemma~\ref{ltd-weak-decoding-lemma}.
\section{Proof of The Decoding Lemma for the Weak Interference Case}
\label{Proof_Gaussian_Weak_Case}
In order to get a similar result as in Lemma \ref{decoding_lemma_g_strong} for the weak interference case, we need to prove a modified version of \cite[Lemma 12]{Niesen-Ali}, which is
\begin{lemma}
\label{decoding_general}
Let $n_3$, $n_1$, $n_2 \in \mathbb{Z}_+$ such that $n_1\geq n_2 \geq n_3$ and $\tfrac{n_3}{n_1}<\tfrac{1}{2}$, and let $R_1,R_2,R_3,R_4,R_5\in \mathbb{Z}_+$. Define the event 
\begin{IEEEeqnarray*}{rCl}
B&(&u_1,u_2,u_3,u_4,u_5):=
 \lbrace (g_1,g_2,g_3)\in (1,4]^3:\\
 &&\quad|g_2u_2+g_1(u_1+u_4+u_5)+g_3 u_3|\leq 2^{5-n_1} \rbrace
\end{IEEEeqnarray*} 
and set 
\begin{equation}
B:=\bigcup_{(u_1,u_2,u_3,u_4,u_5)\in \mathcal{U}\setminus\{(0,0,0,0)\}} B(u_1,u_2,u_3,u_4,u_5).
\end{equation}
For any $\delta \in (0,1]$ satisfying 
\begin{IEEEeqnarray*}{rCl}
R_1+R_2+R_5+R_4+R_3&\leq &  n_1-6-\log(1008/\delta)\\
R_2+R_5+R_4+R_3&\leq &  n_2-6\\
R_5+R_4+R_3&\leq &  (n_1-n_3)\\
R_4+R_3&\leq &  n_3
\end{IEEEeqnarray*} 
we have $\mu(B)\leq \delta$. 
\end{lemma}
Where the set $\mathcal{U}$ is defined as
\begin{IEEEeqnarray*}{rCl}
\mathcal{U}(n^-,n^+)&:=&\lbrace u\in [-1,1]: [u]_i=0\\
&\forall& i\in \{1,\ldots,n_1-n^-\} \cup \{n_1-n^++1,\ldots\} \rbrace
\end{IEEEeqnarray*} 
and 
\begin{IEEEeqnarray*}{rCl}
\mathcal{U}&:=&\mathcal{U}(n_1,n_1-R_1)\times \mathcal{U}(n_2,n_2-R_2)\times\\
&\times&\mathcal{U}((n_1-n_3),(n_1-n_3)-R_5)\times \\
&\times&\mathcal{U}(n_3,n_3-R_3)\times\mathcal{U}(R_4,0).
\end{IEEEeqnarray*} 
This means that $\mathcal{U}$ represents the set of possible inputs, if constraint to the specific bit-structure of the LTDM scheme.
\begin{proof}
The proof follows closely the one from \cite{Niesen-Ali}. We will therefore just point out the main differences and refer the reader to the original proof for a more detailed exposition. First of all, note that due to the specific scheme in the LTD model we have the equations
\begin{IEEEeqnarray*}{rCl}
R_5+R_4+R_3&\leq &  (n_1-n_i)\\
R_4+R_3&\leq &  n_i.
\end{IEEEeqnarray*} 
This is because the parts $R_5$, $R_4$, $R_3$ are constructed such that they do not overlap (except the aligning interference signals), see Fig.~\ref{rateallo}.
We now represent the set $B(u_1,u_2,u_3,u_4,u_5)$ as
\begin{IEEEeqnarray*}{rCl}
&B&(u_1,u_2,u_3,u_4,u_5):= \lbrace (g_1,g_2,g_3)\in (1,4]^3: \\
&&|g_22^{n_1}u_2+g_12^{n_1}(u_1+u_4+u_5)+2^{n_1}g_3 u_3|\leq 2^{5} \rbrace
\end{IEEEeqnarray*} and decomposes the shifted signals as
\begin{IEEEeqnarray}{rCl}
2^{n_1}u_2&=& A_2q_2\IEEEyessubnumber\\
2^{n_1}u_1&=& A_1 q_1\IEEEyessubnumber\\
2^{n_1}u_3&=& A_3 q_3\IEEEyessubnumber \\
2^{n_1}u_4 &=& q_4\IEEEyessubnumber\\
2^{n_1}u_5 &=& A_5q_5\IEEEyessubnumber
\end{IEEEeqnarray}
with a bit-offset at the receiver of $A_{k}:=2^{n_k-R_k}$ for $k\in{1,2,3,5}$ and a used number of bits 
\begin{equation*}
q_k \in \{-Q_k,-Q_k+1,\ldots,Q_k-1,Q_k\}
\end{equation*}
for $k\in\{1,2,3,4,5\}$ and $Q_k:=2^{R_k}$. Observe that our set $B$ has the same structure as the one in \cite{Niesen-Ali}, but with an additional variable
attached to the aligning part $g_1$. To solve it in the same fashion, we would need to proof a generalization of {\cite[Lemma~14]{Niesen-Ali}}. However,
we can circumvent this by exploiting the structure of the weak interference case. We can use that the private allocations of the stronger user are not overlapping with the aligned interference parts.
This enables a bound on those parts, which simplifies the analysis and lets us apply {\cite[Lemma~14]{Niesen-Ali}} as it is. The following Lemma introduces this idea.

\begin{lemma}
\label{fitting_lemma}
For $g_1,g_2\in (1,4]$, $A_{k}$ and $q_k$ as defined above, it holds that
\begin{equation}
|g_1(q_4+A_5q_5)+g_3 A_3 q_3| \leq 2^{R_3+R_4+R_5+2}.
\end{equation}
\end{lemma}
\begin{proof}
Observe that due to the specific structure of the parts $u_3$, $u_4$, $u_5$ in $\mathcal{U}$ which is 
\begin{IEEEeqnarray*}{rCl}
u_5&=&b_12^{-(n_1-(n_1-n_3))-1}+\cdots+b_{R_5}2^{-(n_1-((n_1-n_3)-R_5))}\\
u_3&=&b_12^{-(n_1-n_3)-1}+\cdots+b_{R_3}2^{-(n_1-(n_3-R_3))}\\
u_4&=&b_12^{-(n_1-R_4)-1}+\cdots+b_{R_4}2^{-n_1}
\end{IEEEeqnarray*}
with $b_i\in\{0,1\}$
and the LTD scheme constraints such that $R_5\leq (n_1-2n_3)$, $R_3\leq \tfrac{n_3}{2}$ and $R_4\leq \tfrac{n_3}{2}$ we can rewrite the sum
\begin{equation}
\label{Fitting}
A_3 q_3^+ +q_4^+ +A_5q_5^+=q_6^+
\end{equation}
where $R_6:=R_3+R_4+R_5$ and $q_k^+ \in\{0,\ldots,Q_k-1,Q_k\}$ for $k\in\{3,4,5,6\}$ and $Q_k:=2^{R_k}$.
In particular this means that 
\begin{IEEEeqnarray}{rCl}
\label{priv_decod_cond}
R_5+R_4+R_3&\leq &  (n_1-n_3)\\
R_4+R_3&\leq &  n_3,
\end{IEEEeqnarray}
establishing two of the inequalities.
We can now show that 
\begin{IEEEeqnarray*}{rCl}
|g_1(q_4+A_5q_5)+g_3 A_3 q_3| &\leq& |g_1q_4|+|g_1A_5q_5|+|g_3 A_3 q_3|\\
&&\:\leq |g_1||q_4|+|g_1||A_5q_5|+|g_3| |A_3 q_3|\\
&&\:\leq 2^2(q_4^++A_5q_5^++A_3 q_3^+)\\
&&\:= 2^{R_3+R_4+R_5+2},
\end{IEEEeqnarray*}

where we used \eqref{Fitting}.\end{proof}
We can further rewrite $B$ using the triangle inequality and in addition use Lemma~\ref{fitting_lemma} such that

\begin{IEEEeqnarray*}{rCl}
\IEEEeqnarraymulticol{3}{l}{
B(u_1,u_2,u_3,u_4,u_5)}\\
&&=\{|g_2A_2q_2+g_1(A_1 q_1+q_4+A_5q_5)+g_3 A_3 q_3|\leq 2^{5}\}\\
&&\subseteq \: \{|g_2A_2q_2+g_1A_1 q_1|\leq 2^{5} + |g_1(q_4+A_5q_5)+g_3 A_3 q_3|\}\\
&&\subseteq \: \{|g_2A_2q_2+g_1A_1 q_1|\leq 2^{5} + 2^{R_3+R_4+R_5+2}\}\\
&&\subseteq \: \{|g_2A_2q_2+g_1A_1 q_1|\leq \beta \}\\
&&:= B^{'}(q_2,q_1)
\end{IEEEeqnarray*}

where $\beta:=2^{R_3+R_4+R_5+6}$.
Now we can apply Groshev's theorem. In particular, we need a generalisation of the theorem for an asymmetric and non-asymptotic setting, which was proven in \cite[Lemma~14]{Niesen-Ali}. In particular we need just the special case of two parameters of the Lemma:
\begin{lemma}
\label{groshev}
Let $\beta \in (0,1]$, $A_2 \in \mathbb{N}$, and $Q_1,Q_2 \in \mathbb{N}$. Define the event 
\begin{equation*}
B(q_1,q_2):=\{ (g_1,g_2)\in (1,4]^2 : |g_1q_1+A_2g_2q_2| < \beta\}
\end{equation*}
and set 
\begin{equation*}
B:= \bigcup_{\substack{q_1,q_2 \in \mathbb{Z}\\q_1,q_2 \neq \mathbf{0},\\|q_k| \leq Q_k \forall k}} B(q_1,q_2).
\end{equation*}

Then
\begin{equation}
\mu(B)\leq 1008\beta \left( \min \left\{Q_2, \frac{Q_1}{A_2}\right\} \right).
\end{equation}

\end{lemma}
Now we can continue with the proof of Lemma \ref{decoding_general}. We need to normalize 
\begin{equation*}
|g_2A_2q_2+g_1A_1 q_1|< \beta
\end{equation*} in order to fit Lemma~\ref{groshev}. We first assume that $A_1 \leq A_2$.
Define \begin{IEEEeqnarray*}{rCl}
A'_1&:=&1\\
A'_2&:=&\tfrac{A_2}{A_1}=2^{-n_1+n_2-R_2+R_1}\\
\beta'&:=& \tfrac{\beta}{A_1} = 2^{R_1+R_3+R_4+R_5+6-n_1}
\end{IEEEeqnarray*}
Now, observe that we need $\beta' \in (0,1]$ and we therefore have the inequality 
\begin{equation}
\label{n_1_INEQ}
R_1+R_3+R_4+R_5\leq n_1 -6.
\end{equation}

Since we know that $n_1 \geq n_2$, we also know 
\begin{equation*}
Q_1=2^{R_1}\geq 2^{R_1+n_2-n_1}=A'_2Q_2.
\end{equation*}
We therefore have
\begin{equation*}
\mu(B')\leq 1008 \beta' Q_2.
\end{equation*}
Substituting the definitions yields
\begin{equation}
\mu(B')\leq 1008 \cdot 2^{R_1+R_2+R_2+R_3+R_4+R_5+6-n_1}.
\end{equation}
Combined with \eqref{n_2_INEQ}, this shows that if
\begin{IEEEeqnarray*}{rCl}
R_1+R_2+R_5+R_4+R_3&\leq &  n_1-6-\log(1008/\delta),\\
R_1+R_5+R_4+R_3&\leq &  n_1-6,
\end{IEEEeqnarray*}
we have $\mu(B')\leq \delta.$

The case for $A_2\leq A_1$ is done in the same fashion with appropriate index switching. For $\beta'$ it gives the condition
\begin{equation}
\label{n_2_INEQ}
R_2+R_3+R_4+R_5\leq n_2 -6,
\end{equation} which is stronger than \eqref{n_1_INEQ}. Furthermore, it yields the bound
\begin{equation*}
\mu(B')\leq 1008 \beta' \frac{Q_2}{A_1'},
\end{equation*}which gives again
\begin{equation}
\mu(B')\leq 1008 \cdot 2^{R_1+R_2+R_3+R_4+R_5+6-n_1}.
\end{equation}
We therefore have that the above equations, together with \eqref{priv_decod_cond} yields
\begin{IEEEeqnarray*}{rCl}
R_1+R_2+R_5+R_4+R_3&\leq &  n_1-6-\log(1008/\delta),\\
R_2+R_5+R_4+R_3&\leq &  n_2-6,\\
R_5+R_4+R_3&\leq &  (n_1-n_3),\\
R_4+R_3&\leq &  n_3,
\end{IEEEeqnarray*}
which results in $\mu(B')\leq \delta.$ This completes the proof of Lemma~\ref{decoding_general} and it can be used to show the following result: 

\begin{lemma}
\label{decoding_lemma_g_weak}
Let $\delta\in (0,1]$ and $n_1$, $n_2$, $n_i \in \mathbb{N}$, $n_1 \geq n_2 \geq n_i$ and $\tfrac{n_i}{n_1}< \tfrac{1}{2}$. Assume $R^P_{11}$, $R^{C}_{11}$, $R_{12}$, $R_{21}$, $R^P_{22}$ $R^{C}_{22}$ $\in \mathbb{Z}_+$ satisfy,
\begin{IEEEeqnarray*}{rCl}
R^{C}_{11}+R^C_{12} +R_{21}^C + R^{P_1}_{11}+ R^{P_2}_{11}&\leq &  n_{1}-\log\left(\tfrac{c}{\delta}\right)\\
R^C_{12} +R_{21}^C + R^{P_1}_{11}+ R^{P_2}_{11}&\leq&  n_{2}-6\\
R_{21}^C + R^{P_1}_{11}+ R^{P_2}_{11}&\leq&  (n_1-n_i)\\
R_{21}^C+ R^{P_2}_{11}&\leq&  n_i
\end{IEEEeqnarray*} 

and
\begin{IEEEeqnarray*}{rCl}
R^{C}_{21}+R_{22}^C+R^{C}_{11}+ R^{P_1}_{22}+ R^{P_2}_{22}&\leq&  n_1-\log\left(\tfrac{c}{\delta}\right)\\
R_{22}^C+R^{C}_{11}+ R^{P_1}_{22}+ R^{P_2}_{22} &\leq&  n_2-6\\
R^{C}_{11}+ R^{P_1}_{22}+ R^{P_2}_{22}&\leq&  (n_1-n_i)\\
R^{C}_{11}+R^{P_2}_{22}&\leq&  n_i
\end{IEEEeqnarray*} 
where $c:=2016$ and $R_{ik}$ is the rate of the signal $u_{ik}$. Then, the bit allocation of the LTD-IMAC applied to the Gaussian IMAC results in a minimum constellation distance $d\geq32$ at each receiver for all channel gains ($g_{ik}^j\in(1,2]^{2\times 2}$) except for a set $B\subset (1,2]^{2 \times 2}$ of Lebesgue measure $\mu(B)\leq \delta$.
\end{lemma}
\end{proof}
\section{Verification of Decoding Conditions on LTDM Schemes}

We check the cases
\begin{IEEEeqnarray*}{rCl}
\text{I}&:&\ \alpha \in [0,\tfrac{1}{2}],\ \text{II}:\ \alpha \in (\tfrac{1}{2},\tfrac{3}{5}),\ \text{III}:\ \alpha \in [\tfrac{3}{5},1]\\
\text{IV}&:&\ \alpha \in (1,\tfrac{3}{2}],\ \text{V}:\ \alpha \in (\tfrac{3}{2},\infty)
\end{IEEEeqnarray*} separately. We restrict the following section on cell 1. The corresponding cases for cell 2 can be checked accordingly with adjusted indices.

{\bf Case I.1 ($n^{j\neq i}_{i} \leq \Delta_i$):}
In this case, $R^c_{12}=\min\{\left\lceil\tfrac{1}{2}n_1^2\right\rceil,(n^1_{12}-(n^1_{11}-n^2_1))^+\}=(n^1_{12}-(n^1_{11}-n^2_1))^+=0$. This is in accordance with the discussion above.
For the first restriction we have:
\begin{IEEEeqnarray*}{rCl}
R^c_{11}&+&R^{p1}_{11}+R^{p2}_{11}+R^c_{12}+\max\{R^c_{22},R^c_{21}\}\\
&=& \left\lfloor\tfrac{1}{2}n_1^2\right\rfloor+\left\lceil\tfrac{1}{2}n_2^1\right\rceil+\left\lfloor\tfrac{1}{2}n_2^1\right\rfloor+n_{11}^1-n_1^2-n_2^1\\
&=&n_{11}-\left\lceil\tfrac{1}{2}n_1^2\right\rceil \leq n_{11}
\end{IEEEeqnarray*}
The second one gives:
\begin{IEEEeqnarray*}{rCl}
R^{p1}_{11}&+&R^{p2}_{11}+R^c_{12}+\max\{R^c_{21},R^c_{22}\}\\
&=&\left\lceil\tfrac{1}{2}n_2^1\right\rceil+\left\lfloor\tfrac{1}{2}n_2^1\right\rfloor+n_{11}^1-n_1^2-n_2^1\\
&=&n_{11}^1-n_1^2
\end{IEEEeqnarray*}
which obeys the restriction since $n_{11}^1-n_1^2 \geq n^1_{11}-\Delta_1= n_{12}^1$.

{\bf Case I.2 ( $\tfrac{1}{2}n^{j\neq i}_{i} < \Delta_i < n^{j\neq i}_{i}$):}
Here, $R^c_{12}=\min\{\left\lceil\tfrac{1}{2}n_1^2\right\rceil,(n^1_{12}-(n^1_{11}-n^2_1))^+\}=n^1_{12}-(n^1_{11}-n^2_1)>0$. $R^c$ becomes valuable, and multi-user gain starts increasing.
\begin{IEEEeqnarray*}{rCl}
R^c_{11}&+&R^{p1}_{11}+R^{p2}_{11}+R^c_{12}+\max\{R^c_{22},R^c_{21}\}\\
&=& n^1_{12}-(n^1_{11}-n^2_1)+\left\lfloor\tfrac{1}{2}n_1^2\right\rfloor+\left\lceil\tfrac{1}{2}n_2^1\right\rceil+\left\lfloor\tfrac{1}{2}n_2^1\right\rfloor\\
&&+\:n_{11}^1-n_1^2-n_2^1\\
&=&n_{12}^1+\left\lfloor\tfrac{1}{2}n_1^2\right\rfloor\leq n_{12}^1+\Delta_1 = n_{11}^1.
\end{IEEEeqnarray*}
The next conditions result in:
\begin{IEEEeqnarray*}{rCl}
R^{p1}_{11}&+&R^{p2}_{11}+R^c_{12}+\max\{R^c_{21},R^c_{22}\}\\
&=&n^1_{12}-(n^1_{11}-n^2_1)+\left\lceil\tfrac{1}{2}n_2^1\right\rceil+\left\lfloor\tfrac{1}{2}n_2^1\right\rfloor+n_{11}^1-n_1^2-n_2^1\\
&=&n_{12}^1,
\end{IEEEeqnarray*}
which obeys the restriction since $n_{11}^1-n^2_1 < n_{11}^1-\Delta_1 = n_{12}^1$.

{\bf Case I.3 ($\Delta_i \leq \tfrac{1}{2}n^{j\neq i}_{i}$):}
In this case, $R^c_{12}=\min\{\left\lceil\tfrac{1}{2}n_1^2\right\rceil,(n^1_{12}-(n^1_{11}-n^2_1))^+\}=\left\lceil\tfrac{1}{2}n_1^2\right\rceil$. $R^c$ is fully valuable, and multi-user gain is at maximum.
\begin{IEEEeqnarray*}{rCl}
R^c_{11}&+&R^{p1}_{11}+R^{p2}_{11}+R^c_{12}+\max\{R^c_{22},R^c_{21}\}\\
&=&\left\lfloor\tfrac{1}{2}n_1^2\right\rfloor+\left\lceil\tfrac{1}{2}n_1^2\right\rceil+\left\lceil\tfrac{1}{2}n_2^1\right\rceil+\left\lfloor\tfrac{1}{2}n_2^1\right\rfloor+n_{11}^1-n_1^2-n_2^1\\
&=& n_{11}^1.
\end{IEEEeqnarray*}
The second condition results in
\begin{IEEEeqnarray*}{rCl}
R^{p1}_{11}&+&R^{p2}_{11}+R^c_{12}+\max\{R^c_{21},R^c_{22}\}\\
&=& \left\lceil\tfrac{1}{2}n_1^2\right\rceil+\left\lceil\tfrac{1}{2}n_2^1\right\rceil+\left\lfloor\tfrac{1}{2}n_2^1\right\rfloor+n_{11}^1-n_1^2-n_2^1\\
&=&n_{11}^1-\left\lfloor\tfrac{1}{2}n_1^2\right\rfloor
\end{IEEEeqnarray*}
which obeys the restriction since $n_{11}^1-\left\lfloor\tfrac{1}{2}n_1^2\right\rfloor \leq n_{11}^1-\Delta_1 = n_{12}^1$.
For the last two conditions we have
\begin{IEEEeqnarray*}{rCl}
R^{p1}_{11}&+&R^{p2}_{11}+\max\{R^c_{22},R^c_{21}\}\\
&=&\left\lceil\tfrac{1}{2}n_2^1\right\rceil+\left\lfloor\tfrac{1}{2}n_2^1\right\rfloor+n_{11}^1-n_1^2-n_2^1\\
&=&n^1_{11}-n_1^2.
\end{IEEEeqnarray*}

Furthermore,
\begin{IEEEeqnarray*}{rCl}
R^{p1}_{11}+\max\{R^c_{22},R^c_{21}\}&=&\left\lceil\tfrac{1}{2}n_2^1\right\rceil+\left\lfloor\tfrac{1}{2}n_2^1\right\rfloor=n_2^1.
\end{IEEEeqnarray*}
which is applicable to the cases I.1-3.

{\bf Case II.1 :}
The first sub-case is when $n_2\geq n_i+\left\lfloor\tfrac{1}{2}(n_1-n_i)\right\rfloor$.
Therefore $\left\lfloor\tfrac{1}{2}(n_1-n_i)\right\rfloor\leq (n_2-n_i)$ and $R^c_{k2}:=\left\lfloor\tfrac{1}{2}(n_1-n_i)\right\rfloor$. The second direct signal has enough power to provide the full multi-user gain. 
\begin{IEEEeqnarray*}{rCl}
R^c_{11}&+&R^{p1}_{11}+R^c_{12}+\max\{R^c_{22},R^c_{21}\}\\
&=&3\left\lfloor\tfrac{1}{2}(n_1-n_i)\right\rfloor+ n_i-\left\lfloor\tfrac{1}{2}(n_1-n_i)\right\rfloor\\
&\leq &n_i +2(\tfrac{1}{2}(n_1-n_i))=n_1.
\end{IEEEeqnarray*}
\begin{IEEEeqnarray*}{rCl}
R^{p1}_{11}&+&R^c_{12}+\max\{R^c_{22},R^c_{21}\}\\
&=&2\left\lfloor\tfrac{1}{2}(n_1-n_i)\right\rfloor+ n_i-\left\lfloor\tfrac{1}{2}(n_1-n_i)\right\rfloor\\
&\leq &n_2.
\end{IEEEeqnarray*}
\begin{IEEEeqnarray*}{rCl}
R^{p1}_{11}&+&\max\{R^c_{22},R^c_{21}\}\\
&=&\left\lfloor\tfrac{1}{2}(n_1-n_i)\right\rfloor+ n_i-\left\lfloor\tfrac{1}{2}(n_1-n_i)\right\rfloor\\
&=&n_i.
\end{IEEEeqnarray*}

{\bf Case II.2 :}
The second sub-case is when $n_2<n_i+\left\lfloor\tfrac{1}{2}(n_1-n_i)\right\rfloor$.
Therefore $\left\lfloor\tfrac{1}{2}(n_1-n_i)\right\rfloor>(n_2-n_i)$ and $R^c_{k2}:=(n_2-n_i)$. The weaker user has not enough power to provide the full multi-user gain. 
\begin{IEEEeqnarray*}{rCl}
R^c_{11}&+&R^{p1}_{11}+R^c_{12}+\max\{R^c_{22},R^c_{21}\}\\
&=&2\left\lfloor\tfrac{1}{2}(n_1-n_i)\right\rfloor+ n_i-\left\lfloor\tfrac{1}{2}(n_1-n_i)\right\rfloor+(n_2-n_i)\\
&\leq &\tfrac{1}{2}n_1+n_2-\tfrac{1}{2}n_i\\
&<&\tfrac{1}{2}n_1+n_i+\left\lfloor\tfrac{1}{2}(n_1-n_i)\right\rfloor-\tfrac{1}{2}n_i\leq n_1.
\end{IEEEeqnarray*}
\begin{IEEEeqnarray*}{rCl}
R^{p1}_{11}&+&R^c_{12}+\max\{R^c_{22},R^c_{21}\}\\
&=&n_i-\left\lfloor\tfrac{1}{2}(n_1-n_i)\right\rfloor+(n_2-n_i)+\left\lfloor\tfrac{1}{2}(n_1-n_i)\right\rfloor\\
&=&n_2.
\end{IEEEeqnarray*}
The last condition follows from case II.1. 

{\bf Case III.A.1 (B.1) :}
The cases are active for $n_2>n_1-\tfrac{1}{3}n_i$.
Therefore $\left\lfloor(n_2+\tfrac{2}{3}n_i-n_1)\right\rfloor\geq\left\lfloor\tfrac{1}{3}n_i\right\rfloor$ and $R^c_{k2}:=\left\lfloor\tfrac{1}{3}n_i\right\rfloor$. Full multi-user gain can be achieved. 
\begin{IEEEeqnarray*}{rCl}
R^c_{11}&+&R^{p1}_{11}+R^c_{12}+\max\{R^c_{22},R^c_{21}\}\\
&=&3\left\lfloor\tfrac{1}{3}n_i\right\rfloor+(n_1-n_i)\\
&\leq &n_1.
\end{IEEEeqnarray*}
\begin{IEEEeqnarray*}{rCl}
R^{p1}_{11}&+&R^c_{12}+\max\{R^c_{22},R^c_{21}\}\\
&=&2\left\lfloor\tfrac{1}{3}n_i\right\rfloor+(n_1-n_i)\\
&\leq&n_1-\tfrac{1}{3}n_i<n_2.
\end{IEEEeqnarray*}
\begin{IEEEeqnarray*}{rCl}
R^{p1}_{11}&+&\max\{R^c_{22},R^c_{21}\}\\
&=&\left\lfloor\tfrac{1}{3}n_i\right\rfloor+(n_1-n_i)\\
&\leq &n_1-\tfrac{2}{3}n_i\overset{\left(\tfrac{3}{5}\leq\alpha\right)}{\leq}n_i.
\end{IEEEeqnarray*}
Note that case III.C.1 follows on the same lines for $n_2\geq n_1-\tfrac{2}{3}n_i$ by careful evaluation of the decoding bounds, and re-indexing for $n_2\leq n_i$.

{\bf Case III.B.2 :}
The second sub-case of B is when $n_i<n_2\leq n_1-\tfrac{1}{3}n_i$.
Therefore $R^c_{k1}:=\left\lfloor\tfrac{1}{3}n_i\right\rfloor$ and $R^c_{k2}:=\left\lfloor((n_2-n_i)^+ + \tfrac{5}{3}n_i-n_1)^+\right\rfloor $. 
\begin{IEEEeqnarray*}{rCl}
R^c_{11}&+&R^{p1}_{11}+R^c_{12}+\max\{R^c_{22},R^c_{21}\}\\
&\leq&\tfrac{2}{3}n_i+(n_1-n_i)+n_2-n_i + \tfrac{5}{3}n_i-n_1\\
&\leq &\tfrac{1}{3}n_i+n_2<\tfrac{1}{3}n_i+n_1-\tfrac{1}{3}n_i=n_1.
\end{IEEEeqnarray*}
\begin{IEEEeqnarray*}{rCl}
R^{p1}_{11}&+&R^c_{12}+\max\{R^c_{22},R^c_{21}\}\\
&\leq &\tfrac{1}{3}n_i+(n_1-n_i)+n_2-n_i + \tfrac{5}{3}n_i-n_1\\
&=& n_2.
\end{IEEEeqnarray*}
The last case follows from III.B.1.

{\bf Case III.B.3 :}
The third sub-case is for $n_i\geq n_2>\tfrac{2}{3}n_i$ and $\tfrac{2}{3} < \alpha \leq \tfrac{3}{4}$.
Therefore $R^c_{k1}:=\left\lfloor\tfrac{1}{3}n_i\right\rfloor$ and $R^c_{k2}:=\left\lfloor(\tfrac{5}{3}n_i-n_1)^+\right\rfloor $. We have that
\begin{IEEEeqnarray*}{rCl}
R^c_{11}&+&R^{p1}_{11}+R^c_{12}+\max\{R^c_{22},R^c_{21}\}\\
&\leq &\tfrac{2}{3}n_i+(n_1-n_i)+\tfrac{5}{3}n_i-n_1\\
&= & \tfrac{4}{3}n_i \overset{\alpha \leq \tfrac{3}{4}} {\leq} n_1.
\end{IEEEeqnarray*}
\begin{IEEEeqnarray*}{rCl}
R^{p1}_{11}&+&R^c_{12}+\max\{R^c_{22},R^c_{21}\}\\
&\leq &\tfrac{1}{3}n_i+(n_1-n_i)+\tfrac{5}{3}n_i-n_1\\
&=&n_i.
\end{IEEEeqnarray*}
The last condition follows from B.1 by re-indexing (switching $n_2$ with $n_i$).

{\bf Case III.B.4 :}
The last sub-case of III.B is for $\tfrac{2}{3}n_i \geq n_2>n_1-n_i$ and again $\tfrac{2}{3} < \alpha \leq \tfrac{3}{4}$.
Therefore $R^c_{k1}:=\left\lfloor\tfrac{1}{3}n_i\right\rfloor+ \min \{\left\lfloor (\tfrac{2}{3}n_i-n_2)^+\right\rfloor , \left\lfloor (n_1-\tfrac{4}{3}n_i)^++\tfrac{1}{2}(2n_i-n_2-n_1)^+\right\rfloor \}$ and 
$R^c_{k2}:=(n_2-(n_1-n_i))^+$. For $n_2 \geq 2n_i-n_1$ we have $R^c_{k1}:=\left\lfloor\tfrac{1}{3}n_i\right\rfloor+ \left\lfloor (\tfrac{2}{3}n_i-n_2)^+\right\rfloor $ which results in
\begin{IEEEeqnarray*}{rCl}
R^c_{11}&+&R^{p1}_{11}+R^c_{12}+\max\{R^c_{22},R^c_{21}\}\\
&\leq &\tfrac{2}{3}n_i + (\tfrac{4}{3}n_i-2n_2)^+ +(n_1-n_i)+(n_2-(n_1-n_i))^+\\
&=& 2n_i-n_2\leq n_1.
\end{IEEEeqnarray*}
\begin{IEEEeqnarray*}{rCl}
R^{p1}_{11}&+&R^c_{12}+\max\{R^c_{22},R^c_{21}\}\\
&\leq &\tfrac{1}{3}n_i + (\tfrac{2}{3}n_i-n_2)^+ +(n_1-n_i)+(n_2-(n_1-n_i))^+\\
&=& n_i.
\end{IEEEeqnarray*} For $n_2 < 2n_i-n_1$ we have $R^c_{k1}:=\left\lfloor\tfrac{1}{3}n_i\right\rfloor+ \left\lfloor (n_1-\tfrac{4}{3}n_i)^++\tfrac{1}{2}(2n_i-n_2-n_1)^+\right\rfloor $ which results in 
\begin{IEEEeqnarray*}{rCl}
R^c_{11}&+&R^{p1}_{11}+R^c_{12}+\max\{R^c_{22},R^c_{21}\}\\
&\leq &\tfrac{2}{3}n_i +2(n_1-\tfrac{4}{3}n_i)^++(2n_i-n_2-n_1)^+\\
&& +\:(n_1-n_i)+(n_2-(n_1-n_i))^+\\
&=& n_1.
\end{IEEEeqnarray*}
\begin{IEEEeqnarray*}{rCl}
R^{p1}_{11}&+&R^c_{12}+\max\{R^c_{22},R^c_{21}\}\\
&\leq &\tfrac{1}{3}n_i + (n_1-\tfrac{4}{3}n_i)^++\tfrac{1}{2}(2n_i-n_2-n_1)^+\\
&& +\:(n_1-n_i)+(n_2-(n_1-n_i))^+\\
&=& \tfrac{n_2}{2}+\tfrac{n_1}{2}<n_i.
\end{IEEEeqnarray*}
The last condition follows from III.B.1.

{\bf Case III.A.2 :}
The second sub-case is when $n_i\leq n_2<n_1-\tfrac{1}{3}n_i$ and $\alpha<\tfrac{2}{3}$.
Therefore $\left\lfloor(n_2+\tfrac{2}{3}n_i-n_1)\right\rfloor<\left\lfloor\tfrac{1}{3}n_i\right\rfloor$ and $R^c_{k2}:=\left\lfloor(n_2+\tfrac{2}{3}n_i-n_1)\right\rfloor$. The weaker user has not enough power to provide the full multi-user gain.
\begin{IEEEeqnarray*}{rCl}
R^c_{11}&+&R^{p1}_{11}+R^c_{12}+\max\{R^c_{22},R^c_{21}\}\\
&=&2\left\lfloor\tfrac{1}{3}n_i\right\rfloor+(n_1-n_i)+\left\lfloor(n_2+\tfrac{2}{3}n_i-n_1)\right\rfloor\\
&\leq&\tfrac{2}{3}n_i+(n_1-n_i)+(n_2+\tfrac{2}{3}n_i-n_1)\\
&=&\tfrac{1}{3}n_i+n_2<n_1.
\end{IEEEeqnarray*}
\begin{IEEEeqnarray*}{rCl}
R^{p1}_{11}&+&R^c_{12}+\max\{R^c_{22},R^c_{21}\}\\
&=&\left\lfloor\tfrac{1}{3}n_i\right\rfloor+(n_1-n_i)+\left\lfloor(n_2+\tfrac{2}{3}n_i-n_1)\right\rfloor\\
&\leq& n_2.
\end{IEEEeqnarray*}
The last condition follows from case III.A.1. 

{\bf Case III.C.2 :}
This sub-case is when $n_1-n_i\leq n_2<n_1-\tfrac{2}{3}n_i$ and $\alpha>\tfrac{3}{4}$.
Therefore $(n_2-(n_1-n_i))^+<\left\lfloor\tfrac{1}{3}n_i\right\rfloor$ and $R^c_{k2}:=n_2-(n_1-n_i)$. The weaker user has not enough power to provide the full multi-user gain.
\begin{IEEEeqnarray*}{rCl}
R^c_{11}&+&R^{p1}_{11}+R^c_{12}+\max\{R^c_{22},R^c_{21}\}\\
&=&2\left\lfloor\tfrac{1}{3}n_i\right\rfloor+(n_1-n_i)+n_2-(n_1-n_i)+[n_1-n_2-\tfrac{2}{3}n_i]\\
&\leq& n_1.
\end{IEEEeqnarray*}
\begin{IEEEeqnarray*}{rCl}
&R&^{p1}_{11}+R^c_{12}+\max\{R^c_{22},R^c_{21}\}\\
&=&\left\lfloor\tfrac{1}{3}n_i\right\rfloor+(n_1-n_i)+n_2-(n_1-n_i)+\tfrac{1}{2}[n_1-n_2-\tfrac{2}{3}n_i]\\
&\leq& \tfrac{n_1}{2}+\tfrac{n_2}{2} < \tfrac{1}{2}(n_1-\tfrac{2}{3}n_i)+\tfrac{n_1}{2}\leq n_i.
\end{IEEEeqnarray*}
The last condition follows from case III.C.1.

{\bf Case IV.1 :}
The first sub-case is when $n_2>\tfrac{1}{3}n_i$.
Therefore $R^c_{k2}:=\left\lfloor\tfrac{1}{3}n_i\right\rfloor$. Since we are in the range $\alpha\geq 1$, the private part vanishes. Furthermore, full multi-user gain can be achieved. 
\begin{IEEEeqnarray*}{rCl}
R^c_{11}&+&R^c_{12}+\max\{R^c_{22},R^c_{21}\}\\
&=&3\left\lfloor\tfrac{1}{3}n_i\right\rfloor\\
&\leq&n_i.
\end{IEEEeqnarray*}
\begin{IEEEeqnarray*}{rCl}
R^c_{12}&+&\max\{R^c_{22},R^c_{21}\}\\
&=&2\left\lfloor\tfrac{1}{3}n_i\right\rfloor\\
&\leq &\tfrac{2}{3}n_i\overset{\left(\alpha \leq \tfrac{3}{2}\right)}{\leq}n_1.
\end{IEEEeqnarray*}

{\bf Case IV.2 :}
The second sub-case is when $n_2\leq\tfrac{1}{3}n_i$.
Therefore $R^c_{k2}:=n_2$, and the second term in $R^c_{k1}:=\left\lfloor\tfrac{1}{3}n_i\right\rfloor+\left\lfloor\tfrac{1}{2}(\tfrac{1}{3}n_i-n_2)^+\right\rfloor$ gets activated. The weaker user has not enough power to provide the full multi-user gain.
\begin{IEEEeqnarray*}{rCl}
R^c_{11}&+&R^c_{12}+\max\{R^c_{22},R^c_{21}\}\\
&=&2(\left\lfloor\tfrac{1}{3}n_i\right\rfloor+\left\lfloor\tfrac{1}{2}(\tfrac{1}{3}n_i-n_2)^+\right\rfloor) + n_2\\
&\leq&n_i.
\end{IEEEeqnarray*}
\begin{IEEEeqnarray*}{rCl}
R^c_{12}&+&\max\{R^c_{22},R^c_{21}\}\\
&=&\left\lfloor\tfrac{1}{3}n_i\right\rfloor+\left\lfloor\tfrac{1}{2}(\tfrac{1}{3}n_i-n_2)^+\right\rfloor + n_2\\
&\leq&\tfrac{1}{2}n_i-\tfrac{1}{2}n_2 \overset{\left(\alpha \leq \tfrac{3}{2}\right)}{\leq}\tfrac{3}{4}n_1-\tfrac{1}{2}n_2<n_1.
\end{IEEEeqnarray*}

{\bf Case V.1.1 :}
The first sub-case is when $n_2\geq\tfrac{1}{2}n_1$.
Therefore $R^c_{k2}:=\left\lfloor\tfrac{1}{2}n_1\right\rfloor$. Full multi-user gain can be achieved. 
\begin{IEEEeqnarray*}{rCl}
R^c_{11}&+&R^c_{12}+\max\{R^c_{22},R^c_{21}\}\\
&=&3\left\lfloor\tfrac{1}{2}n_1\right\rfloor\\
&\overset{\alpha\geq \tfrac{3}{2}}{\leq}&n_i.
\end{IEEEeqnarray*}
\begin{IEEEeqnarray*}{rCl}
R^c_{12}&+&\max\{R^c_{22},R^c_{21}\}\\
&=&2\left\lfloor\tfrac{1}{2}n_1\right\rfloor\\
&\leq &n_1.
\end{IEEEeqnarray*}

{\bf Case V.1.2 :}
The second sub-case is when $n_2<\tfrac{1}{2}n_1$ but $2n_1<n_i+n_2$.
Therefore $R^c_{k2}:=n_2$ and $R^c_{k1}:=n_1-n_2$. Full multi-user gain can be achieved. 
\begin{IEEEeqnarray*}{rCl}
R^c_{11}&+&R^c_{12}+\max\{R^c_{22},R^c_{21}\}\\
&=&n_2+2(n_1-n_2)=2n_1-n_2<n_i.
\end{IEEEeqnarray*}
\begin{IEEEeqnarray*}{rCl}
R^c_{12}&+&\max\{R^c_{22},R^c_{21}\}\\
&=&n_1.
\end{IEEEeqnarray*}

{\bf Case V.2 :}
The third sub-case is when $n_2<\tfrac{1}{2}n_1$ and $2n_1 \geq n_i+n_2$.
Therefore $R^c_{k2}:=n_2$ and $R^c_{k1}:=\left\lfloor\tfrac{1}{2}n_i-\tfrac{1}{2}n_2\right\rfloor$. The weaker user has not enough power to provide the full multi-user gain.
\begin{IEEEeqnarray*}{rCl}
R^c_{11}&+&R^c_{12}+\max\{R^c_{22},R^c_{21}\}\\
&=&n_2+2(\left\lfloor\tfrac{1}{2}n_i-\tfrac{1}{2}n_2\right\rfloor)\\
&\leq& n_i.
\end{IEEEeqnarray*}
\begin{IEEEeqnarray*}{rCl}
R^c_{12}&+&\max\{R^c_{22},R^c_{21}\}\\
&=&n_2+\left\lfloor\tfrac{1}{2}n_i-\tfrac{1}{2}n_2\right\rfloor\\
&\leq & \tfrac{1}{2}n_i+\tfrac{1}{2}n_2\\
&\leq& \tfrac{1}{2}n_i+ \tfrac{1}{2}(2n_1-n_i)=n_1.
\end{IEEEeqnarray*}

\bibliographystyle{IEEEtran}
\bibliography{ref}

\end{document}